\documentclass[reqno,11pt]{amsart}
\usepackage{amsmath, latexsym, amsfonts, amssymb, amsthm, amscd}
\usepackage{enumitem}
\usepackage[foot]{amsaddr}

\usepackage{verbatim}

\usepackage{xcolor}
\usepackage{enumitem}
\usepackage{graphics,epsf,psfrag}
\usepackage{graphicx}
\usepackage{subcaption}
\usepackage{hyperref}

\numberwithin{equation}{section}

\newtheorem{theorem}{Theorem}[section]
\newtheorem{lemma}[theorem]{Lemma}
\newtheorem{proposition}[theorem]{Proposition}
\newtheorem{cor}[theorem]{Corollary}
\newtheorem{rem}[theorem]{Remark}

\DeclareMathOperator{\Log}{\mathrm{Log}}
\DeclareMathOperator{\dist}{\mathrm{dist}}

\newcommand{\ind}{\mathbf{1}}

\renewcommand{\tilde}{\widetilde}

\newcommand{\cP}{{\ensuremath{\mathcal P}} }

\newcommand{\cC}{{\ensuremath{\mathcal C}} }

\newcommand{\cL}{{\ensuremath{\mathcal L}} }
\newcommand{\cT}{{\ensuremath{\mathcal T}} }
\newcommand{\cD}{{\ensuremath{\mathcal D}} }

\newcommand{\bP}{{\ensuremath{\mathbf P}} }
\newcommand{\bE}{{\ensuremath{\mathbf E}} }

\newcommand{\bN}{{\ensuremath{\mathbf N}} }

\DeclareMathSymbol{\leqslant}{\mathalpha}{AMSa}{"36} 
\DeclareMathSymbol{\geqslant}{\mathalpha}{AMSa}{"3E} 
\DeclareMathSymbol{\eset}{\mathalpha}{AMSb}{"3F}     
\newcommand{\dd}{\,\text{\rm d}}             
\newcommand{\inftwo}[2]{\inf_{\substack{#1 \\ #2}}} 
\newcommand{\sumtwo}[2]{\sum_{\substack{#1 \\ #2}}} 
\newcommand{\limtwo}[2]{\lim_{\substack{#1 \\ #2}}}     

\DeclareMathOperator{\Arg}{Arg}

\newcommand{\bbC}{{\ensuremath{\mathbb C}} }

\newcommand{\bbE}{{\ensuremath{\mathbb E}} }

\newcommand{\bbN}{{\ensuremath{\mathbb N}} }

\newcommand{\bbP}{{\ensuremath{\mathbb P}} }

\newcommand{\bbR}{{\ensuremath{\mathbb R}} }

\newcommand{\bbZ}{{\ensuremath{\mathbb Z}} }

\newcommand{\ga}{\alpha}
\newcommand{\gb}{\beta}
\newcommand{\gd}{\delta}
\newcommand{\gep}{\varepsilon}       
\newcommand{\gp}{\varphi}

\newcommand{\gz}{\zeta}

\newcommand{\gD}{\Delta}

\newcommand{\go}{\omega}

\newcommand{\gl}{\lambda}

\newcommand{\gs}{\sigma}

\makeatletter
\def\captionfont@{\footnotesize}
\def\captionheadfont@{\scshape}

\long\def\@makecaption#1#2{%
  \vspace{2mm}
  \setbox\@tempboxa\vbox{\color@setgroup
    \advance\hsize-6pc\noindent
    \captionfont@\captionheadfont@#1\@xp\@ifnotempty\@xp
        {\@cdr#2\@nil}{.\captionfont@\upshape\enspace#2}%
    \unskip\kern-6pc\par
    \global\setbox\@ne\lastbox\color@endgroup}%
  \ifhbox\@ne 
    \setbox\@ne\hbox{\unhbox\@ne\unskip\unskip\unpenalty\unkern}%
  \fi
  \ifdim\wd\@tempboxa=\z@ 
    \setbox\@ne\hbox to\columnwidth{\hss\kern-6pc\box\@ne\hss}%
  \else 
    \setbox\@ne\vbox{\unvbox\@tempboxa\parskip\z@skip
        \noindent\unhbox\@ne\advance\hsize-6pc\par}%
\fi
  \ifnum\@tempcnta<64 
    \addvspace\abovecaptionskip
    \moveright 3pc\box\@ne
  \else 
    \moveright 3pc\box\@ne
    \nobreak
    \vskip\belowcaptionskip
  \fi
\relax
}
\makeatother
\def\writefig#1 #2 #3 {\rlap{\kern #1 truecm
\raise #2 truecm \hbox{#3}}}

\newcommand{\tf}{\textsc{f}}

\newcommand{\newnorm}[1]{{\left\vert\kern-0.25ex\left\vert\kern-0.25ex\left\vert #1 
    \right\vert\kern-0.25ex\right\vert\kern-0.25ex\right\vert}}

\begin{document}

\title[The zeros of the partition function of the pinning model]{The zeros of the partition function\\ of the pinning model}

\author[G. Giacomin]{Giambattista Giacomin }
\address[G. Giacomin (corresponding author)]{Universit\'e Paris Cit\'e,  Laboratoire de Probabilit{\'e}s, Statistiques  et Mod\'elisation (UMR 8001),  8 place Aur\'elie Nemours,
            F-75205 Paris, France.  \emph{E-mail address: } {\tt giambattista.giacomin@u-paris.fr}}
           
\author[R. L. Greenblatt]{Rafael L. Greenblatt}
\address[R. L. Greenblatt]{Scuola Internazionale Superiore di Studi Avanzati, Mathematics Area, 
	via Bonomea 265,
 34136 Trieste, Italy}

\begin{abstract}
We aim at understanding for which (complex) values of the  potential 
 the pinning partition function vanishes. The pinning model is  a  Gibbs measure based on discrete renewal processes with power law inter-arrival distributions. 
We obtain some results for  rather general inter-arrival laws, but we achieve a substantially more complete understanding   for 
a specific one parameter family of inter-arrivals. We show, for such a specific family, that the zeros asymptotically lie on (and densely fill) a closed curve
that, unsurprisingly, touches the  real axis only in one point (the critical point of the model). 
We also perform a sharper analysis of the zeros close to the critical point and we exploit this analysis  to 
approach the challenging problem of Griffiths singularities for the disordered pinning model. The techniques we exploit are both probabilistic
 and analytical.  Regarding the first, a central role is played by  limit theorems for heavy tail random variables. As for the second, 
 potential theory and  singularity analysis of generating functions, along with their interplay, will be at the heart of several of our arguments.

\bigskip

\noindent  \emph{AMS  subject classification (2020 MSC)}:
82B27, 
30C15,  
31B05, 
60E10,  
82B44, 
60K35, 

\end{abstract}

\maketitle

\section{Introduction and results}

\subsection{The pinning  model: general framework}
\label{sec:intro0}
We denote by $\tau=(\tau_j)_{j=0,1,2, \ldots}$ a discrete recurrent renewal process with $\tau_0=0$. So 
$(\tau_{j+1}-\tau_j)_{j=0,1,2, \ldots}$ is a sequence of IID random variables taking values in $\bbN:=\{1,2, \ldots\}$. 
Often one considers the rather general framework   that for $n \in \bbN$
\begin{equation}
\label{eq:Kgeneral}
K(n)\, :=\, \bP(\tau_1=n) \stackrel{n \to \infty} \sim \frac{c}{ n^{1+\ga}}\ \text{ and } K(1)>0\,,  
\end{equation}
with 
$\ga\in (0,1)$ and $c>0$ and, unless  otherwise stated, we make the choice $c= 1/( -\Gamma(-\ga))$.
Of course $K(\cdot)$, called \emph{inter-arrival distribution}, determines the law of $\tau$.
The requirement 
 $K(1)>0$ is not essential, but allowing $K(1)$ to be zero  does complicate some arguments and notations (see comment after \eqref{eq:PN}).  
Note that any of the two requirements in \eqref{eq:Kgeneral} implies that $\tau$ is aperiodic and that recurrence means that $\sum_n K(n)=1$.

\medskip

\begin{rem}
\label{rem:a_n}
The precise value of $c$ has only a   minor impact on the model.
The choice we make of $c$   is customary when dealing with  stable laws because it simplifies some expressions. In fact \cite[pp.~448-449]{cf:Feller2} we know that there exists a sequence of positive real numbers $(a_n)$ such that 
 $\tau_n/a_n$ converges to the (stable) limit law with support on the positive semi-axis and with Laplace transform  $s\mapsto \exp( -s^\ga)$. The normalizing sequence $(a_n)$ turns out to  be asymptotically proportional to $n^{1/ \ga}$, so we can choose $a_n$ equal to $n^{1/ \ga}$ times a positive constant: this constant is equal to $1$ if  $c= 1/( -\Gamma(-\ga))$. 
 \end{rem}
 \medskip

We consider
\begin{equation}
\label{eq:Z}
Z_{N,h}\,:=\, \bE \left[ \exp \left( h \sum_{n=1}^N \gd_n\right) \gd_N\right]\, ,
\end{equation}
where $ \gd_n:=\ind_{n \in \tau}$:  we are viewing $\tau=\{\tau_0, \tau_1, \ldots\}$ as a random subset of $ \bbN \cup \{0\}$.
We will work with $h \in \bbC$, but let us first consider the case $h \in \bbR$.
 It is straightforward to see that, in this case, $(\log Z_{N, h})$ is a super-additive sequence, so the limit
\begin{equation}
\label{eq:F}
\tf(h)\, :=\,  \lim_N \frac 1N \log Z_{N,h}\, ,
\end{equation}
exists for every $h \in \bbR$ and it is equal to the supremum of the sequence. Moreover, $\tf(h)$ can be identified via an elementary computation, e.g. \cite[pp. 7 and 8]{cf:G} and the result is that 
\begin{equation}
\label{eq:Fh-expl}
\tf(h)\, =\begin{cases}
	\text{unique solution } \tf \text{ of } \sum_n K(n) \exp(-n \tf) = \exp(-h) & \text{ if } h\ge 0\, ,
\\
0 & \text{ if } h< 0\,.
\end{cases}
\end{equation}  
It can be seen  from \eqref{eq:Fh-expl} that 
$h\mapsto \tf(h)$ is (strictly) increasing and strictly convex on the positive semi-axis, while it is non decreasing and convex over all $\bbR$: all these properties can be extracted also directly from \eqref{eq:F}. It is  also clear that $\tf(\cdot)$ is not (real) analytic at the origin.
What \eqref{eq:Fh-expl}  tells us beyond this is that the origin is the only singular point (\emph{critical point}): $\tf \mapsto   \sum_n K(n) \exp(-n \tf)$ is a real analytic invertible map   from $(0, \infty)$ to $(0, 1)$ so $h \mapsto \tf(h)$ is real analytic simply because $ \tf(h)$ is obtained by applying the  inverse of the map to $\exp(-h)$.

On the other hand, $Z_{N,h}$ is just a polynomial of degree  $N$ in $\exp(h)$ and in fact it is sometimes more practical to use the polynomial notation 
\begin{equation}
\label{eq:PN}
P_N(w)\, :=\, Z_{N, \log w}\, . 
\end{equation}
Note that the degree of $P_N(w)$ is $N$ because we are assuming $K(1)>0$. If $K(1)=0$ the degree would be smaller, for example
if $K(1)=0$ and $K(2)>0$ then $P_N(w)$ is a polynomial of degree 
$\lfloor N/2\rfloor$. Choosing $K(1)>0$ hence simplifies the normalization of 
the empirical probability of the zeros. Other (non essential, albeit welcome) simplifications due to this choice are connected  to 
$P_N(w)\sim K(1) w$  for $w$ small.

\smallskip

\begin{rem}
\label{rem:period-2pi}
We mostly work with the variable $h$ which is more natural in the statistical mechanics language. It must be however 
noted that $Z_{N, h}$ is $2\pi$-periodic in $\Im(h)$: this periodicity is  just the periodicity of the exponential function in the imaginary direction. So, strictly speaking $Z_{N, h}$ always has infinitely many zeros, but they are just periodic copies of 
the $N-1$ zeros in $\bbC$ with imaginary part (say) in $(-\pi, \pi]$:  note that  the origin is a simple zero of $P_N(w)$ (again, $K(1)>0$), but this zero is at $-\infty$  for $Z_{N,h}$.
It is therefore natural to introduce $\bbC_{2\pi}:=\bbC/(2\pi i \bbZ)$, which 
we will identify with $\{z\in \bbC: \, \Im(z) \in (-\pi, \pi]\}$, and restrict ourselves to this set when dealing with the $N-1$ zeros of $Z_{N, h}$.   
\end{rem}
\smallskip

We refer to \cite{cf:Fisher,cf:GB,cf:G,cf:dH} for a thorough discussion of the model in statistical mechanics terms: the \emph{critical point} $h=0$ captures a delocalization ($h<0$) to localization ($h>0$) transition. Taking 
the Lee-Yang viewpoint \cite{cf:LY}, we remark that $(1/N)\log  Z_{N, h}$ is real analytic on the whole of $\bbR$. But $Z_{N, h}$ is an entire function and the singularities in the complex plane of $(1/N)\log  Z_{N, h}$
 are due to the zeros of $Z_{N, h}$: in the limit $N \to \infty$ these singularities may accumulate on the real axis. In our case, they are going to accumulate on the real axis only at zero. Our purpose is to determine the location of the zeros of 
$Z_{N, h}$ for $N \to \infty$.  We stress that, unless otherwise stated, by $\log(\cdot)$ we mean the principal branch of the complex logarithm:  this is discussed more in detail after \eqref{eq:Freal} where we introduce the notation Log$(\cdot)$ for the principal branch, but Log$(\cdot)$ will be used only when strictly needed.

Even if results for pinning models with $h \in \bbR$ are typically obtained assuming only
\eqref{eq:Kgeneral} (or in even wider frameworks: for example regularly varying inter-arrival distributions \cite{cf:GB}) and the results \emph{essentially} depend only on $\ga$, it appears to be really
  challenging to extend such a universal behavior to $h\in \bbC$. So, we will give some results assuming only \eqref{eq:Kgeneral}, but we are able to obtain a  good control on the location of the zeros when working with a much more restrictive choice of  $K(\cdot)$.
   But let us start with a general result that 
holds in our most general framework: 
  
  \medskip
  
  \begin{proposition}
  \label{th:res0}
  We fix $K(\cdot)$ that satisfies 
  \eqref{eq:Kgeneral}. Then there exists $C>0$ and, for every $\gep>0$ and $\gep'>0$, 
  there exist $N_{\gep, \gep'}\in \bbN$, $C_{\gep'}>0$  and a subset $V_{\gep, \gep'}$ of the complex plane that contains
  \begin{enumerate}
  \item the half plane $\Re(h) \le -\gep$;
  \item the half plane $\Re(h) \ge C$; 
  \item the set of $h$'s with $\Re(h) > \gep' $ and $\vert \Im(h)\vert< C_{\gep'}$;
  \end{enumerate} 
   such that 
  $Z_{N, h}\neq 0$ for every $h\in V_{\gep, \gep'}$ and every $N\ge N_{\gep, \gep'}$.
  \end{proposition}
  
  \medskip
  
 This statement is  visualized in Figure~\ref{fig:D0}.  
 
  \begin{figure}[h]
\centering
\includegraphics[width=11 cm]{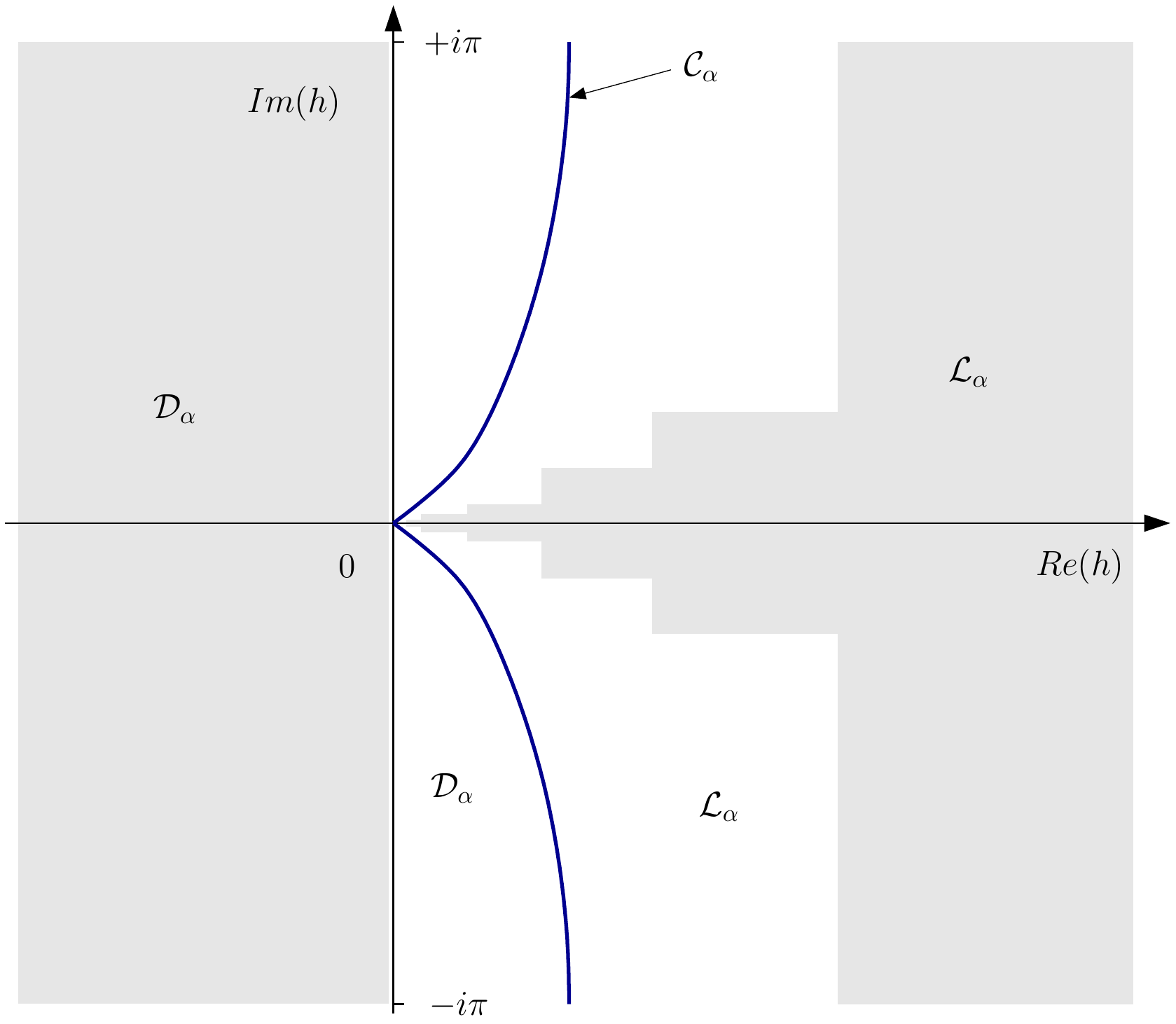}
\caption{\label{fig:D0} 
A pictorial  vision of the content of Proposition~\ref{th:res0}, and beyond: there is no zero in the shadowed region for $N$ sufficiently large.
The shadowed region is obtained by applying Proposition~\ref{th:res0} for more than one value of $\gep'$, in particular a very small value.
 In the restricted set-up with inter-arrival distribution \eqref{eq:Kspecial} we can show that 
the zeros asymptotically lie on a \emph{critical} curve $\cC_\ga$ which splits
$\bbC_{2\pi}$  into a \emph{delocalized} region $\cD_\ga$, in which $\vert Z_{N, h}\vert$ does not grow exponentially, and a \emph{localized} one $\cL_\ga$ in which $\vert Z_{N, h}\vert$  grows exponentially. 
We stress that the lower bound we obtain on $C_{\gep'}$, see Proposition~\ref{th:res0},  is $o(\gep' )$,  so, in the general framework,  we do not establish that there exists  $m>0$ such that there is  no zero in the cone $\{z\in \bbC:\, \Re(z)>0$ and $\vert \Im(z)\vert< m \Re(z)   \}$. This feature however does hold in the restricted set-up  \eqref{eq:Kspecial}.
}
\end{figure}

We will give a sketch of the proof of Proposition~\ref{th:res0} in \S~\ref{sec:sketch}: while Proposition~\ref{th:res0} is rather rough,  \S~\ref{sec:sketch} will be of help in understanding some of the tools we repeatedly use (also in the proof of much sharper results) and why in the general framework we are limited to Proposition~\ref{th:res0}.

\subsection{Special inter-arrival distributions}
\label{sec:special}

A special choice of $K(\cdot)$ for which we are able to go much farther is 
\begin{equation}
\label{eq:Kspecial}
K(n)\, =\, \frac{\Gamma(n-\ga)}{-\Gamma(-\ga)\,  n!}\,  =\, \frac{-(n-1-\ga) (n-2-\ga) \cdots (1-\ga)(-\ga)}{n!}
 \stackrel{n\to \infty} \sim \frac{n^{-(1+\ga)}}{-\Gamma(-\ga)}\,.
\end{equation}
We remark also that $K(1)=\ga$. 
One of the important features of this distribution is that its $z$-transform (or characteristic function) has an explicit expression:
\begin{equation}
\label{eq:Kz}
\widehat K(z)\, :=\, 
 \sum_{j=1}^\infty  z^j K(j)\, =\, {1- \left( 1-z \right)^\ga}\, .
\end{equation}
The power series defining $\widehat K(z)$ has radius of convergence $1$, and this is of course true  also for the general framework \eqref{eq:Kgeneral}. The explicit expression in \eqref{eq:Kz} is saying that, with the special choice  \eqref{eq:Kspecial}, $\widehat K(z)$ can be extended to $\bbC \setminus \{z:\, \Re(z)\ge 1\}$. 

\smallskip

\begin{rem}
\label{rem:fabry}
This is not a generic feature: in fact, under the hypothesis \eqref{eq:Kgeneral}, one can  exhibit $K(\cdot)$ such that 
$\widehat K(z)$ has a natural boundary on the unit circle. For example, if $K_A(n)=C \ind_A(n)/n^2$ with 
$A=\{n_j:\, j\in \bbN\}$ with $n_j/j \to \infty$, $C>0$ chosen so that $\sum_n K_A(n)=1$ and if we further assume that $n_j$ does not diverge too fast so that the series defining $\widehat K_A(z)$ has 1 as radius of convergence 
(take for example $n_j=j^2$) then $\widehat K_A(\cdot )$ is singular everywhere on $\partial B_0(1)$
\cite[Ch. XI, in particular p. 373]{cf:Dienes}. 
So the inter-arrival distribution $K(n)=(K_1(n)+K_A(n))/2$, with $K_1(\cdot)$ given in \eqref{eq:Kspecial}, 
satisfies \eqref{eq:Kgeneral} (except possibly for the value of $c$, but this can easily be fixed) and  has $\partial B_0(1)$ as natural boundary.  
\end{rem}
\smallskip

Let us also observe right away that by applying the formula for $\tf(h)$ given right after \eqref{eq:F} we have that the choice  \eqref{eq:Kspecial} yields 
 for $h>0$
\begin{equation}
\label{eq:Freal}
\tf(h)\, =\, - \log \left( 1 -\left( 1- \exp(-h)\right)^{1/\ga} \right)\, .
\end{equation}
How far into $\bbC$ can this function be analytically continued? One problem comes from 
$h\mapsto (1-\exp(-h))^{1/\ga}$ that has a cut discontinuity starting at origin (unless
$1/\ga =2,3 , \ldots$  for which $h\mapsto (1-\exp(-h))^{1/\ga}$ is entire). Unless otherwise stated, by $z^c$,  $c \in \bbR$, we mean  $\exp(c \Log z)$ with
$\Log(\cdot)$ the principal branch of the logarithm (for $z\in (-\infty,0)$ we set $\Log z:= \log \vert z \vert +i \pi$) . In particular, with this choice, the cut of $h\mapsto (1-\exp(-h))^{1/\ga}$ is on $(-\infty,0)$. 
Therefore, if  $\log_R(\cdot)$ is the logarithm defined from its natural Riemann surface (infinitely many copies of $\bbC$) to $\bbC$ \cite[Ch.~8]{cf:A}, we have that
\begin{equation}
\label{eq:Fcomplex}
\tf(h)\, :=\, - \log_R \left( 1 -\exp\left( (1/ \ga)\Log \left( 1- \exp(-h)\right)^{1/\ga}\right) \right)\, ,
\end{equation}
is analytic on $\bbC \setminus (-\infty,0]$ and coincides with $\tf(h)$ for $h>0$. 

On the other hand $\tf(h)=0$ for $h<0$, 
which of course can be continued to the whole $\bbC$: so, at this stage there is no reason to believe that the continuation defined in \eqref{eq:Fcomplex} is relevant, at least not over the whole region where we have defined it. 

In terms of \emph{continuation} of the real free energy outside $\bbR$ the related \emph{harmonic continuation} is a priori more straightforward: $\Re(\tf(h))$ (cf. \eqref{eq:Fcomplex}) is harmonic on $\bbC \setminus (-\infty,0]$ when viewed as a function of two real variables (just set $h=x+iy$).  Such harmonic extension is at this stage equivalent with the analytic one, but it avoids the arbitrary choice of the branch of the logarithm. The problem of choosing the branch of the $\log$ is present also at fixed $N$, but it is avoided if we simply consider the harmonic function $\Re( \log Z_{N, h})= \log \vert Z_{N, h}\vert$. We will see (Section~\ref{sec:potential}) that $\lim (1/N)  \log \vert Z_{N, h}\vert$ is not only uniquely defined, but it also converges in the whole of $\bbC$. 
The limit coincides with the harmonic continuation $\Re(\tf(h))$ of the free energy on the positive real axis 
 up to where it vanishes, or in other words where it matches the harmonic continuation from the  negative real axis. More precisely, we are going to show that  the connected component of the set $\{h:\,  \Re(\tf(h))>0\}$ that contains the positive real axis is a subset of the half plane $\{z \in \bbC: \, \Re(z)>0\}$ and the \emph{relevant continuation} of $\tf(\cdot)$, defined on $\bbR$, is $\Re(\tf(h))(>0)$ on this connected component, and it is zero on the rest of $\bbC$.
Hence the \emph{critical region} is  identified by the values of $h$ with $\Re(h)\ge 0$ and  $\tf(h)=i \theta$ for some $\theta \in \bbR$. It is not too difficult to see that this set can be  written more explicitly as 
\begin{equation}
\label{eq:crit}
\cC_\ga \, :=\,  \left \{ - \Log \left( 1-(1- \exp(-i \theta))^\ga\right):\,  \theta \in [0, 2\pi)\right\}\, .
\end{equation}
This set appears  in Figures~\ref{fig:D0}, \ref{fig:zeros-horizon} and \ref{fig:curve4}). Here are some properties:

\medskip

\begin{lemma}
\label{th:crit}
We have that $\cC_\ga$ (see Fig.~\ref{fig:D0} and Fig.~\ref{fig:zeros-horizon}) is invariant under complex conjugation,  that $\cC_\ga$ is a subset of the strip $0 \le \Re(h)\le -\log(2^\ga-1)$ and touches the boundary of this strip only at the origin and at $-\log(2^\ga-1)+ i\pi$.
Moreover $\cC_\ga$ is a simple closed curve in the cylinder $\bbC_{2\pi}$.  This curve is 
  smooth, except at $0$, of finite length and it is  not homotopic to a point: hence  $\bbC_{2\pi}\setminus \cC_\ga$ is the union of  two disjoint connected sets that we call $\cL_\ga$ and $\cD_\ga$. $\cL_\ga$   contains the positive real axis and 
   $\cD_\ga$ contains the $\Re(h)<0$ half plane. 
\end{lemma}

\medskip

\begin{rem}
The pinning model transition that we observe at $h=0$, see for example \eqref{eq:Fh-expl}, is a (de)localization transition: this is discussed at length for example in \cite{cf:Fisher,cf:GB,cf:G}. Of course 
 $\cL_\ga$, $\cD_\ga$ and $\cC_\ga$ are, respectively, the continuation in the complex plane of the delocalized region $(-\infty,0)$, of the localized region $(0,\infty)$ and of the critical point $0$. We stress that we do not know how to do this continuation (at least, not in such a complete sense) in the general framework \eqref{eq:Kgeneral}. 
 Moreover we do not attach a pathwise sense to the notion of (de)localization for  $h \in \bbC \setminus \bbR$. Nevertheless it may be natural to identify $\cL_\ga$ (in the general context) as the region in which 
 $\liminf_N (1/N)\log\vert Z_{N, h}\vert>0$, see notably the caption of Figure~\ref{fig:D0}, the content of Section~\ref{sec:alternative} and  Remark~\ref{rem:L2}. 
\end{rem}

\medskip

It turns out that the zeros of $Z_{N,h}$ accumulate on $\cC_\ga$: this is what we  explain next.

  \begin{figure}[h]
\centering
\includegraphics[width=14 cm]{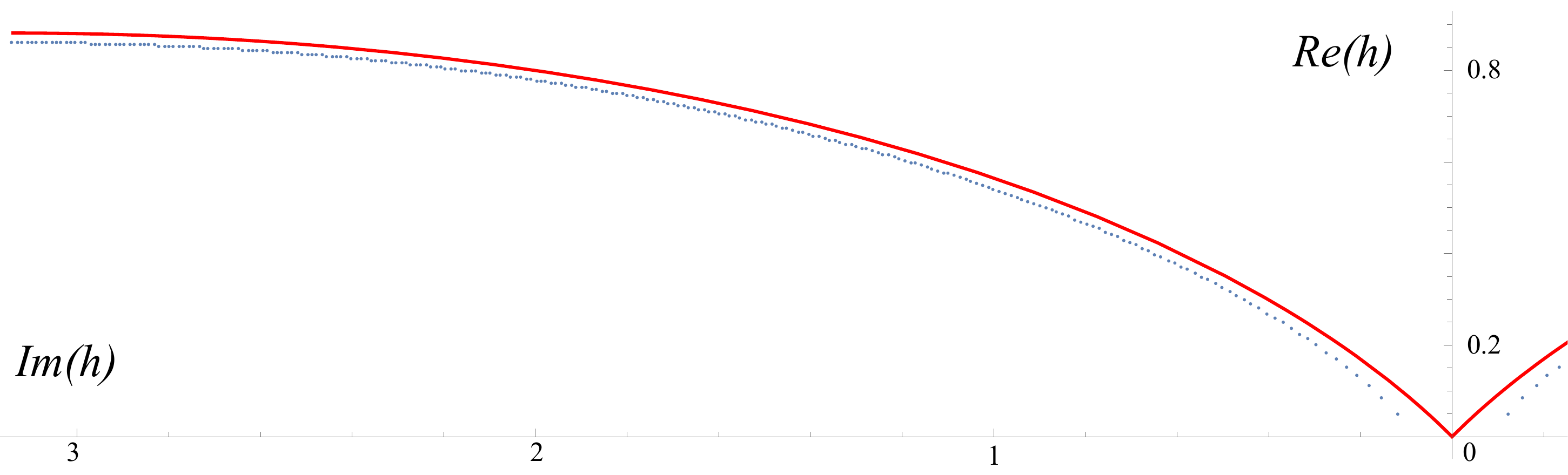}
\caption{\label{fig:zeros-horizon}
In red the plot of the critical curve $\cC_{1/2}$, hence we are working with the inter-arrival law \eqref{eq:Kspecial}: the complex axes are rotated by $90$ degrees and we cut the part with $\Im(h)<0$ which is just the specular image of 
$\Im(h)>0$.  The blue dots mark the locations 
  of the zeros of $Z_{500, h}$. 
	 Note that the black dots in Figure~\ref{fig:lac} show comparable zeroes, corresponding to a different inter-arrival law \eqref{eq:Kgeneral}, but still  with $\ga=1/2$. More precisely  the black dots in Figure~\ref{fig:lac} mark the zeros of $P_{500} (w)$, so the black dots in Figure~\ref{fig:lac} 
should be compared  
 to  the exponential of the blue dots in here (and the qualitative behavior is the same).
 } 
\end{figure}

\subsection{The zeros of the partition function: macroscopic limit}
We call $h_{N,1}, h_{N,2}, \ldots,$ $h_{N,N-1}$ the $N-1\ge 1$ zeros of $Z_{N, h}$ in $\bbC_{2\pi}$ 
and we introduce the empirical probability
\begin{equation}
\mu_N\, :=\, \frac 1 {N-1} \sum_{j=1}^{N-1} \gd_{h_{N,j}}\, ,
\end{equation}
where $\gd_z$ is the  probability on $\bbC_{2\pi}$ that is concentrated on $z$. We remark that if $Z_{N, h}=0$ then  
$Z_{N, \bar h}=0$, so $\mu_N$ is symmetric with respect to conjugation.

 The notion of convergence of probability measures is the standard notion of convergence in law: $\lim_N \mu_N= \mu$ if $\lim \int_{\bbC_{2\pi}} f \dd \mu _n = \int_{\bbC_{2\pi}} f \dd \mu$ for every $f: \bbC_{2\pi} \to \bbR$ which is continuous and bounded. In particular,  $\mu$ is a probability too. 
\medskip

\begin{theorem}
\label{th:empprobab}
With the inter-arrival distribution \eqref{eq:Kspecial}  we have that $\lim_N \mu_N= \mu$.
The support of $\mu$ coincides with $\cC_\ga$. Moreover
  $\mu$ is absolutely continuous with respect to the  arc-length measure on $\cC_\ga$ and its density vanishes only at $0$.
\end{theorem}
\medskip

The argument of proof of Theorem~\ref{th:empprobab} -- i.e., sharp asymptotic control on $ Z_{N, h}$ -- directly yields the following result that 
makes more evident the role of the harmonic continuation of the free energy.

\medskip

\begin{theorem}
\label{th:fe}
With the inter-arrival distribution \eqref{eq:Kspecial}  we have that 
\begin{enumerate}
\item
if $h\in \cD_\ga$ then $\lim_N (1/N) \log \vert Z_{N, h} \vert =0$;
\item 
if $h\in \cL_\ga$ then
\begin{equation}
\lim_{N\to \infty} \frac 1N \log \left \vert Z_{N, h} \right\vert\, = \, \Re\left( \tf(h)\right) \, ,
\end{equation}
where  $ \tf(h)$ is given in  \eqref{eq:Fcomplex}. 
\end{enumerate}
Both results hold uniformly if $h$ is chosen bounded away from $\cC_\ga$.

Moreover $\lim_N (1/N) \log \vert Z_{N, h} \vert =0$ also if $h \in \cC_\ga$. 
\end{theorem}
\medskip

\subsection{The zeros of the partition function: local control}
Theorem~\ref{th:empprobab} appears to be definitely sharper than Proposition~\ref{th:res0}, but we draw the attention on the fact that
Theorem~\ref{th:empprobab} is just about the empirical probability of the zeros and $o(N)$ of the zeros may behave in an arbitrary way without affecting the empirical probability. However, we do have also a stronger result

\medskip

\begin{theorem}
\label{th:sharp-loc}
With the inter-arrival distribution \eqref{eq:Kspecial}, for every $\gep>0$ there exists $N_\gep\in \bbN$ such that
the distance between  the support of $\mu_N$ and $\cC_\ga $ is smaller than $\gep$ for $N\ge N_\gep$.
\end{theorem}
\medskip

Theorem~\ref{th:sharp-loc} is therefore saying that all the zeros are at distance $\gep$  from the support of the limit empirical probability, and this is an important improvement on Theorem~\ref{th:empprobab}. But it is  still a very imprecise result near the most interesting point, that is $h=0$. In particular, it is straightforward to check that $\cC_\ga $ near the origin is asymptotically close to the angle 
$\{z: \, \vert\Arg(z)\vert = \ga \pi/2 \}$.  It thus seems natural to conjecture that the zeros  are close to $\{z: \, \vert\Arg(z)\vert = \ga \pi/2 \}$ when we look at the zeros that are very close to the origin, i.e.\ the zeros that are at a distance from the origin that vanishes for $N\to \infty$: for $\ga=1/2$ we will  show that the closest zeros are at a distance proportional to $1/ \sqrt{N}$ from the origin and that in a ball of radius $r/ \sqrt{N}$ we can find arbitrarily many zeros by choosing $r$ large, in the limit $N \to \infty$. However the  natural distance to consider on this scale is the one  rescaled by the size of the neighborhood we are considering.  Whether the zeros  are close to $\{z: \, \vert\Arg(z)\vert = \ga \pi/2 \}$ or not on such intermediate scales depends on a finer analysis that one can in principle deal with for every $\ga\in (0,1)$, but there are some obstacles to this analysis. The most important one is the implicit characterization  of the $\ga$-stable laws. There is however the notable 
exception of $\ga=1/2$ and, in fact, in this case we are able to go rather far.  Here and in all the rest of the paper we order $( h_{N,j})$ so that  $(\vert h_{N,j}\vert)$ is non decreasing.
 We actually choose $h_{N,1}$  one of the closest zeros to the origin with argument  in $[0, \pi]$.  Moreover $h_{N,2}:= \overline{ h_{N,1}}$, unless 
$h_{N,1}$ is real.

The next result is therefore restricted to $\ga=1/2$ (but we stress that it does not require the special set-up  \eqref{eq:Kspecial}). To state it we introduce the entire function 
\begin{equation}
\label{eq:F0zeta}
F_0(\zeta) \, :=\,  e^{\zeta^2}\zeta\left(1+\mathrm{erf}(\zeta)\right)+\frac{1}{\sqrt{\pi }}\, ,
\end{equation}
where erf$(\zeta)=(2/\sqrt{\pi}) \int_0^\zeta e^{-t^2} \dd t $ is the error function \cite[Ch.~7]{cf:DLMF}. 
We show in Lemma~\ref{th:zeta_1} that all the solutions  $F_0(\zeta_j)=0$ are given by an infinite  sequence $(\zeta_j)_{j=1, 2, \ldots }$ that can be ordered in such a way that the modulus is non decreasing and  $\zeta_{2j-1}=  \overline{\zeta_{2j}}$. With this we are implicitly saying that there is no real zero. Moreover,   all the zeros have  positive real part and they are all 
 simple (cf.~Remark~\ref{rem:simplezeros}). 

\medskip 

\begin{theorem}
\label{th:closest}
With the inter-arrival distribution \eqref{eq:Kgeneral}, $\ga=1/2$ and possibly by properly arranging the order 
of $(h_{N,j})$  we have that for every $j$ 
\begin{equation}
\label{eq:closest}
h_{N,j} \stackrel{N \to \infty} \sim \frac{\zeta_j}{N^{1/2}}\, .
\end{equation} 
\end{theorem}

\medskip

In particular we will see that $\zeta_1= 1.225 \ldots + i 2.547\ldots $, that is $\Arg(\zeta_1)=1.122\ldots = (\pi/4) 1.429\ldots$.
This  is therefore saying that the closest zeros to the origin are not close to $\{z: \, \vert\Arg(z)\vert = \ga \pi/2 \}$. 

 It is highly plausible that $(\vert \zeta_{2j-1} \vert)$ is strictly increasing for every $j$ and therefore so is 
 $(\vert h_{N,2j-1}\vert )$ for $N$ sufficiently large: if this is the case we avoid the nuisance of properly choosing the 
 ordering of  $(h_{N,j})$ when there are values of $j$ such that $\vert h_{N,2j-1}\vert= \vert h_{N,2j+1}\vert$.

\subsection{The disordered pinning model and  Griffiths singularities}
\label{sec:introG}
The partition function of the  disordered version of the pinning model is
\begin{equation}
\label{eq:Zdismod}
Z_{N, \go,h} \, :=\, \bE\left[
\exp\left( \sum_{n=1}^{N} (\go_n+h) \gd_n\right) \gd_N\right]\, ,
\end{equation}
where  $\go=(\go_n)_n$ is an IID sequence, and $h$ a real parameter. Results on this disordered model have been obtained under mild conditions on the law of $\go_1$, but let us choose $\go_1$ that takes just two values. And, to be ready for the specific analysis we want to perform, 
we choose $\go_n= sb_n$, with $s$ a real number and $(b_n)$ IID Bernoulli$(1-p)$, $p \in (0,1)$. Moreover, $(b_n)$ and $(\tau_n)$ are independent and we are therefore working on a probability space that is the product of the space in which the disorder variables are defined, and the space on which the renewal process is defined: the probability is then the product probability of $\bbP$ and $\bP$.  
Note that if $s=0$ then the model is non disordered -- sometimes called \emph{pure} -- and $Z_{N, \go,h}= Z_{N, h}$. The reason for this atypical 
choice of disorder is that we are going to be interested in the limit $s \nearrow \infty$, but let us recall some general facts for now. 

First of all the \emph{free energy density} $\tf_s(h)= \lim_N (1/N) \bbE \log Z_{N, \go,h}$
exists and it is a convex non decreasing function of $h$. In fact,
$\tf_s(h)=0$ for every $h\le h_c(s, p)\in \bbR$ and $h\mapsto \tf_s(h)$ is strictly increasing (and strictly convex) for 
$h> h_c(s, p)$. Many estimates are available on the  value of $h_c(s, p)$, and in some cases it can even be computed exactly, but this will not be important for us here. While it is clear that $h_c(s, p)$ is a critical point, i.e. 
$h\mapsto \tf_s(h)$ is not real analytic at $h_c(s, p)$, the only result available on the regularity of $h\mapsto \tf_s(h)$ for $h>h_c(s, p)$ for $s\neq 0$ is that it is $C^\infty$ \cite{cf:GTAlea}. On the other hand, for $s=0$ (non disordered case), 
the free energy density $h\mapsto \tf_0(h)=\tf(h)$ is real analytic except at $h=h_c(0, p)=0$. The transition at $h_c(s, p)$
is a delocalization to localization transition as explained in detail for example in \cite{cf:GB,cf:G,cf:dH}: we refer to \cite{cf:BL} for an updated bibliography.

The obstruction to showing analyticity in the presence of disorder is not just a technical problem: R. B. Griffiths showed in 1969 that disorder may induce singularities. Griffiths' full argument was given for the ferromagnetic Ising model with dilution; that is, Ising model on a  lattice, $\bbZ^d$, in which some bonds are deleted.  In spite of a large amount of literature on the issue, the  understanding of Griffiths singularities is still poor. In particular Griffiths singularities are expected to be rather generic, but their existence is proven only in very specific cases (for example, in presence of dilution, which corresponds to introducing  infinite potentials in the system). 

It is very natural to ask whether Griffiths singularities are present for the pinning model: is $h\mapsto \tf_s(h)$
analytic for $h>h_c(s, p)$ or are there other non analyticity points or regions?
This question has been tackled in \cite{cf:KM} by considering the $s\nearrow \infty $ limit of the model we just introduced. To deal with this limit
it is practical to consider also the discrete renewal  process $\gs=(\gs_n)$ that marks the sites where $b_n=1$  and set $N_\gs:= \sup \{j:\, \gs_j \le N \}$, with $\gs_0:=0$. By this we mean that, if $N_\gs >0$, 
$\{\gs_1, \ldots, \gs_{N_\gs}\}=\{n =1, \ldots, N:\, b_n=1\}$. 
Otherwise $\{n =1, \ldots, N:\, b_n=1\}$ is empty.
Separating out the contribution of the realizations where the renewal process $\tau$ visits all of the sites in $\gs$, we have
\begin{equation}
\label{eq:sinfty}
Z_{N, \go,h} \, =\, 
\exp \left( N_\sigma s \right) 
\left(\prod_{j=1}^{N_\gs} Z_{\gs_j- \gs_{j-1}, h}  \right)
Z_{N- \gs_{N_\gs}, h} 
+ O\left( \exp\left((N_\sigma -1\right)s \right)
\end{equation}
for $N$ fixed and $s \nearrow \infty$.

 It is straightforward to check that $\lim_{s\nearrow \infty} h_c(s,p)=-\infty$, so the limit model is always localized. 
One can now consider as reduced model the first term in the right-hand side of \eqref{eq:sinfty} 
and the free energy density of this model 
is (of course a.s. $\lim_N N_\gs /N =1-p$) 
\begin{equation}
s(1-p) + (1-p)^2 \sum_{n=1}^\infty p^{n-1} \log Z_{n, h}\,,
\end{equation}
where we have used the Law of Large Numbers: $\bbE[\vert \log Z_{\gs_1, h} \vert ]< \infty$ because 
 $e^h K(n) \le Z_{n, h} \le \exp(n \max(h,0))$.
Note that  the existence of a Griffiths singularity in this reduced model boils down to 
determining whether 
\begin{equation}
h \mapsto \sum_{n=1}^\infty p^{n-1} \log Z_{n, h}\, =:\, \tilde \tf_p (h)\, ,
\end{equation}
is real analytic or not and the prediction is straightforward: 
$ \tilde \tf_p (\cdot)$ does have a singularity in zero, because  the zeros of $Z_{n, h}$ in the complex plane have a unique real  accumulation point, as $n \to \infty$, in the origin. 

The fact that the singularity is expected to happen at $h=h_c(0,p)=0$ in this specific model is very much in the spirit 
of Griffiths' idea. The critical point of the pure model ($s=0$) is $h=0$. For $s>0$ and large the system is essentially a collection of independent pure models pinned at the points on which $b_n=1$
and its (localization) critical point is $h_c(s, p)$. All of the pure systems in the collection are finite, so their contribution is analytic, but in this collection there are systems that are arbitrarily large (the larger, the fewer). 
And the larger they are, the less their contribution can be continued outside the real line in the proximity of $h=0$.
Therefore the total contribution is not analytic, but the free energy turns out to be  in any case $C^\infty$ at $h=0$ because  the large pure systems in the collection are \emph{exponentially rare}. 


\smallskip

Here is the result that we have:

\medskip

\begin{theorem}
\label{th:griffiths}
In the framework of \eqref{eq:Kspecial} with $\ga=1/2$, 
$h \mapsto \tilde \tf_p (h)$ is real analytic except at $0$ where for $k \to \infty$
\begin{equation}
\label{eq:griffiths}
\frac {\partial_h^k \tilde\tf_p (h)}{(k-1)!}
  \bigg \vert_{h=0}\, =\,  C_1 C_2^k \exp(A \sqrt{k})
 \Gamma \left(\frac  k 2 +1\right) \left(\cos\left( {\mathtt a} k + {\mathtt b} \sqrt{k} +  \mathtt{c} \right) + O \left( \frac {(\log k)^2}{\sqrt{k}} \right)\right)\, ,
\end{equation}
where $C_1$, $C_2$, $A$, ${\mathtt a}$, ${\mathtt b}$ and $\mathtt{c}$ are real constants that we give explicitly in the proof (see Remark~\ref{rem:constants}).
In particular, as a consequence of the fact that ${\mathtt b} \neq 0$,  we have that there exists 
a $\bN_0\subset \bbN$ of density zero in $\bbN$ such that for $k \to \infty$ with $k\notin \bN_0$ we have
\begin{equation}
\label{eq:griffiths-2}
\frac {\partial_h^k\tilde\tf_p (h)}{(k-1)!}
  \bigg \vert_{h=0}
\, \sim\,  C_1 C_2^k \exp(A \sqrt{k})
 \Gamma \left(\frac  k 2 +1\right) \cos\left( {\mathtt a} k + {\mathtt b} \sqrt{k} +  \mathtt{c} \right)\, .
\end{equation}
\end{theorem} 
\medskip

Theorem~\ref{th:griffiths} is
strongly related to Theorem~\ref{th:closest}, notably to  
\eqref{eq:closest} for the case $j=1$: the two zeros that are closest to the origin determine the leading behavior of the singularity.
However, to obtain \eqref{eq:griffiths} we have employed a substantial refinement of \eqref{eq:closest} in the case $j=1$: see Proposition~\ref{th:sharper0}.

A priori  \eqref{eq:griffiths} may not be very informative because $ \cos\left( {\mathtt a} k + {\mathtt b} \sqrt{k} +  \mathtt{c} \right) $ may be arbitrarily close to zero and $O( ( \log k)^2/ \sqrt{k})$ may become leading. But, as  we will explain in the proof, $\vert \cos\left( {\mathtt a} k + {\mathtt b} \sqrt{k} +\mathtt{c} \right)\vert \gg ( \log k)^2/ \sqrt{k}$ except on a density zero subsequence of values of $k$. This is spelled out in 
\eqref{eq:griffiths-2}, which implies non analyticity of $\tf_p (\cdot)$ at the origin because of the superexponential growth of the right-hand side in \eqref{eq:griffiths-2}.

Theorem~\ref{th:griffiths} also shows that the picture of the phenomenon (location of the zeros, Griffiths singularities) given in 
\cite{cf:KM}, while qualitatively to a certain extent correct, it is quantitatively imprecise. The limit of the analysis in \cite{cf:KM} is that it plays on the fact  that the zeros accumulate along the lines with slope $\pm \tan (\ga \pi /2)$ near the origin. This is true in an appropriate \emph{mesoscopic} sense, but, as we have seen, the leading Griffiths singularity of the reduced model (introduced in  \cite{cf:KM}) depends only on the two conjugate zeros closest to the origin, and they are not close (on the correct microscopic scale) to those two lines.

\subsection{About  the tools we use, organization of proofs, perspectives}
\label{sec:sketch}
\subsubsection{How we tackle the problem} 
Our approach mixes probabilistic tools and  analytic ones. 
We discuss in some detail the  proof of  Proposition~\ref{th:res0} because it contains some of the main tools we also use  for the sharper results that follow. 
A direct consequence of a basic result in renewal theory \cite[Th.~A]{cf:Doney} is that for $h< 0$ \cite[Ch.~2]{cf:GB}: 
\begin{equation}
\label{eq:fD}
Z_{N, h} \stackrel{N\to \infty} \sim \frac{e^h}{ (1-e^h)^2} K(N)\, .
\end{equation} 
Proposition~\ref{th:Reh<0} says that  \eqref{eq:fD} holds also in the complex plane, provided that $\Re(h)<0$.
Moreover this asymptotic behavior is uniform if $\Re(h)$ is bounded away from zero: this directly entails that, asymptotically,  
$Z_{N, h}\neq 0$ in the left complex half plane, and that no zero escapes to $-\infty$ as $N \to \infty$.  
The proof of Proposition~\ref{th:Reh<0} uses \cite[Th.~A]{cf:Doney} much in the same way as for \eqref{eq:fD}. 

On the other hand, we already know that $\log Z_{N, h} \sim N \tf(h)$, with $\tf(h)>0$ and increasing  for $h>0$, so we definitely expect that also for $\Re(h)>0$ and $\vert \Im(h)\vert $ \emph{somewhat small} with respect to  $\Re(h)$ the partition function still grows exponentially. 
In fact, we will show that for $\Re(h)$ sufficiently large, exponential growth holds regardless of the value of  $\Im(h)$. 
In order to make this concrete and quantitative we exploit  the singularity analysis of the $z$-transform (characteristic function). Recall \eqref{eq:Kz} for the notation: the $z$-transform of $(Z_{N, h})$ can be easily computed 
in terms of the $z$-transform of $(K(N))$. In fact 
with $Z_{0, h}:=1$ we have
\begin{multline}
\label{eq:z-t}
\widehat Z_h(z)\, :=\,
\sum_{n=0}^\infty z^n Z_{n,h}\, 
=\, 1+ \sum_{n=1}^\infty \sum_{k=1}^n \sumtwo{\ell \in \bbN^k: }{ \sum_j \ell_j=n} 
\prod_{j=1}^k \left( e^h z^{\ell_j} K(\ell_j) \right)
\\ =\, 1+ \sum_{k=1}^\infty \sum_{\ell \in \bbN^k}\prod_{j=1}^k \left( e^h z^{\ell_j} K(\ell_j) \right) \,
=\, 
\frac 1{1- e^h\sum_{j=1}^\infty  z^j K(j)}\,=\, 
\frac 1{1- e^h \widehat K(z)}\, .
\end{multline}
These steps are justified only for $\vert z \vert$ small, as it can be seen also from the rightmost term:
 the radius of convergence of  $\widehat K(z)$ is one, so $\widehat Z_h(z)$ is meromorphic in the unit disk.
The precise asymptotic behavior of $(Z_{N, h})$ can be obtained by analyzing the singularities of  $\widehat Z_h(z)$: in particular it is well known \cite{cf:FS} that, if
$1- e^h \widehat K(z_0)=0$ for $\vert z_0 \vert < 1$ (let us assume that there is a unique zero with minimal modulus and that  this zero, which we call $z_0$, is simple: of course general results are available) 
then the leading behavior of $Z_{N, h}$ for $N \to \infty$ is $\vert z_0 \vert^{-N}$ times an explicit $h$ dependent non zero constant.  
One can actually show that this result is uniform in a neighborhood of $h$ and, as before, this excludes  $Z_{N, h}=0$ in such a neighborhood and for $N$ sufficiently large. 

We are therefore at the level of the grey regions of Figure~\ref{fig:D0} and it is natural, in analogy with the real case, to dub as \emph{localized} the region in which the free energy has exponential growth: we could therefore define $\cL_\ga$ as the values of $h$ such that 
$1- e^h \widehat K(z)=0$ can be solved for $\vert z \vert< 1$. 
We have chosen to introduce $\cL_\ga$ only for the special one parameter family of inter-arrival laws in \eqref{eq:Kspecial}
because our main focus is on the location of the zeros. As a matter of fact we have defined first the critical curve $\cC_\ga$ (Lemma~\ref{th:crit}), on which 
the zeros lie in the $N\to \infty$ limit, and this curve splits the whole space in two open regions $\cL_\ga$ and $\cD_\ga$ that are natural continuation of the localized and delocalized (non critical) real regions.
A posteriori (see Section~\ref{sec:alternative}), we do verify that $h \in \cL_\ga$ if and only if $1- e^h \widehat K(z)=0$ can be solved for $\vert z \vert< 1$. 

But if we understand why $\cL_\ga$ is asymptotically zero free, 
Figure~\ref{fig:D0} is telling us that the fact that $\cL_\ga$ and $\{h:\, \Re(h)<0\}$ are zero free 
leaves open a substantial region on which the zeros may end up being. 
It turns out that  complex analytic singularity analysis is useful in this region too, 
but only under the  requirement of being  able to analytically  continue $1- e^h \widehat K(z)$ beyond the unit circle. Note that  $\widehat K(z)$ has a singularity in $1$ and that this singularity is not a pole, but this does not exclude a continuation to the centered ball of radius $R>1$ minus  $[1, \infty)$ (or minus  a proper cone containing $[1, \infty)$, see for example \cite[Ch.~VI]{cf:FS}). In this case, singularity analysis does yield, again, sharp asymptotic control on $Z_{N, h}$ that excludes that $Z_{N, h}=0$ for $N$ large. 

\smallskip

\begin{rem}
\label{rem:L2}
In the special framework \eqref{eq:Kspecial} $\cC_\ga= \partial  \cL_\ga= \partial \cD_\ga$, but we have no reason to believe that this holds in full generality.
Our arguments heavily rely on a suitable analytic continuation and the general context does not grant this, see Remark~\ref{rem:fabry}. 
Moreover we do know that in different contexts the \emph{critical region} on which the zeros accumulate is not a curve: 
see the end of Section~\ref{sec:sketch}, notably the considerations on \cite{cf:D91}. 
\end{rem}

\smallskip


Let us go more deeply into the special framework of \eqref{eq:Kspecial}. 
In this case,  by \eqref{eq:Kz}, we have 
\begin{equation}
\label{eq:z-tspec}
\widehat Z_h(z)\,=\, \frac 1{1- e^h \left(1-(1-z)^\ga \right)}\,,
\end{equation}
and we readily see that $z \mapsto1- e^h (1-(1-z)^\ga)$ can be continued to an analytic function to the whole of $\bbC$ minus a \emph{cut curve} that starts at $1$. 
Singularity analysis, once again, yields the sharp asymptotic behavior from which we conclude that also $\cD_\ga$ (and even $\cC_\ga$!) is eventually (i.e., for $N$ sufficiently large) zero free. But where are the zeros then? The point is that the results we obtain in $\cL_\ga$
and $\cD_\ga$ are uniform in $h$ bounded away from the critical curve $\cC_\ga$. This leaves the door open  to the possibility  that the zeros asymptotically accumulate on their boundary $\cC_\ga$. And, by exploiting tools from potential theory, we do prove that this  happens.

The limit of potential theory is that it yields only \emph{macroscopic results}, much in the sense that controlling the free energy yields a control on macroscopic observables. But we may be interested in sharper aspects: the crucial relevance of sharper estimates is definitely clear for $h \in \bbR$  \cite{cf:GB,cf:G}, notably (but not only) for $h$ close to the critical point, i.e. zero. And we are able to produce finer estimates precisely 
in  a complex neighborhood of the origin: for this we  exploit once again a probabilistic approach and identify the scaling behavior of $Z_{N, h}$ with $h$ that tends to zero with $N$ in a suitable way. We are thus able to understand the critical window in the complex plane (see \cite{cf:JS} for the real case). Results here are  mostly limited to $\ga=1/2$ because of  the non explicit character of the stable laws for $\ga\neq1/2$, even if we do not need to restrict to \eqref{eq:Kspecial}.  The Argument Principle is exploited, in conjunction with the scaling limits, to identify the position of the zeros. 

Corrections to the leading asymptotic locations of the zeros are obtained   in the  special framework of \eqref{eq:Kspecial} (still assuming $\ga=1/2$) and this is central for proving the results in connection with Griffiths singularities. 

\subsubsection{Organization of the rest of the paper}
In Section~\ref{sec:LDC} we prove Lemma~\ref{th:crit} and we provide alternative characterizations of $\cL_\ga$,
$\cD_\ga$ and $\cC_\ga$ that we use in the sections that follow.

Section~\ref{sec:sharp} exploits singularity analysis to obtain the sharp behavior of $(Z_{N, h})$. 
Proposition~\ref{th:Reh<0} is the only result in this section that does not rely  on singularity analysis and 
Proposition~\ref{th:Reh<0} plus Proposition~\ref{th:sharp-Reh>0} provide a full proof of Proposition~\ref{th:res0}.
The rest of the Section is devoted to the proof of Theorem~\ref{th:fe}: in fact, much more precise results are proven, see notably Proposition~\ref{th:Zsharp}.  Theorem~\ref{th:sharp-loc}  is also a direct consequence of Proposition~\ref{th:Zsharp}.  

Section~\ref{sec:potential} is devoted to the potential theory analysis. 
Theorem~\ref{th:empprobab} is a direct corollary of Proposition~\ref{th:mu}, but several other estimates of independent interest, notably about the limit density of the zeros on $\cC_\ga$, are given.

Section~\ref{sec:near0} is devoted to the precise analysis of the zeros close to the origin. One finds here a proof of Theorem~\ref{th:closest}, which follows from the general result in Proposition~\ref{th:zeros0} and the $\ga=1/2$ control on the scaling limit of Lemma~\ref{th:zeta_1}. This section contains also the much sharper estimate of Proposition~\ref{th:sharper0} which demands hypothesis \eqref{eq:Kspecial} and is crucial for the Griffiths singularity analysis. 

The Griffiths singularity analysis, with the proof of Theorem~\ref{th:griffiths}, is in Section~\ref{sec:G}.

\subsubsection{Perspectives and open problems}
The following are a few aspects of the related literature and plausible future developments.
\smallskip

\begin{itemize}[leftmargin=.2 in]
\item The pinning model may be considered the easiest exactly solvable statistical mechanics model. Yet, it does not enjoy the \emph{surprisingly rigid} structure of the  
Lee-Yang Circle Theorem \cite{cf:LY}, see \cite{cf:FV} and \cite{cf:Ruelle} for many developments and references. Nonetheless, in the special framework, the zeros do (asymptotically) lie on a closed curve that is smooth (except for the corner at the real critical point), but only in the limit. There is numerical evidence, see Figure~\ref{fig:zeros-horizon}, that the zeros approach the critical curve $\cC_\ga$ from the delocalized region $\cD_\ga$ and we believe that this is within reach of our tools (but we do not develop this aspect).
 Moreover, we do have (and present)  a good control of the zeros which are at distance $O(1/ \sqrt{N})$ from the origin when $\ga=1/2$, but results appear to be much more challenging if $\ga\neq 1/2$, or even for $\ga=1/2$ but on an \emph{intermediate scale}. By intermediate scale we mean  studying the  points close to the origin, but at a distance   
much larger than $1/ \sqrt{N}$.
\item What happens in the general framework of \eqref{eq:Kgeneral}? Theorem~\ref{th:closest} does shed some light, but ultimately only for $\ga=1/2$ and, worse, only for the zeros at distance $O(1/N^{\ga})$ from the origin. 
This  suffices to exclude the validity of  the  generalization to the pinning model, stated in  \cite{cf:KM},  of the conjecture in \cite{cf:IPZ} that the zeros should approach  the real critical point close to the lines with slope $\pm \tan (\ga \pi /2)$. However, this fact should hold on intermediate scales, i.e. for zeros that are far from the origin on the scale $1/ N^\ga$, but  a distance $o(1)$ from the origin. But this is precisely the intermediate scale region on which the control is poor. 
\item In \cite{cf:D91} (see also \cite{cf:DdSI} for models on hierarchical lattices) the random energy model is analyzed and the zeros densely fill a subset of $\bbC$ with non empty interior. Can this type of phenomena happen also for pinning models? We do not know the answer, but the fact that the critical region $\cC_\ga$ is a curve is by no means granted in the general framework (see Figure~\ref{fig:lac}). 
\item Our analysis is restricted to the case of $\ga \in (0,1)$. Larger values of $\ga$ can be treated as well, at least to a certain extent, but it is lengthy and  not straightforward. In particular, the case
$\ga \in (1,2)$ (inter-arrivals with finite first moment, but infinite variance) is different from the  $\ga \ge 2$ case, for which the inter-arrival law is in the domain of attraction of the Gaussian law. 
\item It is certainly possible to give a general statement for inter-arrival laws whose characteristic function satisfies a number of hypotheses, in particular suitable continuation properties, not only for the characteristic function  but also for its inverse  (defined a priori on the positive real axis). This is rather involved and, ultimately, we can verify the conditions only for \eqref{eq:Kspecial}, at least if we want to treat every $h \in \bbC$. 
\item It is very unfortunate that we control the Griffiths singularity only for the reduced model introduced in \cite{cf:KM}. As it is claimed in \cite{cf:KM}, the result should be somewhat robust and should hold also for the original model (at least close to the limit in which the reduced model emerges). How to prove this remains a challenge. But this challenge is a special case of  the (much more) general problem of showing the existence of Griffiths singularities for non diluted models.
\item A number of works, e.g. \cite{cf:ACCMM,cf:CMM,cf:GZ}, studied the dynamical counterpart of Griffiths singularities and rather sharp results have been obtained for some diluted lattice models.  We cite also \cite{cf:subdiff} for another type of dynamical phenomenon due to rare regions of Griffiths type. 
For pinning models the dynamical analysis is up to now limited to the nondisordered case: we cite \cite{cf:CMT} that deals with the localized phase, the one relevant for the Griffiths singularity, of the pinning model, but the results in \cite{cf:CMT} are without disorder. 
\end{itemize}

\subsubsection{Recurrent notations} We use $\overline z$ for the complex conjugate of $z$,
$B_w(r):= \{z \in \bbC: \, \vert z-w \vert < r\}$ for the open ball of radius $r>0$ centered in $w \in \bbC$ and 
 Sect$(\gb):=\{z:\, \vert \arg(z)\vert < \gb\}$ for the symmetric sector centered on the positive real axis, of angle opening $2\gb$. 

\section{On the localized, delocalized and critical regions (assuming \eqref{eq:Kspecial})}
\label{sec:LDC}

$\cL_\ga$, $\cD_\ga$ and $\cC_\ga$ are defined in Lemma~\ref{th:crit}, assuming \eqref{eq:Kspecial}: in this section we work only in this restricted framework.
We start by giving  a proof of Lemma~\ref{th:crit}, so $\cL_\ga$, $\cD_\ga$ and $\cC_\ga$ are well defined. Then we give alternative characterizations 
of these three sets. 

\subsection{About the critical curve: proof of Lemma~\ref{th:crit}}
$\cC_\ga$ is just the image of $\theta \mapsto (1-\exp(-i \theta))^\ga$ under the map $z \mapsto \Log (1-z)$. 
So we start with the following result:

\medskip

\begin{lemma}
\label{th:beforeLog}
The map  $\theta \mapsto (1-\exp(-i \theta))^\ga$ draws a simple closed curve  
in $\cC$ when $\theta$ runs from $0$ to $2\pi$. This curve is invariant under complex conjugation, is  contained in the closure of Sect$(\ga \pi/2)$ and in 
 the closure of $B_0(2^\ga)$ (hence it is also contained   in the strip $\{z: 0 \le \Re(z)  \le 2^\ga$, see Fig.~\ref{fig:combined}(A)). For 
 $\theta\searrow 0$ and  $\theta \nearrow 2 \pi$ the curve is tangent to the boundary of Sect$(\ga\pi/2)$.  
Moreover, it is smooth, except at the origin. 
\end{lemma} 
\medskip

\begin{proof}
The proof follows by elementary arguments based on the fact that the curve is the map of the circle $\partial B_1(1)=\{1-\exp(-i \theta): \, \theta \in [0, 2\pi)\}$ 
under  $z \mapsto z^\ga$. 
\end{proof} 


\begin{figure}[h!]
  \centering
  \begin{subfigure}[b]{0.55 \linewidth}
    \includegraphics[width=7.7 cm]{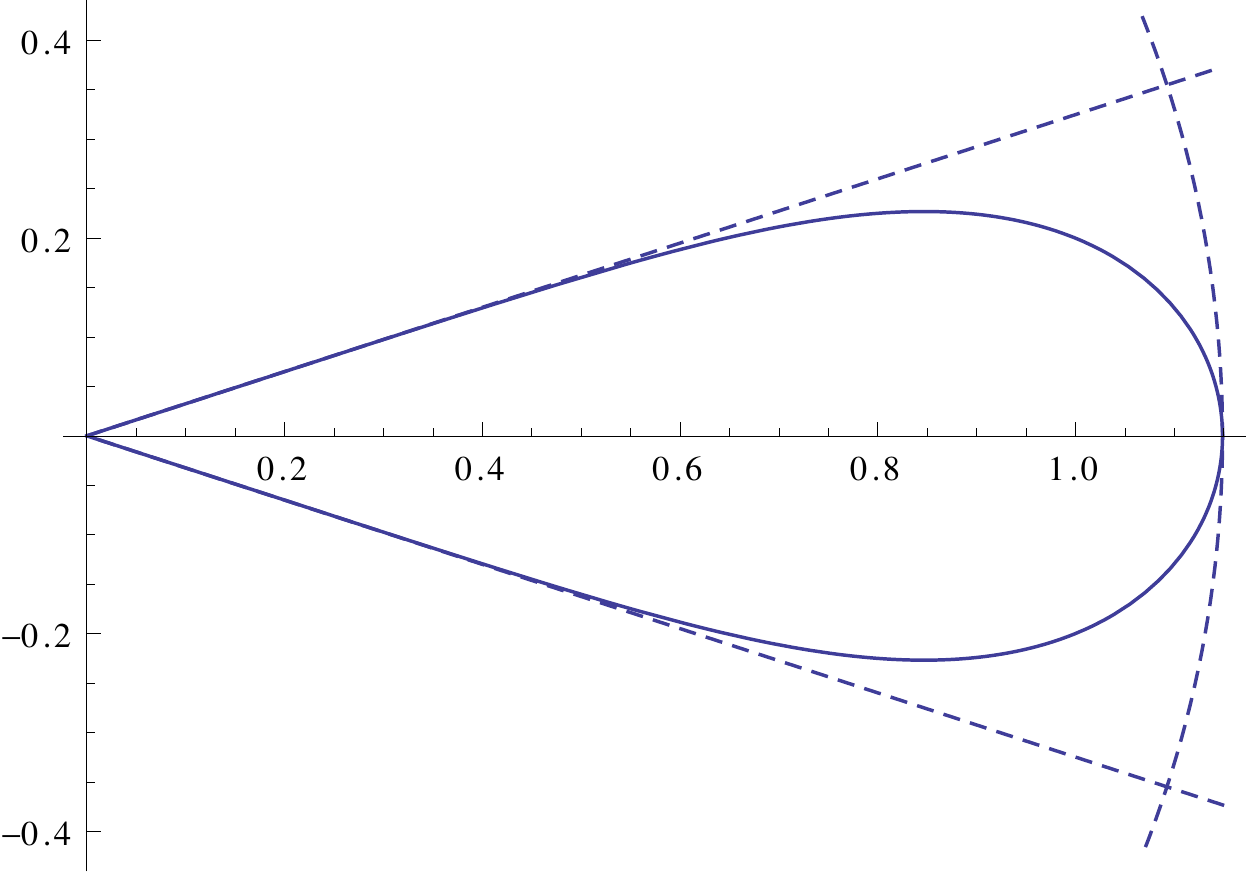}
    \caption{}
  \end{subfigure}
  \begin{subfigure}[b]{0.4\linewidth}
    \includegraphics[width=5.8 cm]{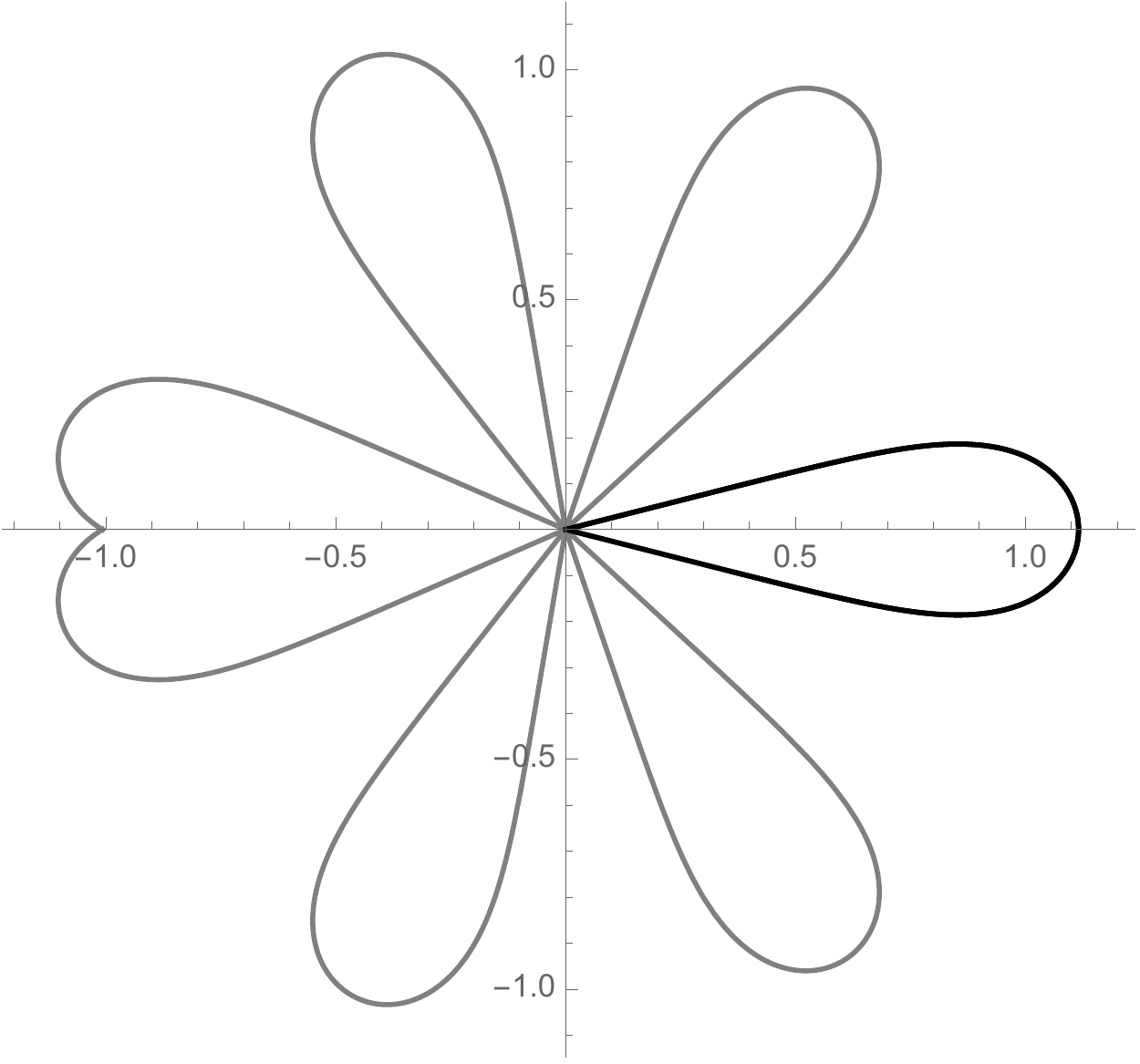}
    \caption{}
  \end{subfigure}
  \caption{In (A) the curve drawn by the map map  $\theta \mapsto (1-\exp(-i \theta))^\ga$, with $\ga=0.2$. The dashed line are the bounds given in Lemma~\ref{th:beforeLog}. In (B) we plot the set $\{\eta: \, \vert 1- \eta^{1/ \ga} \vert =1\}$ (recall that when we write $\eta^{1/ \ga}$ we mean that $\eta$ is in the domain of $z \mapsto z^{1/\ga}$, i.e. $\eta \notin (-\infty, 0]$)
   for  $\ga=1/\sqrt{40}\approx 0.1581$. }
 \label{fig:combined} 
\end{figure}

\medskip

\begin{proof}[Proof of Lemma~\ref{th:crit}]
This follows directly from Lemma~\ref{th:beforeLog} and some elementary considerations:
notably the fact that the curve  of Lemma~\ref{th:beforeLog} is in Sect$(\ga \pi /2)$  and tangent to its boundary approaching the origin says that $\cC_\ga$ does not enter 
Sect$(\ga\pi /2)$ and it is also tangent to this set approaching the origin. Moreover  
the curve  of Lemma~\ref{th:beforeLog} is in the closure of 
$B_0(2^\ga)$ (in fact,  the intersection with the boundary of $B_0(2^\ga)$ is just  the point $2^\ga$) and this 
 yields  that $\cC_\ga$ is in the strip $0 \le \Re(h) \le - \log (2^\ga -1)$ (and that the point of contact with the boundary of the strip are only $0$ and $ - \log (2^\ga -1) + i \pi$).
The curve of Lemma~\ref{th:beforeLog}
separates $\bbC$ into two connected components: the bounded one is mapped into $\cL_\ga$, and the unbounded one is mapped into $\cD_\ga$. 
\end{proof}
\medskip

\begin{rem}
\label{rem:branches}
 Figure~\ref{fig:combined}(B) identifies a phenomenon we need to watch out for: $\{(1-\exp(-i \theta))^\ga:\, \theta \in \bbR\}$ is a subset of
 $\{\eta: \, \vert 1- \eta^{1/ \ga} \vert =1\}$ and, unless $\ga \ge 2/3$, it is a proper subset. This is due to the fact that $(\eta \exp(2\pi i k \ga))^{1/\ga}= \eta^{1/\ga}$ if $\eta \exp(2\pi i k \ga)\in \bbC \setminus (-\infty,0)$. So 
 $\{\eta: \, \vert 1- \eta^{1/ \ga} \vert =1\}$ in general contains several copies of $\{(1-\exp(-i \theta))^\ga:\, \theta \in \bbR\}$ rotated by $ \exp(2\pi i k \ga)$, except that the phase $2\pi \ga + \arg(\eta)$ of the points in the rotated copies must be  in $(-\pi, \pi]$. In view of the (sharp) bounds in Lemma~\ref{th:crit} we see
 that the two sets coincide if and only if the curve for  $k=1$ has empty intersection with the upper half plane
 (equivalently,  the curve for  $k=-1$ has empty intersection with the lower half plane). This amounts to 
  $2\pi \ga - (\pi/2) \ga \ge  \pi$, i.e. $\ga \ge 2/3$.   
\end{rem}

 \begin{figure}[h]
\centering
\includegraphics[width=12 cm]{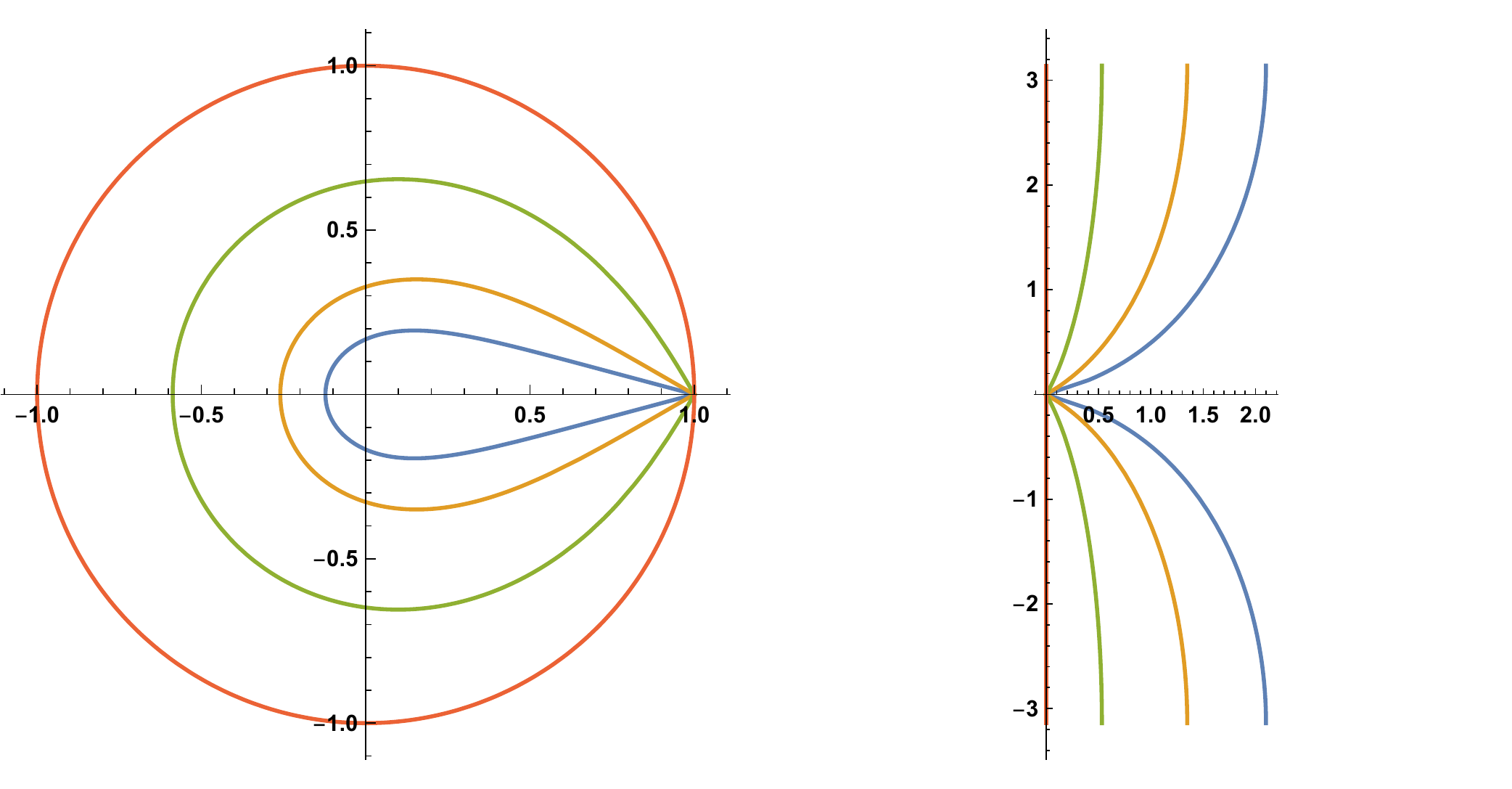}
\caption{\label{fig:curve4} 
On the left the plot of the curve $\theta \to 1-(1-\exp(-i\theta))^\ga$, for $\ga=1/6,1/3,2/3,1$ (i.e., blue, yellow, green red). On the right the logarithm of the same curve: i.e., on the right we have $\cC_\ga$ for the same values of $\ga$.
 The monotonicity of the curves on the right and (say) in the first quadrant has the simple geometric interpretation that when one goes though one of the curves with $\ga <1$ on the left, the distance of the curve
$\theta \mapsto 1-(1-\exp(-i\theta))^\ga$ to the origin decreases for $\theta$ that goes from $0$ to $\pi$. A proof of the monotonicity can be found in Lemma~\ref{th:appC}.
}
\end{figure}

\medskip

Another fact that follows directly from Lemma~\ref{th:beforeLog} is that $\cC_\ga$ can be seen as the graph of a function of the imaginary coordinate. It is actually an increasing (respectively, decreasing) function of the imaginary coordinate if the imaginary coordinate is positive (respectively, negative) as it is apparent from the curves on the right of Figure~\ref{fig:curve4}.
This can be shown by making  the parametric representation  $\cC_\ga$ explicit: with  
a rather cumbersome computation we can write $\cC_\ga$ as $\{f_1(\theta)+ i f_2(\theta): \, \theta \in [0, 2\pi)\}$
 (we set $a:=1-\ga$) with
\begin{equation}
\label{eq:f1}
f_1(\theta)= 
-\frac 12\log \left({2^{2 a} \sin ^2\left(\frac{ a (\pi -\theta)}2\right) \sin ^{2a}\left(\frac{\theta}{2}\right)+\left(1-2^{a} \sin ^{a}\left(\frac{\theta}{2}\right) \cos \left(\frac{a (\pi -\theta)}2\right)\right)^2}\right) ,
\end{equation}
and
\begin{equation}
\label{eq:f2}
f_2(\theta)= 
\arctan_0 \left(
\frac{2^{a} \sin \left(\frac{a (\pi -\theta)}2\right) \sin ^{a}\left(\frac{\theta}{2}\right)}
{
1-2^{a} \sin ^{a}\left(\frac{\theta}2\right) \cos \left(\frac{ a (\pi -\theta)}2\right)
}
\right)\,,
\end{equation}
where  $\arctan_0(\cdot): \bbR \to [0,\pi]$ is a version of $\arctan(\cdot): \bbR \to (-\pi/2, \pi/2)$ defined as 
$\arctan_0(t)=\arctan(t) $ if $t\ge 0$ and $\arctan_0(t)=\arctan(t) + \pi$ if $t<0$. 
And now it is \emph{just} a matter of showing that $f'_1(\theta)>0$ for $\theta\in (0, \pi)$.
In Lemma~\ref{th:appC} we show this along with an independent proof of $f'_2(\theta)>0$.


\subsection{Alternative characterizations of the localized and  delocalized regions}
\label{sec:alternative}
We start by defining the open set 
\begin{equation}
\label{eq:cLstar}
 \cL_\ga^\star\, :=\, \left\{ h \in \bbC:\, \Re(\tf(h))>0 \right\}\, ,
\end{equation}
and remark that $(0, \infty)\subset \cL_\ga^\star$. We then introduce two more subsets of $\bbC$:
\begin{equation}
\label{eq:L'}
\cL'_\ga\text{ is the connected component of }\cL_\ga^\star \text{ that contains }(0, \infty)\, ,
\end{equation}
and 
\begin{equation}
\label{eq:L''}
\cL''_\ga\, :=\, \left \{ h \in \bbC:\, \text{ there exists } z \in B_0(1) \text{ such that } 1-e^h(1-(1-z)^\ga)=0 \right\}\,.
\end{equation}
Let us point out from now that, for $h\in \cL''_\ga$, the solution to $1-e^h(1-(1-z)^\ga)=0$ is of course unique: 
in fact, for $z \in B_0(1) $, 
$1-e^{-h}=(1-z)^\ga$  is equivalent to $z=1-(1-e^{-h})^{1/ \ga}$. What is also straightforward is to check  that $z$ is a simple zero.

\medskip

\begin{lemma}
\label{th:3=}
$\cL_\ga= \cL'_\ga= \cL''_\ga$.
\end{lemma}

\medskip

\begin{proof}
 Since $h \mapsto 1 - e^{-h}$ and $z \mapsto 1-z$ are both one-to-one,
this is the same as asking whether $\cT_\ga= \cT'_\ga= \cT''_\ga$ with 
\begin{enumerate}
\item $\cT_\ga$ defined by considering the closed curve $\{(1-\exp(i \theta))^\ga:\, \theta \in [0, \pi)\}$ that splits $\bbC$ into two connected components:
$\cT_\ga$ is the bounded one; 
\item $\cT'_\alpha$ the connected component of $\left\{ \eta : \ \left| 1 - \eta^{1/\alpha}  \right| < 1 \right\}$ containing $(0,1)$;
\item $\cT''_\alpha := \left\{ \eta :\, \text{there exists } \zeta \text{ such that } \ |1-\zeta | < 1 \text{ and }\zeta^\alpha=\eta \right\}$.
\end{enumerate}

$\cT_\ga= \cT'_\alpha $ is a direct consequence of 
Lemma~\ref{th:beforeLog}.

Moreover we have $\cT''_\alpha\supset  \cT'_\alpha $ because if $\eta\in \cT'_\alpha $ we can set $\zeta =\eta^{1/\ga}$, which is in Sect$(\pi/2)$, hence $\zeta^\ga= \eta$, besides of course $\vert 1- \zeta\vert <1$. Therefore $\eta\in \cT''_\alpha $.

For $\cT''_\alpha\subset  \cT'_\alpha $ we start by claiming that 
$\cT''_\alpha \subset \{ \eta:\, \vert 1- \eta^{1/ \ga}\vert <1\}$. In fact if $\eta \in \cT''_\alpha$ there exists    $\zeta\in B_1(1)$ such that  $\zeta^\ga=\eta$,  so $\zeta=\eta^{1/ \ga}$.  And taking $\eta^{1/ \ga} \in B_1(1)$ yields the claim. Now we remark that Lemma~\ref{th:beforeLog} 
implies that  
\begin{equation}
\label{eq:forhr2}
\{ \eta:\, \vert 1- \eta^{1/ \ga}\vert <1\} \setminus \cT'_\alpha\subset\left\{ \eta:\, \vert \arg (\eta) \vert \ge 2  \pi \ga  - \frac \pi 2 \ga \right\}\, .
\end{equation} 
But $\vert 1- \zeta \vert < 1$ implies $\vert \arg (\zeta)\vert < \pi/2$, so $\vert \arg (\zeta^\ga ) < \ga \pi/2$.
So $\zeta^\ga $ is not contained in the set in the right-hand side of  \eqref{eq:forhr2}. 
Therefore $\{ \eta:\, \vert 1- \eta^{1/ \ga}\vert <1\} \setminus \cT'_\alpha$ and $\cT''_\alpha$ have empty intersection. Hence $\cT''_\alpha\subset  \cT'_\alpha $ and the proof is complete.
\end{proof}

\medskip

Lemma~\ref{th:3=} implies that if $h \in \cD_\ga\cup \cC_\ga$ then $1-e^h(1-(1-z)^\ga)=0$ has no solution $z\in B_0(1)$. We need to refine 
this statement:

\medskip

\begin{lemma}
\label{th:D+C}
\begin{enumerate}
\item If $h\in \cC_\ga$ then there exists a unique solution $z$ to
$1-e^h(1-(1-z)^\ga)=0$ and $\vert z \vert =1$. Moreover $z$ is a simple zero if $h\neq 0$.
\item For every $\gep>0$ there exists $r_\gep>1$ such that if $h\in \cD_\ga$ and $\dist(h, \cC_\ga) \ge \gep$ then 
$1-e^h(1-(1-z)^\ga)\neq 0$ for every $z\in B_0(r_\gep)$.
\end{enumerate}
\end{lemma}

\medskip 


\begin{proof}
For (1) we see that the equation $1-e^h(1-(1-z)^\ga)=0$, for $h= -\Log(1-(1-\exp(-i\theta))^\ga )$, reduces to $z=\exp(-i \theta)$. Moreover $\partial_z (1-e^h(1-(1-z)^\ga))= -e^h(1-z)^{\ga -1}$ is clearly non zero for $z=\exp(-i \theta)$, $\theta \in (0, 2 \pi)$, so the zero is simple. Let us remark that the problem with $h=0$, i.e. $z=1$, is that it is a singular point for $1-e^h(1-(1-z)^\ga)$.

For (2) we remark that if $r_\gep>1$ does not exist then we can find sequences $(h_j)$ and $(z_j)$ with $h_j\in \cD_\ga$, $\dist(h_j, \cC_\ga)> \gep$, $\vert z_j\vert >1$ and
$1-e^{h_j} (1-(1-z_j)^\ga)=0$ for every $j$, but $z_j \to z \in \partial B_0(1)$.   Since $1-e^{h_j} (1-(1-z_j)^\ga)=0$ and the fact that $(z_j)$ stays in a compact set, $\Re(h_j)$ is bounded below.  
Therefore there  is no loss of generality in assuming also $h_j\to h$ and of course $h$ is in $\cD_\ga$ and at distance $\gep$ or more from the boundary $\cC_\ga$. But this implies that 
$1-e^{h} (1-(1-z)^\ga)=0$, with $z \in \partial B_0(1)$, that is $h\in \cC_\ga$, which is impossible. 
So part (2) is proven.
\end{proof}


\section{Sharp estimates on the partition function}
\label{sec:sharp}

In this section we mostly exploit complex analysis tools, except for the first result (Proposition~\ref{th:Reh<0}) that is based on a more probabilistic estimate.  

\subsection{Sharp estimates in the general framework} 
In the  general context   \eqref{eq:Kgeneral}, 
for $n\to \infty$ and uniformly in $j$ such that $n/a_j\to \infty$, i.e. $j/ n^\ga\to 0$ (recall Remark~\ref{rem:a_n} for the definition of $(a_j)$), we have \cite[Th.~A]{cf:Doney}
\begin{equation}
\label{eq:Doney}
\bP\left( \tau_j=n\right) \, \sim \, j K(n)\,.
\end{equation}

\medskip

\begin{proposition}
\label{th:Reh<0} 
In the general context of \eqref{eq:Kgeneral}, 
if $\Re(h)<0$ we have 
\begin{equation}
Z_{N, h} \stackrel{N \to \infty} \sim K(N) \frac{\exp \left(h \right)}{\left(1-\exp \left( h\right)\right)^2}\, ,
\end{equation} 
and this result holds uniformly if $\Re(h)$ is bounded away from $0$. 
\end{proposition}
\medskip

\begin{proof}
We write $Z_{N, h}=\sum_{j=1}^N e^{hj}\bP( \tau_j=N)$ and for $\Re(h)<0$ by \eqref{eq:Doney} we have that if we choose a decreasing sequence  $(\gamma_N)$ of positive numbers, say $\gamma_N:= 1/ \log(N)$, then there exists  $(\gep_N)$, $\gep_N \searrow 0$, such that for $N$ sufficiently large
\begin{multline} 
\left \vert Z_{N, h}- K(N) \sum_{j=1}^N e^{hj} j \right \vert \, \le \, 
\sum_{j=1}^N \exp(-j \vert\Re(h)\vert) \left \vert \bP\left( \tau_j=N\right) -j K(N) \right \vert
\\
\le \, K(N)  \sum_{j\le N^\ga \gamma_N}  j
\exp(-j \vert\Re(h)\vert)  \left \vert \frac{\bP\left( \tau_j=N\right)}{jK(N)}-1 \right \vert 
  + 2\sum_{j> N^\ga \gamma_N} j \exp(-j \vert\Re(h)\vert)
  \\
  \le \, 
  K(N) \gep_N   \frac{\exp \left( -\vert \Re(h) \vert\right)}{\left(1-\exp \left( -\vert \Re(h) \vert\right)\right)^2}
+ 3 N^\ga \gamma_N \vert\Re(h)\vert \exp(-N^\ga \gamma_N \vert\Re(h)\vert) \, ,
\end{multline}
where from the second to the third line we have used $jK(N)=O(1/N^\ga)=o(1)$.
Therefore if $\Re(h) \le -\gep$ for an $\gep>0$, there exists $c_\gep>0$ such that
\begin{equation}
\label{eq:forrhs56}
\left \vert Z_{N, h}- K(N) \frac{\exp \left(h \right)}{\left(1-\exp \left( h\right)\right)^2}\right \vert \, \le \, c_\gep
K(N) \gep_N   {\exp \left( -\vert \Re(h) \vert\right)}\, ,
\end{equation}
and this is the uniform estimates we claimed:
since ${\exp \left(h \right)}/{\left(1-\exp \left( h\right)\right)^2} \sim \exp \left(h \right)$
for $\Re(h) \to -\infty$, for every $c>0$ the ratio between  the error term (i.e., the right-hand side of \eqref{eq:forrhs56}) 
and the leading behavior of $Z_{N, h}$, i.e. $K(N){\exp \left(h \right)}/{\left(1-\exp \left( h\right)\right)^2}$,
is $O(\gep_N)$ uniformly in $h$ such that $\Re(h) \le -c$.
\end{proof}

\medskip

\begin{proposition}
\label{th:sharp-Reh>0}
We fix  $K(\cdot)$ which satisfies  \eqref{eq:Kgeneral} 
and  consider $W\subset \bbC$ which is the union of
\begin{enumerate}
  \item the half plane  with $\Re(h) \ge C>0$;
  \item the set of $h$'s with $\Re (h)\ge a>0$ and $\vert \Im(h)\vert < \gep$.  
  \end{enumerate} 
If we choose  $C$  suitably large  and $\gep$ suitably small ($\gep $ depends on $a$, $C$ does not)  then for every $h\in W$
   there exists a unique solution $z=z_h\in B_0(1)$ to  $\widehat K(z) =\exp(-h)$  with minimal absolute value
  and such that,   
  uniformly in $h \in W$, we have 
\begin{equation}
\label{eq:leading-1}
Z_{N, h} \stackrel{N \to \infty} \sim  \frac{(1-\exp(-h))^{(1-\ga)/\ga}}{\ga \exp(h)z_{h} }  z_{ h}^{-N}\,.
\end{equation}
\end{proposition}

\medskip

\begin{proof} 
We proceed by obtaining the sharp asymptotic behavior of $Z_{N, h}$ for $N\to \infty$ and uniformly in $h$ in appropriate subsets of $\bbC$.
We will be in the case in which we can identify $r\in (0,1)$ such that
$\widehat Z_{N, h}$ has only one pole, a single pole that we call $z_h\in B_0(r)$ and no pole on $\partial B_0(r)$, so we have  (we recall that $\widehat K(z)$ is defined in \eqref{eq:Kz} and  $\widehat Z_{h}(z)$ in \eqref{eq:z-t})
\begin{equation}
  \label{eq:Taylor-1}
  \widehat Z_{h}(z)\, =\, - \frac{\exp(-h)}{\widehat K ' (z_h)(z-z_h)} + R_{h}(z)\, ,
 \end{equation}  
 which defines  $R_{h}(z)$. Therefore $z \mapsto R_{h}(z)$ is analytic in a neighborhood of the closure of $B_0(r)$ and 
\begin{equation}
 \label{eq:residue-N}
 Z_{N, h} \,=\,  \frac{\exp(-h)}{\widehat K ' (z_h)} z_h^{-N-1}+ \frac 1{2\pi i} \oint \frac{ R_{h}(z)}{z^{N+1}} \dd z\, ,
 \end{equation}
 with $z$ running, counterclockwise,  on $\partial B_0(r)$. 
 
 We treat the two regions separately.
\smallskip

For case (1) we observe that for $z$ small $\widehat{K}(z)\sim K(1)z$ and $\widehat{K}'(0)=K(1)>0$.
This entails that there exists $r>0$ such that $\widehat{K}'(z)\neq 0$ for every $z \in B_0(r)$ and 
$\widehat K: B_0(r) \to \widehat K(B_0(r))$ is invertible. Of course  $ \widehat K(B_0(r))$ is a neighborhood of the origin.
Therefore  there exists $h_0>0$ such that $\exp(-h)\in  \widehat K(B_0(r))$ for $\Re(h)>h_0$ and, for such values of $h$, $ z_h=\widehat K^{-1}(\exp(-h))$  
 is the unique solution  $B_0(r)$  of  $1- e^h \widehat{K}(z)=0$. 
 Possibly by replacing $r$ by a smaller value, we can assume also that $\widehat{K}(z) \neq 0$ for every $z$ with
 $\vert z \vert =r$, so $\inf_{z\in \partial B_0(r)} \vert \widehat{K}(z)\vert \ge c_r>0$. Hence, always for  $z\in \partial B_0(r)$, we have 
 $\vert 1- \exp(h) \widehat{K}(z)\vert \ge \exp( \Re(h)) c_r/2$ for $\Re(h) \ge h_0>0$, with suitable choice of $h_0$.
 Note that $z_h \sim \exp(-h)/K(1)$ for $\Re (h)\to \infty$ so,  in particular,  
 $\widehat K ' (z_h) \sim K(0)$. Therefore, possibly by choosing $h_0$ smaller, we have $\vert z - z_h \vert \ge r/2$ and 
 \begin{equation}
\sup_z \vert R_{h}(z) \vert \, \le \, \sup_z \left\vert  \widehat Z_{h}(z) \right \vert 
+    \sup_z\frac{ e^{-\Re(h)} }{\vert \widehat K ' (z_h)\vert \,\vert z-z_h\vert } 
\, \le\, C_r e^{-\Re(h)} \, ,
 \end{equation}
 where  the $z$ runs in  $\partial B_0(r)$ and $C_r$ can be chosen equal to $2(1/c_r+ 1/(rK(0))$.
 So, by using \eqref{eq:residue-N} we obtain
 that  for every $h$ with $\Re(h) \ge h_0$ we have 
 \begin{equation}
 \left \vert  Z_{N, h}\, - \frac{\exp(-h)}{\widehat K ' (z_h)} z_h^{-N-1} \right \vert \, \le \, C_r e^{-\Re(h)} r^{-N}\, .
 \end{equation} 
 Such a uniform estimate guarantees that there exists $h_0$ and $N_0$ such that, if $\Re(h) \ge h_0$,  $Z_{N, h}\neq 0$ for $N\ge N_0$.
 
 \smallskip

For case (2) we start by recalling  that $\widehat{K}' (z)>0$  for $z \in (0,1)$, so $\widehat{K}(z)=\exp(-h)$ has a unique (positive) solution $z=z_{h}$ for $h>0$. This may not be the unique solution in $\bbC$, but if there is another one, call it $w_{h}\in \bbC\setminus (0, \infty)$, then
$\vert w_{h} \vert >  z_{h}$. In  fact if $\vert w_{h} \vert <  z_{h}$ then $\exp(-h)= \widehat{K}(w_{h} )\le 
\widehat{K}(\vert w_{h} \vert)< \widehat{K}(z_{h})=\exp(-h)$, which is impossible. And $\vert w_{h} \vert =  z_{h}$ is excluded by aperiodicity of $K(\cdot)$. Another immediate fact is that $\widehat{K}(z)=\exp(-h)$ can be solved for $h$ in a neighborhood of the real axis and $z=z_h$ which is also in a neighborhood of the real axis: in fact, by the analytic inverse function theorem, $z_h$ is an analytic function on   $B_{a,b}(\gep):=\{z \in \bbC: 
\inf_{x\in (a,b)} \vert z -x\vert < \gep\}$,  with $0<a <b < \infty$ and $\gep>0$ sufficiently small.  
Let us fix a $\gep>0$ such  we have that $\vert w_h\vert > \vert z_h\vert$ for every $h\in B_{a,b}(2\gep)$ and such that $z_h$ is analytic in $B_{a,b}(2\gep)$. 
 We aim at showing that  there exists $\gd >0$ such that, if there exists $w_h \neq z_h$ such that 
$\widehat{K}(w_{h} )= \exp(-h)$ for $h \in B_{a,b}(\gep)$, then $\vert w_h \vert \ge \vert z_h\vert + \gd$. 
The proof is by contradiction: if this is  false, then we can find $(h_j)$ in  $B_{a,b}(\gep)$  for which $w_{h_j}$ exists for every $j$ (so  $\vert w_{h_j}\vert > \vert z_{h_j} \vert$) and $\vert w_{h_j}\vert - \vert z_{h_j} \vert \to 0$.
Without loss of generality we can assume that these three sequences converge (for the limits we just omit the subscript). By passing to the limit we see that $\vert w_h\vert = \vert z_h\vert$ for $h$ which is 
in the closure of  $B_{a,b}(\gep)$, hence in $B_{a,b}(2\gep)$, which is impossible. 

The proof now proceeds in the same way as case (1) and the final result is that for $0<a<b<\infty$ there exists $\gep>0$ and a two positive constants $c$ and $C$ such that
\begin{equation}
 \left \vert  Z_{N, h}- \frac{\exp(-h)}{\widehat K ' (z_h)} z_h^{-N-1} \right \vert \, \le \, C  \left((1+c)\vert z_h\vert \right)^{-N-1}\, ,
 \end{equation} 
for every $h \in B_{a,b}(\gep)$. This is of course sufficient to cover case (2) in view of case (1).
\end{proof}

\subsection{Sharp estimates in the special framework}

\begin{proposition}
\label{th:Zsharp}
With the inter-arrival distribution \eqref{eq:Kspecial}
\begin{enumerate}
\item if  $h \in \cL_\ga$ then $z \mapsto 1-\exp(h)(1-(1-z)^\ga)$ has a unique zero $z_{\ga, h}$ in the open unit disk around the origin and 
\begin{equation}
\label{eq:ZsharpL}
Z_{N, h} \stackrel{N \to \infty} \sim  \frac{(1-\exp(-h))^{(1-\ga)/\ga}}{\ga \exp(h) }  z_{\ga, h}^{-(N+1)}\,,
\end{equation}
and this result is uniform for $h$ bounded away from $\cC_\ga$;
\item if $h \in \cD_\ga$
\begin{equation}
\label{eq:ZsharpD}
  Z_{N, h} \stackrel{N \to \infty} \sim  K(N)\frac{e^h }{\left(1-e^h\right)^2}
  \,,
  \end{equation}
  and  also this result is uniform for $h$ bounded away from $\cC_\ga$;
\end{enumerate}
\end{proposition}

\medskip

\begin{proof}
For case (1) 
let us first remark that, by Proposition~\ref{th:sharp-Reh>0}, it suffices to show the result for $\Re(h) \le C$.
So we focus on the compact set $V_\gep:= \{h \in \cL_\ga:\, \dist (h, \cC_\ga)\ge \gep$ and $\Re(h) \le C\}
\subset \bbC_{2\pi}$.
The denominator in \eqref{eq:z-tspec}, that we call $D(h,z)$ in this proof,  for every  $h \in \cL_\ga$ has a unique zero  $z=z_h \in B_0(1)$. Consider now $r_\gep:= \sup_{h \in V_\gep} \vert z_h \vert<1$ and choose $\eta>0$
such that $\exp(-\eta) > r_\gep$.  
We can now use the same argument as in Proposition~\ref{th:sharp-Reh>0}: 
we apply  \eqref{eq:Taylor-1} and \eqref{eq:residue-N} with $r=\exp(-\eta)$. Since $V_\gep$ is compact we readily see that 
  there exists $c_ {\gep, \eta}>0$ such that $\vert R_{ h}(z)\vert \le c_ {\gep, \eta}$ for every $h\in V_\gep$ and
  $\vert z \vert =\exp(-\eta)$. This directly yields
 \begin{equation}
 \left \vert  Z_{N, h}\, - \frac{\exp(-h)}{\widehat K ' (z_h)} z_h^{-N-1} \right \vert \, \le \, c_ {\gep, \eta}\exp(\eta N)\, .
 \end{equation} 
 This completes the proof in case (1).
 
 \smallskip
 
 For case (2) we follow closely  \cite[Section~VI.3]{cf:FS}, also from the notational viewpoint. In particular, we use the fact
 that we can find $R>1$ such that  $\widehat Z_h(\cdot)$ is analytic in the open domain
 \begin{equation}
 \label{eq:3.14}
 \gD (\phi, R)\, :=\, \left \{z: \, \vert z \vert < R,\, z \neq 1, \vert \arg(z-1) \vert > \phi\right\}\, , 
 \end{equation}
 and this for any choice of $\phi\in (0, \pi/2)$. 
This follows from Lemma~\ref{th:D+C}(1) and we can (and do) choose $R=R_\gep=1+(r_\gep-1)/2$.   
This directly yields that $\sup_h \sup_{z: \vert z \vert =R}\widehat Z_h(z)< \infty$ where $h\in K_\gep$ with $K_\gep$ a compact subset with two requirements:  $K_\gep \subset \cD_\ga$ 
and  $\dist (K_\gep,\cC_\ga) \ge \gep$. 
Moreover
 \begin{equation}
 \label{eq:expand9}
 \widehat Z_h(z)\, =\, 
\frac{1}{1-e^h (1-(1-z)^\ga)}\, =\,  \frac{1}{1-e^h}-
 \frac{e^h (1-z)^\ga}{\left(1-e^h\right)^2} + O \left( \vert 1-z\vert ^{2\ga} \right)\, ,
 \end{equation} 
 uniformly for  $z$ in  the intersection of $\gD (\phi, R)$ with a neighborhood of $1$ (since it suffices to show the result for $\gep$ small, 
 A direct application of  \cite[Th.~VI.3]{cf:FS} (see also \cite[Th.~VI.1]{cf:FS} for the details on how to extract the leading order term from the $(1-z)^\ga$ term in \eqref{eq:expand9}) yields
  \begin{equation}
 \label{eq:forZsharpD}
  Z_{N, h} \stackrel{N \to \infty} =  \frac{e^h }{(-\Gamma(-\ga))\left(1-e^h\right)^2}\frac 1{N^{1+ \ga}}
  +O \left( \frac 1{N^{1+ 2\ga}}\right)\,.
  \end{equation}
 The proof of  \cite[Th.~VI.3]{cf:FS} (pp.131-132) is based on estimates on a contour integral involving $\widehat Z_h(z)$ and the uniform control we have just claimed on $\widehat Z_h(z)$ yields that \eqref{eq:forZsharpD} holds uniformly in $h\in K_\gep$.
 \end{proof}
 
 \smallskip
 
 \begin{rem}
 The argument we just presented can be upgraded to deal with $K_\gep$ non compact (still satisfying the two requirements of being a subset of $\cD_\ga$ which is bounded away from the boundary $\cC_\ga$): it is a matter of following carefully what happens for $-\Re(h)$ large. This provides an alternative argument for Proposition~\ref{th:Reh<0}.
 \end{rem}

\medskip

Theorem~\ref{th:fe} is a  corollary of the sharp estimates we just established, except for the critical case.

\medskip

\begin{proof}[Proof of Theorem~\ref{th:fe}]
For $h\in \cL_\ga$ we use Proposition~\ref{th:Zsharp}(1): $\lim_N (1/N) \log \vert Z_{N, h} \vert = -\log \vert z_{\ga, h} \vert$ and, by making $z_{\ga, h}$ explicit, we see that it is equal to $\Re (\tf(h))$ in the whole of $\cL_\ga$.  
The case of  $h\in \cD_\ga$ is even more straightforward and uses Proposition~\ref{th:Zsharp}(2). And in both cases   
uniformity follows because  Proposition~\ref{th:Zsharp} is proven  uniformly, away from $\cC_\ga$.

In the critical case, by Lemma~\ref{th:D+C},  there is a simple pole $z_{\ga, h}$  on the unit ball (except for $h=0$: in this case $Z_{N,0}\sim c_\ga N^{1-\ga}$ for an explicit $c_\ga>0$, see \cite[Th.~B]{cf:Doney}
so $\lim_N (1/N) \log \vert Z_{N, 0}\vert =0$, even if this last result is easily established without sharp control on $Z_{N, 0}$). Of course there is no pole in the unit circle, because that happens only for $h \in \cL_\ga$ (cf. Lemma~\ref{th:3=}). One can check that there is no other pole, but this is not very important because it is obvious that there is no other pole in the closure of $B_0(r)$ for some $r>1$. This allows to choose a circuit of integration 
that coincides with $\partial B_0(r)$, except close to $1$, where we have to use a circuit like the one in the proof of Proposition~\ref{th:Zsharp}(2), see notably \eqref{eq:3.14}. Therefore the 
 sharp asymptotic behavior of $Z_{N, h}$ in this case will be given by the dominant contributions among the pole
 \eqref{eq:ZsharpL} and the essential singularity in 1 \eqref{eq:ZsharpD}. Actually, since $\vert z_{\ga,h}\vert=1$, the pole in this case just contributes an additive constant times an $N$ dependent phase, while the essential singularity contributes a vanishing term (and eventual poles outside the unit circles would just contribute exponentially vanishing terms and 
 the higher order contribution of the essential singularities would be more dominant anyways). If we sum up:
for $h\in \cC_\ga \setminus \{0\}$ we have (recall \eqref{eq:forZsharpD})
 \begin{equation}
\label{eq:ZsharpC}
Z_{N, h}\, =\,  \frac{(1-\exp(-h))^{(1-\ga)/\ga}}{\ga \exp(h) }  e^{-(N+1)\arg (z_{\ga, h}) i }
+   \frac{e^h }{(-\Gamma(-\ga))\left(1-e^h\right)^2}\frac 1{N^{1+ \ga}}
  +O \left( \frac 1{N^{1+ 2\ga}}\right)\,.
\end{equation}
 Therefore $\lim_N (1/N) \log \vert Z_{N, h}\vert =0$ for every $h \in \cC_\ga$.
\end{proof}

\section{Potential theory and empirical measure analysis}
\label{sec:potential}

While the set up of this section is general, all the results in the end depend on the control of $\vert Z_{N, h} \vert$ for every complex $h$ (outside of the critical curve). They are therefore limited to the special framework \eqref{eq:Kspecial}.

In this section we begin by using  the polynomial notation \eqref{eq:PN} and we point out that, since $\bP(\tau_N=N)=K(1)^N$, we have
\begin{equation} 
\label{eq:PN2}
P_N(w)\, =\, K(1)^N w \prod_{j=1}^{N-1} \left(w-w_{N,j}\right)\, ,
\end{equation}
with $(w_{N,j})_{j=1, \ldots, N}$ the zeros of $P_N(\cdot)$ and $w_{N,N}=0$. 
With ${\mathtt x}, {\mathtt y} \in \bbR$ we set
\begin{equation}
u_N({ {\mathtt x}},{{\mathtt y}})\,:=\, \frac 1 {N-1} \Re \log P_N ({ {\mathtt x}}+i{{\mathtt y}})\,=\,  \frac 1 {N-1}  \log \vert P_N ({ {\mathtt x}}+i{{\mathtt y}})\vert\,  .
\end{equation} 
With abuse of notation we write  $u_N(w)$ also for $u_N({ {\mathtt x}},{{\mathtt y}})$ when $w={\mathtt x}+i{\mathtt y}$, in fact we identify $w$ with $({\mathtt x},{\mathtt y})$.
Note that  $u_N(\cdot)$ is smooth out of the zeros  
and \cite[Th.~3.7.8]{cf:Ransford} 
\begin{equation}
\label{eq:subharm1}
\Delta u_N\, =\, 2\pi \frac 1{N-1}\sum_{j=1}^{N-1} \gd_{w_{N,j}}+ \frac{2 \pi }{N-1} \gd_0\, =:\, 2\pi \nu_N + \frac{2 \pi }{N-1} \gd_0 \, ,
\end{equation}
and this  means that for every $g\in C^\infty_0$ (i.e.,  $g\in C^\infty$ and $g$ is compactly supported) 
we have 
\begin{equation}
\label{eq:subharm2}
\int_{\bbR^2} u_N \Delta g \dd \gl \, =\, 2\pi \int_{\bbR^2} g \dd \nu_N + \frac{2 \pi }{N-1} g(0)\, ,
\end{equation}
where $\gl$ is the Lebesgue measure on $\bbR^2$.

We now go back to our original coordinate systems. 
We have $u_N(\exp(h))= \tf_N(h)$ and, with $h=x+iy$ identified with $(x,y)$ 
\begin{equation}
\label{eq:coc1}
\Delta \tf _N \,=\, 2\pi \frac 1{N-1} \sum_{j=1}^{N-1} 
 \gd_{h_{N,j}} =: 2\pi \mu_N \, ,
\end{equation}
where  $h_{N,j}$ is one of the $N-1$ zeros of $h \mapsto Z_{N,h}$ with $\Im(h)\in (-\pi, \pi]$: 
we can choose them so that $w_{N, j} =\exp(h_{N,j})$ and 
$w_{N, N}=0$ is \emph{pushed to $-\infty$} in these variables.
Equation \eqref{eq:coc1} means 
\begin{equation}
\label{eq:coc2}
\int_{\bbR^2} \tf_N(x,y) \Delta f(x,y) \dd x \dd y \, =\, 2\pi \int_{\bbR^2} 
f(x,y) \dd \mu_N ( x , y)\, ,
\end{equation}
for every $f: \bbC_{2\pi} \to \bbR$ which is smooth and compactly supported.
Stepping from \eqref{eq:subharm1} to \eqref{eq:coc1}, i.e. from \eqref{eq:subharm2} to \eqref{eq:coc2}, is a computation, but let us note that to obtain \eqref{eq:coc2} it suffices to consider  \eqref{eq:subharm2} with $g\in C^\infty_0$ and whose support is bounded away from $0$. In fact, with the change of variable $w=\exp(h)$, $h$ in a compact subset of $\bbC_{2\pi}$ means $w$ lives in a compact subset of $\bbC\setminus\{0\}$. 
The computation can be performed in $\bbR^2$, by this we mean that if $w=u+iv$ and $h=x+iy$, the change of coordinates is $u=e^x \cos (y)$ and $v=e^x \sin (y)$. The determinant of the Jacobian of this transformation is $e^{2x}$. On the other hand, $\Delta f(x,y)= (\partial_x^2+ \partial_y^2)g(e^x \cos (y),   e^x \sin (y))= e^{2x} (\Delta g) (e^x \cos (y),   e^x \sin (y))
$, so  \eqref{eq:coc2} follows.

\medskip

\begin{lemma} 
\label{th:Ffmu}
Assume  \eqref{eq:Kspecial}.
We have 
\begin{equation}
\lim_{N \to \infty}
\int_{\bbR^2} \tf_N(x,y) \Delta f(x,y) \dd x \dd y\, =\, \int_{\bbR^2} \tf(x,y) \Delta f(x,y) \dd x \dd y\, .
\end{equation}
\end{lemma}
\medskip

\begin{proof}
By the uniform convergence of $\tf_N(h)$ away from the simple (smooth, finite length) curve $\cC_\ga$, see Lemma~\ref{th:crit} and Theorem~\ref{th:fe}, and because $\tf(\cdot)$ is continuous,  it suffices to show that, with $A_\gep:= \{(x,y):\, x+i y \in  \bbC_{2\pi}$ and $\dist(x+i y, \cC_\ga) < \gep\}$, we have 
\begin{equation}
\label{eq:tbp3-1}
\lim_{\gep \searrow 0}\sup_N \left \vert
\int_{A_\gep} \tf_N(x,y) \Delta f(x,y) \dd x \dd y\right\vert\, =\, 0\,.
\end{equation}
Now let us point out (Proposition~\ref{th:res0}) that there exists $C>1$ such that, uniformly in $N$, all the zeros are in the compact set $K:=\{h:\, \vert \Re(h)\vert \le C\}$  except of course for the zero at $\infty$ which however gives a contribution $x/N$ to $\tf_N(x,y)$ (without loss of generality we assume $A_\gep \subset K$). We write
\begin{equation}
\tf_N(x,y)\,=\, \log K(1) +\frac x N + \frac 1N \sum_{j=1}^{N-1} \log \vert \exp(x+i y)-\exp(h_{N,j}) \vert\, ,
\end{equation}
and focus on  $\log \vert \exp(x+i y)-\exp(h_{N,j}) \vert$, which is of course bounded above by $c_1:=\log 2+ \log C$. On the other hand $\vert \exp(h)-\exp(h_0)\vert \ge \gep /2$, if $h, h_0 \in K$ and  $\vert h-h_0\vert \ge \gep$ (for $\gep$ sufficiently small). Therefore for every $N$, for every $j=1, \ldots, N-1$ and every 
$(x,y) \in K$ with $\vert x+iy-  h_{N,j}\vert \ge \gep$ we have 
\begin{equation}
\left \vert \log \vert \exp(x+i y)-\exp(h_{N,j}) \vert \right \vert \, \le \, c + \vert \log \gep \vert\,, 
\end{equation}
with $c=c_1+ \log 2$
On the other hand $\int_{B_0(\gep)} \log \vert \exp(x+iy)-1\vert \le  4 \gep^2 \vert \log \gep\vert$ for $\gep$ small.
By putting these estimates together we see that 
\begin{equation}
 \left \vert
\int_{A_\gep} \tf_N(x,y) \Delta f(x,y) \dd x \dd y\right\vert\, \le \, 
   \left(\vert\log K(1)\vert+(C/N)+c + \vert \log \gep \vert\right) \gl\left( A_\gep\right) 
+ 4 \gep^2 \vert \log \gep\vert
\,.
\end{equation}
Since $\gl\left( A_\gep\right) = O(\gep)$, \eqref{eq:tbp3-1} is proven.
\end{proof}

\medskip

We can now conclude that

\medskip

\begin{proposition}
\label{th:mu}
Assume  \eqref{eq:Kspecial}. We have that
$(\mu_N)$ converges to $\mu$ in distribution, i.e. $\lim_N \int f \dd \mu_N = \int f \dd \mu$
for every $f$ continuous and bounded: in particular, $\mu$ is a probability. Moreover $\mu$ is identified by
\begin{equation}
\label{eq:mu}
\int_{\bbR^2} \tf(x,y) \Delta f(x,y) \dd x \dd y\, =\, 2 \pi \int_{\bbR^2}  f(x,y) \dd \mu( x,  y)\
\ \ 
\text{ for every } f\in C_0^\infty\,.
\end{equation}
\end{proposition}
\medskip

\begin{proof} 
Proposition~\ref{th:res0} confines all zeros to a compact set, so $(\mu_N)$ is tight.
The limit points $\mu$ are therefore probabilities and 
 Lemma~\ref{th:Ffmu} guarantees that they satisfy \eqref{eq:mu}. But \eqref{eq:mu} uniquely identifies 
 $\mu$ and we obtained the desired convergence. 
 \end{proof}

\smallskip 

 One can extract a number of facts from Proposition~\ref{th:mu}: we list some of them here, in an informal way.

\begin{itemize}[leftmargin=.2 in]
\item
Proposition~\ref{th:mu} directly yields that the support of $\mu$ is contained in $\cC_\ga$  because if $\tf$ is $C^2$ 
in the open ball $B_{(x,y)}(r)$, then $\int f \dd \mu=0$ for every $f$ supported in $B_{(x,y)}(r)$, and $\tf$ is smooth outside of $\cC_\ga$
(that the support of $\mu$ is contained in $\cC_\ga$  can also be seen from
Proposition~\ref{th:Zsharp}). But in fact the support is exactly $\cC_\ga$ and one can show that $\mu$ has a density on its support (for example, by taking as reference measure the arc-length on $\cC_\ga$) and this density vanishes only at zero, which is  the only singular point of the density. 
Establishing this is a bit cumbersome:
it involves  performing  integration by parts on the left-hand side of \eqref{eq:mu} exploiting the monotonicity of the critical curve proven in Appendix~\ref{sec:monotone}.
The result is  
\begin{equation}
\label{eq:after1ibp}
\sum_{\gs=\pm 1} \gs\int_0^{x_0} \left(\partial_1 \tf(x, \gs Y(x))  Y'(x) - \partial_2 \tf(x, \gs Y(x))  \right) f(x, \gs Y(x))\dd x\, ,
\end{equation}
where $x_0=-\log(2^\ga -1)$ and the function $x\mapsto Y(x)$, with $x\in [0, X_0]$, is so that  
the union of $\{(x, Y(x)):\, x\in [0, X_0]\}$ and $\{(x, -Y(x)):\, x\in [0, X_0]\}$ yields $C_\ga$.
From \eqref{eq:after1ibp} one reads that the probability $\mu$ has a density on its support. 
\item
One computation that can be performed in detail with moderate effort is the one that leads to the density of $\mu$ near $0$. We call $s$ the arc length computed (with sign) starting from the origin: let us restrict to the portion of $\cC_\ga$ with positive imaginary part and let 
$\gamma_\ga (s)$,  $s$ from zero to the total length $\ell_\ga$ of the curve (with positive imaginary part), so $\gamma _\ga(0)=0$ and $\gamma_\ga(\ell_\ga)= -\log (2^\ga-1) + i \pi$. Then 
\begin{equation}
\label{eq:density-s}
\frac{\dd \mu ([0, \gamma(s)])} {\dd s} \stackrel{s \searrow 0} \sim \frac{s^{(1-\ga)/ \ga}}{\ga \cos\left( \ga \frac \pi 2 \right)}\,.
\end{equation}
\item One can work out more explicitly the case $\ga=1/2$: it is more practical to extract the density as a function of 
 $x\in \left[0, \log\left(1+ \sqrt{2}\right)\right)$ and we obtain 
\begin{equation}
\frac{8 e^x \sinh (x)}{\sqrt{6 e^{2 x}-e^{4 x}-1}}\stackrel{x\searrow 0}=
4\,  x+\frac 8 3 x^3 +O\left(x^5\right)
\, .
\end{equation}  
This density is convex and diverges approaching $\log\left(1+ \sqrt{2}\right)$, but this is an artefact of the parametrization:
in the arc length parametrisation, that is if we divide by $(1+(Y'(x))^2)^{1/2}$,  we obtain the particularly simple formula
\begin{equation}
\sqrt{2} \left(1-e^{-2 s}\right)\stackrel{s\searrow 0} = 2 \sqrt{2} \, s-2 \sqrt{2}\,  s^2 + O\left(s^3\right)\, ,
\end{equation}
which is a concave bounded function and which is, of course, in agreement with, \eqref{eq:density-s}.
\item A byproduct of the analysis developed in this section is the formula 
\begin{equation}
\label{eq:FreprC}
 \Re\left( \tf(h)\right) \, 
 =\,\log \ga  + \int \log {\vert e^h- e^{\zeta}  \vert} \mu ( \dd \zeta) 
 \, =\, \int \log \frac{\vert e^h- e^{\zeta}  \vert}{\vert 1- e^{\zeta}  \vert} \mu ( \dd \zeta)\, ,
\end{equation}
 where the first expression follows 
 from from \eqref{eq:PN2} (or the middle term in \eqref{eq:PN2.1}, which is just \eqref{eq:PN2} with $w=e^h$), Proposition~\ref{th:mu} and the fact that the zeros are bound to a compact region (Proposition~\ref{th:res0}).  For the second expression it suffices to use 
 the rightmost term in \eqref{eq:PN2.1} and $\log Z_{N,0}= o(N)$.
\end{itemize}

\section{On the zeros close to the origin}
\label{sec:near0}

\subsection{Results in the general setting}
We start off in the general setting of \eqref{eq:Kgeneral}. 
The proof of the following result can be found in  \cite[Ch.~9, \S~49 and \S~50]{cf:GK}.

\begin{theorem}[Local Limit Theorem]
\label{th:LLT}
For  $\ga\in (0,1)$  we set 
$a_j:= j^{1/\ga}$. In the general setting of \eqref{eq:Kgeneral}
we have
\begin{equation}
\lim_{j \to \infty}\sup_n \left\vert a_j \bP \left( \tau_j =n\right) -g_\ga\left(  \frac n{a_j}\right) \right \vert =\, 0\,,
\end{equation}
where $g_\ga(\cdot)$ is the law of the positive stable law identified by $\int_0^\infty g_\ga(y) e^{-ty} \dd y = \exp(-t^\ga)$ for $t> 0$.
\end{theorem}

\medskip

We have \cite[p. 99]{cf:Zolotarev}
\begin{equation}
\label{eq:asympt0}
g_\ga(x) \stackrel{x \searrow 0} \sim 
\frac 1{\sqrt{2 \pi \ga (1-\ga)}} \left( \frac \ga x\right)^{\frac{2-\ga}{2(1-\ga)}}
\exp\left( -(1- \ga)  \left( \frac \ga x\right)^{\frac{\ga}{1-\ga}}\right) \, ,
\end{equation}
and   \cite[p. 90]{cf:Zolotarev}
\begin{equation}
\label{eq:asymptinfty}
g_\ga(x) \stackrel{x \to \infty} \sim \frac {\Gamma(1+ \ga) \sin(\pi \ga)}\pi\,  x^{-(1+ \ga)}\, .
\end{equation}
In the special case $\ga=1/2$ the asymptotic equivalence \eqref{eq:asympt0} becomes an equality (only truly explicit case, $\ga=1/3$ and $\ga=2/3$ can be expressed via McDonald functions): for every $x>0$
\begin{equation}
\label{eq:1/2}
g_{1/2} (x)\,=\, \frac{1}{\sqrt{4 \pi x^3}} \exp\left( -\frac 1{4x}\right)\, .  
\end{equation}

\medskip

\begin{proposition} 
\label{th:usingLLT}
Uniformly for  $\gz$ in  compact subsets of $\bbC$
\begin{equation}
\label{eq:usingLLT}
 Z_{N, \gz/ N^\ga} \stackrel{N \to \infty } \sim \frac {1}{N^{1-\ga}} 
\int_0^\infty \exp\left(\gz x\right) g_\ga \left( x^{-1/\ga}\right) x^{-1/\ga}\dd x \, =:\,
 \frac {1}{ N^{1-\ga}}  F_{0, \ga}\left(\zeta\right)\,,
\end{equation}
and with $Z'_{N, h} = \partial_h Z_{N,h}$
\begin{equation}
\label{eq:usingLLT'}
 Z'_{N, \gz/ N^\ga} \stackrel{N \to \infty } \sim \frac {1}{N^{1-2\ga}} 
\int_0^\infty \exp\left(\gz x\right) g_\ga \left( x^{-1/\ga}\right) x^{1-1/\ga}\dd x 
\, =\, 
 \frac {1}{ N^{1-2\ga}}  F'_{0, \ga}\left(\zeta\right)\,.
\end{equation}
\end{proposition}

\medskip

Proposition~\ref{th:usingLLT}  is a direct consequence of
Theorem~\ref{th:LLT} and Riemann sum approximations, with a control of the tails of the sums. The arguments are standard, but  we provide some details in App.~\ref{sec:P-est}.

\smallskip

One direct consequence of Proposition~\ref{th:usingLLT} is that $F_{0, \ga}(\cdot)$ is an entire function:
in fact $\zeta \mapsto N^{1-\ga} Z_{N, \gz/ N^\ga}$ is entire and the uniform convergence implies that the limit is entire. It is of course easy to check that $F_{0, \ga}(\cdot)$ is not constant, so it  has only isolated zeros.
The following result yields a non negligible control on  the zeros of $Z_{N, \gz/ N^\ga}$ that are at distance $O(1/N^\ga)$ from the origin
 if we know where the zeros of $F_{0, \ga}(\cdot)$ are. 
 
 \medskip
 
 \begin{proposition}
 \label{th:zeros0}
 Suppose that in a bounded simply connected open set $D$ with smooth boundary there are exactly $n$ zeros of $F_{0, \ga}(\cdot)$ counted with their multiplicities. Then
 for $N$ sufficiently large there are exactly $n$ zeros of $h \mapsto Z_{N, h}$ in $D/N^\ga$. 
 \end{proposition}

\medskip 

In particular, Proposition~\ref{th:zeros0} says that if $F_{0, \ga}(\zeta_0)=0$ and if this zero is simple, then 
for every $\gep>0$ there exists $N_\gep$ such that $Z_{N, h}=0$ for exactly one $h \in B_{\zeta_0/N^\ga}(\gep)$ if $N \ge N_\gep$. 

\medskip

\begin{proof}
We have 
\begin{equation}
\frac 1{2\pi i} \oint_{\partial D/N^\ga} \frac{Z'_{N, z} }{Z_{N, z } }\dd z\, =\, 
\frac 1{2\pi   i} \oint_{\partial D}
\frac{Z'_{N, \zeta/N^\ga} }{N^\ga Z_{N, \zeta/N^\ga} } \dd \zeta \stackrel{N \to \infty} \longrightarrow
\frac 1{2\pi   i} \oint_{\partial D}
\frac{F'_{0, \ga}(\zeta) }{F_{0, \ga}(\zeta) } \dd \zeta\, =\, n\,,
\end{equation}
where the last equality is 
 the Argument Principle and the convergence step follows from Proposition~\ref{th:usingLLT}.
Since, again by the Argument Principle, the left-most term is the number of zeros of $h \mapsto Z_{N, h}$ in $D/N^\ga$, the proof is complete.
\end{proof}

\subsection{The $\ga=1/2$ case}
Unfortunately, solving 
$F_{0, \ga}(\zeta)=0$ 
appears to be too challenging. 
It should be possible to show that, for $\zeta$ large, the zeros will be in the first and fourth quadrants and close to the 
lines with directions $\exp (i \ga \pi /2)$ (in analogy with \eqref{eq:s-cf:zeros} below). However this is not straightforward 
and, as we explained in Section~\ref{sec:introG}, one is particularly interested in the zeros that are the closest (they 
come in pairs, unless they are real) to the origin. 

Therefore, in order to go farther, we 
 specialize to $\ga=1/2$. As we announced, in this case things get more explicit:
\begin{equation}
\label{eq:F0}
F_0(\zeta)\,:=\, F_{0,\frac 12}(\zeta)\,=\, 
e^{\zeta^2}\zeta\left(1+\mathrm{erf}(\zeta)\right)+\frac{1}{\sqrt{\pi }}\, ,
\end{equation}
and we record that
\begin{equation}
\label{eq:F0'}
F'_0(\zeta)\,=\, 
\frac{2 \zeta}{\sqrt{\pi}} + \exp(\zeta^2) \left(1 + 2 \zeta^2\right) \left(1 + \mathrm{erf}(\zeta)\right)\, .
\end{equation}
\medskip

\begin{rem}
\label{rem:simplezeros}
Note that  $F'_0(\zeta)=(2\zeta+1/\zeta)F_0(\zeta)- 1/(\sqrt{\pi} \zeta)$ so
\begin{equation}
F_0(\zeta)\, =\,0 \ \ \Longrightarrow \  \
F'_0(\zeta)\, =\, -\frac 1 {\zeta \sqrt{\pi}}\,.
\end{equation}
In particular, the zeros of $F_0$ are simple. We record also, for later use, that $F^{''}_0(\zeta)=-2/\sqrt{\pi}$
if $F_0(\zeta)=0$.
\end{rem}
\medskip

In spite of the rather explicit expression for $F_0(\cdot)$, it does not appear that 
$F_0(\zeta)=0$ can be solved explicitly. 
What we are mostly interest in are the zeros that are closest to the origin: we can only identify them numerically. Nevertheless something can be said rigorously. Moreover the numerical approximations can be controlled rigorously, at least if we accept the assistance of the computer for symbolic computations.  
 
 For the statement,  order the solution of  $\zeta_j$ to $F_0(\zeta_j)=0$ so that $\vert \zeta_j\vert$ is non decreasing in $j$. We can assume that $\Im(\zeta_1) \ge 0$, and set $\zeta_2= \overline{\zeta}_1$ (unless $\zeta_1$ is real).
A priori there could still be more than one choice for $\zeta_1$.  

\medskip

\begin{lemma}
\label{th:zeta_1}
$\Re(z_j)>0$ and $\Im(z_j)\neq 0$ for every $j$ (hence we can stipulate that  $\Im(z_{2k-1})=-\Im(z_{2k})>0$ for $k=1,2, \ldots$). Moreover 
$\zeta_1$ is well defined (i.e., $\vert z_1\vert <\vert z_3 \vert$). In fact,
 $\zeta_1= 1.225  +  2.547 i +r_1$ and $z_3= 2.026 + 3.162 i + r_3$, with $\vert r_1\vert$ and $\vert r_3\vert$ smaller than $0.0005$. 
\end{lemma}
\medskip

Of course Lemma~\ref{th:zeta_1} is also implicitly saying that $\vert \zeta_j\vert \ge  \vert \zeta_3\vert=
 \vert \zeta_4\vert$ for $j =5,6, \ldots$. 

\smallskip

\begin{proof}
For $\Re(\zeta)<0$ we use the representation 
\begin{equation}
\label{eq:intrepr<0}
\frac \pi 2 F_0(\zeta)\, =\, \frac 12  \int_0^\infty \frac{e^{-y} \sqrt{y}}{y+\zeta^2} \dd y \, =\, \int_0^\infty \frac{e^{-x^2} x^2}{x^2+\zeta^2} \dd x\, .
\end{equation}
\medskip

\begin{rem}
\label{rem:DLMF} 
\eqref{eq:intrepr<0} follows from \cite[(7.2.3) and (7.7.2)]{cf:DLMF}. To see this it is quicker to exploit also the complementary error function
$\mathrm{erfc}(\zeta)= 1- \mathrm{erf}(\zeta)$. So the symmetry $\mathrm{erf}(-\zeta)=-\mathrm{erf}(\zeta)$ that holds for $erf(\cdot)$ is equivalent to $\mathrm{erfc}(-\zeta)=2-\mathrm{erf}(\zeta)$ and
\begin{equation}
\label{eq:remDLMF1} 
F_0(\zeta)\, =\, \zeta e^{\zeta^2}\mathrm{erfc}(-\zeta)+ \frac 1 {\sqrt{\pi}}\,. 
\end{equation}
The identities  \cite[(7.2.3) and (7.7.2)]{cf:DLMF} yield that for  $\Im(z)>0$
\begin{equation}
\label{eq:remDLMF2} 
e^{-z^2} \mathrm{erfc}(-iz)\, =\, \frac {2z}{\pi i} \int_0^\infty \frac{\exp(-t^2)}{t^2-z^2}\dd t\,.
\end{equation}
The identities \eqref{eq:remDLMF1} and \eqref{eq:remDLMF2} imply 
 \eqref{eq:intrepr<0}, which holds for $\Re(\zeta)<0$,
and also that for $\Re(\zeta)>0$ we have instead 
\begin{equation}
\label{eq:intrepr>0}
\frac 2 \pi  \int_0^\infty \frac{e^{-x^2} x^2}{x^2+\zeta^2} \dd x\, =\, F_0(\zeta) - 2 \zeta e^{\zeta^2}\, .
\end{equation}
\end{rem}

\medskip

With $\zeta=u+iv$ we see that the real and the imaginary part of the previous quantity are respectively
\begin{equation}
\label{eq:intrepr5}
\int_0^\infty \frac{e^{-x^2} x^2\left(x^2+u^2-v^2\right)}
{\left(x^2+u^2-v^2\right)^2+ 4 u^2 v^2} \dd x\  \  \text{ and } \ \ 
2 u v\int_0^\infty \frac{e^{-x^2} x^2}
{\left(x^2+u^2-v^2\right)^2+ 4 u^2 v^2} \dd x\, ,
\end{equation} 
and, since we are assuming that $u<0$, the second expression -- the imaginary part -- is zero if and only if $v=0$. 
But in that case the first expression -- the real part --  is positive. Therefore $F_0(\zeta)\neq 0$ if $\Re (\zeta)<0$. 

For the case $u=\Re(\zeta)=0$ we directly  use \eqref{eq:F0} and we rewrite it as 
\begin{equation}
F_0(iv)\,  =\, \frac 1{\sqrt{\pi}}-v e^{-v^2} \mathrm{erfi}(v) +i v e^{-v^2}\,, 
\end{equation}
where $\mathrm{erfi}(v): = \mathrm{erf}(iv)/i = (1/\sqrt{\pi}) \int_0^v e^{t^2} \dd t$. So  $\mathrm{erfi}(\cdot)$ is real and odd on the real axis. Moreover   it is    positive on the positive semiaxis. From this we readily infer that $F_0(iv)\neq 0$ for every $v$: in fact $\vert F_0(iv)\vert$ vanishes only for $\vert v\vert \to \infty$. 

The fact that $F_0(\zeta)>0$  for $\zeta>0$, in fact $F_0(\zeta)>1/\pi$, follows from
 $F'_0(\zeta)>0$, see \eqref{eq:F0'}. 
 
In order to determine $\zeta_1$ and $\zeta_3$ (in fact, every $\zeta_j$ in principle) we need to write a sufficiently precise polynomial approximation of $F_0(\cdot)$, with adequate control of the remainder, and use 
the Argument Principle (see for example the proof of Proposition~\ref{th:zeros0}). 
Implementing this approach in practice, however, is quite cumbersome and probably can only usefully be done on a computer.
  
In order to establish that there are infinitely many zeros one can adapt the approach in \cite{cf:zeros}. In fact, one can
identify  
a sequence of simple zeros that satisfy
\begin{equation}
\label{eq:s-cf:zeros}
\zeta_n \, =\, \gl_n - \frac 1{4 \gl_n}\log \left( 8 \sqrt{2 \pi} \gl_n^3  \right) + i \left(
\gl_n + \frac 1{4 \gl_n}\log \left( 8 \sqrt{2 \pi} \gl_n^3  \right) 
\right) + O\left( \vert \log n \vert ^2 / n^{3/2} \right)\, ,  
\end{equation}
with $\gl_n= (\pi(n+1/8))^{1/2}$. One can also show that, sufficiently far from the origin, there is no other zero (up to conjugation). 
We do not go into the lengthy details of this result that is not central for us, but one can use \eqref{eq:intrepr>0}; 
a key point is that
\begin{equation}
\label{eq:erfi0}
\limtwo{\vert \zeta \vert  \to \infty:}{\Re(\zeta)> 0} e^{\zeta^2} \zeta \, \text{erfc}(\zeta)\, =\, \frac 1{\sqrt{\pi }}\, .
\end{equation}
In fact, by the continuous fraction expansion \cite[(7.9.1)]{cf:DLMF}, we have that in the same limit
\begin{equation}
\label{eq:erfi0-2}
\frac1{\sqrt{\pi }}-  e^{\zeta^2} \zeta \, \text{erfc}(\zeta)\, =\, \frac 1{2 \sqrt{\pi} \zeta^2} +O\left( \frac 1{\zeta^4} \right)\, .
\end{equation}
By writing the analog of \eqref{eq:intrepr5} for $\Re(\zeta>0$ and using \eqref{eq:erfi0} and  \eqref{eq:erfi0-2}
one can see that the zeros (that are far from the origin)  need to be close to the diagonal of the first and second quadrant. And a controlled perturbation analysis leading to \eqref{eq:s-cf:zeros}. 

The asymptotic formula \eqref{eq:s-cf:zeros} turns out to be surprisingly accurate even for $n$ small: see Table~\ref{tab:1}.
 \end{proof}

\medskip

\begin{table}[h!]
\centering
 \begin{tabular}{|c |c |c |c|} 
 \hline
 $n$ & $\zeta_{n}$ & $\zeta^\sim_n $& $\vert \zeta_{n} -\zeta^\sim_{n}\vert$  \\ [0.5ex] 
 \hline\hline
 1 & 1.225 + 2.547\,$i$ &1.229 + 2.531\,$i$ & 0.017 \\ 
 2 & 2.026 + 3.162\,$i$ &2.018 + 3.149\,$i$ & 0.015 \\
 3 & 2.629 + 3.656\,$i$ & 2.621 + 3.646\,$i$ & 0.013 \\
 4 & 3.132 + 4.083\,$i$ &  3.125 + 4.075\,$i$ & 0.011 \\
 5 &  3.573 + 4.466\,$i$ &  3.566 + 4.459\,$i$ & 0.010 \\
 6 &  3.969 + 4.817\,$i$ & 3.963 + 4.810\,$i$ & 0.009 \\
 7 &  4.332 + 5.141\,$i$ &  4.326 + 5.136\,$i$ & 0.008 \\ [0.5ex] 
 \hline
 \end{tabular}
 \caption{\label{tab:1} \emph{Exact} (i.e., numerically evaluated) and approximate (i.e., $\zeta^\sim_{n}$ is the right-hand side of  \eqref{eq:s-cf:zeros} without the remainder)  location of the zeros of $F_{1/2}$. Here we consider only the zeros with positive imaginary parts, so 
 $\zeta_n$ is an abuse of notation for  $\zeta_{2n -1}$.}
\end{table}

By combining Proposition~\ref{th:zeros0} and Lemma~\ref{th:zeta_1} we readily reach:

\medskip

\begin{cor}
\label{th:locate0}
For $N$ sufficiently large, $h_{N,1}=\overline{h}_{N,2} \sim \zeta_1/\sqrt{N}$ and for $j=3,\ldots, N-1$
\begin{equation}
\frac{\vert h_{N, j}\vert }{ \vert h_{N,1} \vert}> 1+\frac 12 \left( \frac{\vert \zeta_3\vert}{\vert \zeta_1\vert}-1\right)\, .
\end{equation}
\end{cor}
\medskip

The factor $\frac 12$ is of course arbitrary and may be replaced by any number in $(0, 1)$.


\subsection{Sharper control}
\label{sec:sharper}

Can one go  beyond Corollary~\ref{th:locate0}? For example,
sticking to $\ga =1/2$, one might wonder whether a development  like
$h_{N, j}= z_0/\sqrt{N} + z_1/N+z_2/N^{3/2} +\ldots$, of course with $z_0= \zeta_j$, holds.
It is not difficult to convince oneself that this  cannot hold in the general framework  of \eqref{eq:Kgeneral}. 

We develop this issue in the special case of \eqref{eq:Kspecial} and our motivation is that 
such a precise estimate is needed in Section~\ref{sec:G}.
\medskip

\begin{proposition}
\label{th:sharper0}
Assume that \eqref{eq:Kspecial} holds with $\ga=1/2$ and 
Fix $j \in \bbN$. We have that
\begin{equation}
h_{N, j}= \frac{z_0}{\sqrt{N}} + \frac{z_1}{N}+\frac{z_2}{N^{3/2}}+ O\left( \frac 1{N^2}\right)\, ,
\label{eq:sharper0}
\end{equation}
where $z_0= \zeta_j$ and
\begin{equation}
z_1\,  =\, \frac 12 z_0^2 \ \text{ and } \ z_2\, =\, \, \frac{1}{24} \sqrt{\pi } z_0 \left(12 z_0^4+2 z_0^2-3\right)\, .
\end{equation}
\end{proposition}

\medskip

One can push \eqref{eq:sharper0} to an arbitrary large order, at the price of more and more cumbersome computations: \eqref{eq:sharper0} suffices for our purposes.

It is not difficult to realize that in  the restricted framework \eqref{eq:Kspecial} one can get finer and finer approximations of $Z_{N,\zeta/ \sqrt{N}}$ via 
Stirling expansion, but this turns out to be very involved. We have found it easier to exploit the representation of $Z_{N, h}$ recently given in   
 \cite{cf:EN}: for the special case of
 \eqref{eq:Kspecial} the partition function $Z_{N, h}$ is the $N$-th moment of a positive random variable:
\begin{equation}
\label{eq:what-the-}
Z_{N, h} \, =\, \int_{(0, \infty)} x^N \nu_h(\dd x)\, ,
\end{equation}  
where $\nu_h$ is a probability measure . For $\ga=1/2$ (see \cite{cf:EN} for $\ga \in (0,1)$) 
\begin{equation}
\nu_h(\dd x) \, :=\, 
\frac{e^h}{\pi x} 
\frac{\sqrt{x (1-x)}}{\left(x(1-2e^h)+e^{2h}\right)} \ind_{(0,1)}(x) \dd x + \frac{2 (e^h-1) }{2e^h -1} \ind_{(0, \infty)}(h)  
\gd_{e^{2h}/(2 e^h -1)} (\dd x)\,.
\end{equation}
This result is at first sight surprising because $Z_{N, h}$ is a polynomial  in $\exp(h)$, while the 
right-hand side in \eqref{eq:what-the-} has different expressions for $h>0$ and $h<0$
 because of the delta contribution to $\nu_h$ that we can of course view as  $\nu_h^{\text{abs}}+  \nu_h^{\text{sing}}$ separating thus absolutely continuous and singular part of the measure. 
 The subtlety here is that there is a singularity in  the denominator of  $\nu_h^{\text{abs}}$: note that the density of the absolutely continuous part has a meaning also for $h\in \bbC$, even if of course it is no longer a probability density, while for the singular part the analytic continuation can be done only after integration. We can appreciate better this singularity by remarking that for $x \in [0,1]$ and $h\in \bbC$ small 
\begin{equation}
\label{eq:expa-0}
x(1-2e^h)+e^{2h}\, =\, (1-x) (1+2h)+ (2-x) h^2 +O(h^3)\, ,
\end{equation}
so for $x$ near $1$ the dominant contribution is $(1-x)+h^2$ (the remainder is $O((1-x)h)+O(h^3)$), which yields a non integrable singularity for imaginary $h$. 
As a matter of  fact, one directly checks that the right-hand side in \eqref{eq:what-the-} is analytic for $\Re( h) < 0$ and for 
$\Re( h)>0$. For $\Re( h) <0$ and using the parametrization  $h=\zeta/\sqrt{N}$ we have for $N \to \infty$ 
\begin{equation}
\label{eq:stps83}
\begin{split}
\int_{(0, 1)} x^N \nu_{\zeta/\sqrt{N}}(\dd x)\, &=\, 
\frac{e^{\zeta/\sqrt{N}}}{\pi } \int_{(0,1)} x^{N-1}
\frac{\sqrt{x (1-x)}}{\left(x(1-2e^{\zeta/\sqrt{N}})+e^{2\zeta/\sqrt{N}}\right)} \dd x
\\
&\sim \, 
\frac{1}{\pi } \int_{(0,1)} x^{N-1}
\frac{\sqrt{(1-x)}}{(1-x) + \zeta^2/N} \dd x\\
& \sim \, 
\frac{1}{\pi  } \int_{(0,1)} \exp(-y N)
\frac{\sqrt{y}}{y + \zeta^2/N} \dd y\, \sim \, 
\frac{1}{\pi \sqrt{N} } 
\int_0^\infty \exp(-y )
\frac{\sqrt{y}}{y + \zeta^2} \dd y\, ,
\end{split}
\end{equation}
where in the first asymptotic statement we have used \eqref{eq:expa-0} and the fact that the leading contribution to the integrals involved comes from $x$ close to $1$.
The very same computation holds for $\Re( h) >0$, hence $\Re( \zeta)>0$, because we have restricted the integral to $(0,1)$, so we are effectively only integrating with respect to
  $\nu^{\text{abs}}_{\zeta/\sqrt{N}}$. Without surprise we have that for $\Re(  \zeta)<0$ (see Remark~\ref{rem:DLMF})
\begin{equation}
\frac{1}{\pi  } 
\int_0^\infty \exp(-y )
\frac{\sqrt{y}}{y + \zeta^2} \dd y\,=\, \frac{2}{\pi  } 
\int_0^\infty \exp(-x^2 )
\frac{x^2}{x^2 + \zeta^2} \dd x\, =\, F_0(\zeta)\, 
\end{equation}
and  for $\Re(  \zeta)>0$  
\begin{equation}
\label{eq:Rez>0}
\frac{1}{\pi  } 
\int_0^\infty \exp(-y )
\frac{\sqrt{y}}{y + \zeta^2} \dd y\,=\, 
F_0(\zeta)-2\zeta \exp\left(\zeta^2 \right)\,.
\end{equation}
One then easily verifies that 
\begin{equation}
\label{eq:gNz}
g_N(\zeta)\, :=\, 
\int_{[1, \infty)} x^N \nu_{\zeta/\sqrt{N}}(\dd x) \, =
\, \frac{2\left(e^{z/\sqrt{N}}-1\right) e^{2z \sqrt{N}}
}
{
\left(2e^{z/\sqrt{N}}-1\right)^{N+1}
}
\stackrel{N \to \infty}
\sim \frac{2\zeta \exp\left(\zeta^2 \right)}{\sqrt{N}}
\, .
\end{equation}
Therefore the steps \eqref{eq:stps83}--\eqref{eq:gNz}
 provide an alternative proof of
\eqref{eq:usingLLT} in the restricted set up of \eqref{eq:Kspecial}, only for $\ga=1/2$ and only for $\Re (z) \neq 0$. 
This of course is a very poor result with respect to Proposition~\ref{th:usingLLT}. But 
\eqref{eq:what-the-} turns out to be very efficient when we want to  obtain higher order corrections in $1/ \sqrt{N}$
and that is why we use it now. 

\medskip

\begin{proof}[Proof of Proposition~\ref{th:sharper0}] 
We need  to consider  only  the case  $\Re(h)>0$, but dealing at the same time with $\Re(h)<0$ essentially affects only one formula, i.e. \eqref{eq:EN2}, and the estimates we do work just assuming that $\Re(h)$ is bounded away from $0$. Therefore we treat both cases at the same time till \eqref{eq:both-tillhere}. We set
\begin{equation}
f_N(\zeta,y)\, :=\, 
\frac{e^{\zeta/\sqrt{N}}}{\pi N }
\frac{(1-y/N)^{N}\sqrt{(y/N) /(1-y/N)}}{\left((1-y/N)(1-2e^{\zeta/\sqrt{N}})+e^{2\zeta/\sqrt{N}}\right)}\, ,
\end{equation}
and, by recalling \eqref{eq:gNz}, we see that
\begin{equation}
\label{eq:EN2}
Z_{N, \zeta/\sqrt{N}}\, =\begin{cases} 
\int_0^N f_N(\zeta,y) \dd y & \text{ if } \Re( \zeta)< 0\, ,
\\
\int_0^N f_N(\zeta,y) \dd y + g_N(\zeta) & \text{ if } \Re( \zeta)> 0\, .
\end{cases}
\end{equation}
In what follows 
we  consider $\zeta$ belonging to a compact subset $K\subset \bbC$. We have
\begin{multline}
\label{eq:g-exp}
g_N(\zeta)\,=\, \frac{2}{\sqrt{N}} \zeta \exp\left(\zeta^2\right) - \frac 1N \zeta^2(3+2\zeta^2)  \exp\left(\zeta^2\right)
\\
+
 \frac 1{6N^{3/2}} \zeta^3 \left(26+31 \zeta^2+6 \zeta^4\right)
\exp\left(\zeta^2\right)+ 
O \left( \frac 1 {N^2} \right)\, ,
\end{multline}
and for every $y$, $\zeta\in K$ and for every sufficiently large $N$ 
\begin{equation}
\left \vert 
f_N(\zeta,y) -\frac 1{\sqrt{N}} \frac{e^{-y}\sqrt{y}}{\pi \left(y -\zeta^2\right)}\right \vert \, \le \, 
\frac C N \frac{y}{\left \vert y+ \zeta^2\right \vert^2}\, \, ,
\end{equation}
with $C=C_K$.

\medskip
\begin{lemma}
\label{th:small}
For every $\zeta\in \bbC$
we have $\min\{\vert x+\zeta^2\vert:\, x\ge 0\}\ge (\Re( \zeta))^2$.
\end{lemma}

\medskip
\begin{proof}
We have $\vert x+\zeta^2\vert^2= x^2 +2x \Re (\zeta^2) +\vert \zeta \vert^4$ so the minimum of this expression is reached at 
$x_0=-\Re (\zeta^2)= (\Im( \zeta))^2-(\Re \zeta)^2$ if $x_0>0$ and it is reached at $x=0$ if $x_0\le 0$. In the second case  
$\min (\vert x+\zeta^2\vert)= \vert \zeta\vert^2\ge (\Re( \zeta))^2$.
In the first case $\min (\vert x+\zeta^2\vert)= \vert (\Im( \zeta^2))\vert= 2\vert \Re( \zeta)\vert \vert \Im (\zeta)\vert$ which is bounded below by 
$2(\Re (\zeta))^2$ because $x_0>0$.
\end{proof}

\medskip

Lemma~\ref{th:small} tells us that the limit we are interested can be handled uniformly in $\zeta\in K$ and 
$\vert \Re( \zeta)\vert$ bounded away from zero. This leaves a strip out that a priori is non trivial to handle, but this is of course not a problem 
because we already know that the zeros are not there (cf., Lemma~\ref{th:zeta_1} and Corollary~\ref{th:locate0}). 

We use 
\begin{multline}
\pi \sqrt{N} \frac{\exp(y)}{\sqrt{y}}f_N(\zeta,y)\, =\, \frac 1{\left( y+ \zeta^2\right)}- \frac  1{\sqrt{N} }
\frac{y \zeta}{ \left( y+ \zeta^2\right)^2}
\\
+ \frac 1N 
\frac{6 y^3 -6y^4+30 y^2 \zeta^2 -12 y^3 \zeta^2 +11 y \zeta^4-6 y^2 \zeta^4 -\zeta^6}
{12\left( y+ \zeta^2\right)^3} + \frac1{N^{3/2}} R_N(\zeta,y)\, ,
\end{multline}
where, for $\zeta\in K$ and $\vert \Re( \zeta)\vert \ge \gd>0$ we have 
$\vert  R_N(\zeta,y)\vert \le C_K (1+y^5)/\gd^{8} $, with $C_K$ a constant that depends on the choice of the compact set $K$. Therefore for $\zeta\in K$ and $ \Re( \zeta) \ge \gd>0$ (recall the contribution from
\eqref{eq:g-exp})
\begin{equation}
\label{eq:both-tillhere}
\sqrt{N}Z_{N, \zeta/\sqrt{N}}\, =\, F_0(\zeta) + \frac 1{\sqrt{N}} F_1(\zeta)+ \frac  1 N F_2(\zeta) + O\left(  \frac  1 {N^{3/2}}\right)\,,
\end{equation}
where 
\begin{equation}
F_0(\zeta)\, =\, e^{\zeta^2} \zeta \, \left(1+\textrm{erf}(\zeta)\right) +\frac{1}{\sqrt{\pi }}\, ,
\end{equation}
is an entire function.  Also $F_1$ and $F_2$ are entire functions whose rather awful expressions 
\begin{equation}
F_1(\zeta)\, =\, 
-\frac{1}{2} \zeta \left(e^{\zeta^2} \zeta \left(2 \zeta^2+3\right) (1+\text{erf}(\zeta))+\frac{2 \left(\zeta^2+1\right)}{\sqrt{\pi }}\right)\, ,
\end{equation}
and 
\begin{equation}
F_2(\zeta)\, =\, 
\frac{1}{24} \left(2 e^{\zeta^2} \left(6 \zeta^4+31 \zeta^2+26\right) \zeta^3 (1+\text{erf}(\zeta))+\frac{12 \zeta^6+56 \zeta^4+30 \zeta^2-3}{\sqrt{\pi }}\right)\, ,
\end{equation}
considerably simplify for $\zeta= \zeta_j$, that is for $\zeta$ such that
$F_0(\zeta)=0$. It is more practical to introduce the notation $z_0$ for such values $\zeta_j$ (and this is the notation used in the statement of Proposition~\ref{th:sharper0}): 
\begin{equation}
F_1(z_0) 
\,=\, 
\frac{1}{2\sqrt{\pi}} z_0 \ \ \text{ and } \ \  F_2(z_0)\, =\, -\frac{6 z_0^4+22 z_0^2+3}{24 \sqrt{\pi }}\, .
\end{equation}
Recall that, by Corollary~\ref{th:locate0}, 
$\sqrt{N} h_{N, j}\sim \zeta_j=z_0$ and that 
 $F_0(z_0)=0$ implies  $F'_0(z_0)\neq 0$ (Remark~\ref{rem:simplezeros}). If we expand the left-hand side of $Z_{N, z_0/N^{1/2} +z_1/N+ z_2/N^{3/2}+\ldots}=0$
 and solve the equation order by order,  
 we are lead to guessing  
\begin{equation}
\label{eq:guess}
\sqrt{N} h_{N, j}\, =\, z_0 + \frac{z_1}{\sqrt{N}} + \frac {z_2} N +O\left( \frac 1{N^{3/2}}\right)\,
,
\end{equation}
with
\begin{equation}
z_1=-
 \frac{F_1(z_0)}{F'_0(z_0)} \, =\frac 12 z_0^2 \ \text{ and } \ z_2= \frac{(1/2)F''_0(z_0)z_1^2 -F'_1(z_0) z_1 -F_2(z_0)}
 {F'_0(z_0)}\, .
\end{equation}
Since
\begin{equation}
F'_0(z_0)\, =\, -\frac 1 {z_0 \sqrt{\pi}},\ \  \ F''_0(z)\, =\, -\frac 2{\sqrt{\pi}} \ \ \text{ and } \ \ 
F'_1(z_0)\, =\, \frac{z_0^2+2}{\sqrt{\pi }}\, ,
\end{equation}
we have
\begin{equation}
z_1\,  =\, \frac 12 z_0^2 \ \text{ and } \ z_2\, =\, \, \frac{1}{24} \sqrt{\pi } z_0 \left(12 z_0^4+2 z_0^2-3\right)\, .
\end{equation}

In order to make \eqref{eq:guess} rigorous we need to show 
that there exists $R>0$ and $N_0>0$ such that for $N>N_0$ there exists a (unique)  $\zeta \in B_{z(N)}(R/N^{3/2})$, with 
\begin{equation}
z(N)\, =\, z_0 + \frac{z_1}{\sqrt{N}} + \frac {z_2} N\, , 
\end{equation}
that solves $Z_{N, \zeta/ \sqrt{N}}=0$. 
For this we remark that, while \eqref{eq:guess} is formal, the procedure that leads to it (Taylor expansion) does yield
\begin{equation}
\label{eq:guess-rig}
Z_{N, z(N)/ \sqrt{N}}\, =\, O\left( \frac 1{N^2} \right)\, ,
\end{equation}
and, by applying the Argument principle with $C_N:=\partial B_0(R/N^2)$, it suffices to show that
\begin{equation}
\label{eq:arg-principle}
\frac 1{2 \pi i}
\oint _{C_N}
\frac{
Z'_{N, z(N)/\sqrt{N} + \zeta} }
{Z_{N, z(N)/\sqrt{N} + \zeta} } \dd \zeta -1 \, =\, 0\, ,
\end{equation}
for $N$ sufficiently large. For this we write
\begin{equation}
Z_{N, z(N)/\sqrt{N} + e^{it} R /N^2} \, =\, F'(z_0) e^{it} \frac{R} {N^{2}} + r_N(R,t)\, ,  
\end{equation}
and, by applying the first order Taylor expansion and using \eqref{eq:guess-rig}, there exists $c>0$ such that for every $R$ we have  $\vert r_N(R,t)\vert \le c/N^2$ for every $t$ and $N$ sufficiently large (how large may depend on $R$).
Moreover
\begin{equation}
Z'_{N, z(N)/\sqrt{N} + e^{it} R /N^2} \, \sim\, F'(z_0) \, ,
\end{equation}
uniformly in $t$. 
This means that for every  $R>2c$ 
\begin{equation}
\limsup_N \sup_t \left \vert  \frac{Z'_{N, z(N)/\sqrt{N} + e^{it} R /N^2}}{N^2 Z_{N, z(N)/\sqrt{N} + e^{it} R /N^2}}- \frac{\exp(-it)}{R} \right \vert \, \le \, \frac {3c}{R^2}\, .
\end{equation}
This means that, if we choose $R$ properly large, we can make the absolute value of the left-hand side in \eqref{eq:arg-principle} smaller than $1$ for $N$ sufficiently large.
Hence, for such values of $N$, \eqref{eq:arg-principle}  holds and the proof of Proposition~\ref{th:sharper0} is complete.
\end{proof}

\section{A reduced model for Griffiths singularities: the proof}
\label{sec:G}

This section just deals with the special framework \eqref{eq:Kspecial} and with $\ga=1/2$.
However, several equations are more readable if we write $\ga$ for $1/2$, therefore we will do so.

\medskip

\noindent
\emph{Proof of Theorem~\ref{th:griffiths}.}
Recall that $(h_{N, j})$ are the $N-1$ zeros of $Z_{N, h}$, ordered with non decreasing modulus and $\Im(h_{N, j})>0$ (respectively,
$\Im(h_{N, j})<0$) for $j$ odd (respectively, even). In analogy with \eqref{eq:PN2} we have
\begin{equation} 
\label{eq:PN2.1}
Z_{N,h}\, =\, K(1)^N \exp(h) \prod_{j=1}^{N-1} 
\left(e^h-e^{h_{N,j}} \right)
\, =\, Z_{N,0}e^h
 \prod_{j=1}^{N-1} \left(
 \frac{e^h-e^{h_{N,j}}}{1-e^{h_{N,j}}}
 \right)
\,,
\end{equation}
 where $Z_{N,0}= \bP(N \in \tau)$. 
Since
\begin{equation}
	\left( 1+\frac{e^h-1}{1-e^\eta} \right) \big/ \left( 1- \frac{h}{\eta} \right)
	=
	\frac{e^h - e^\eta}{h-\eta}
	\frac{\eta}{1-e^\eta}\, ,
\end{equation}
has no zeros and only removable singularities for $\Im (h) \in (-\pi, \pi)$ and  $\Im (\eta) \in (-\pi,\pi]$, for such $\eta$'s  we can  extend
\begin{equation}
	h \mapsto 
	\log 
	\left( 1+\frac{e^h-1}{1-e^\eta} \right) 
	-
	\log 
	\left( 1- \frac{h}{\eta} \right)\, ,
\end{equation}
to an analytic function on the strip $\Im (h) \in (-\pi, \pi)$. Note that this function is also bounded for $h$ and $\eta$ in compact subsets of $\bbC$.  
As a result, we can study the regularity of
\begin{equation}
\label{eq:ftilde}
	 f
	:
	h \mapsto 
	\sum_{n=n_0}^\infty 
	p^n
	\sum_{j=1}^{n-1}
	\log\left( 1 - \frac{h}{h_{n,j}} \right)\, ,
\end{equation}
with $n_0$ fixed, but arbitrary:  we can neglect the contribution for $n< n_0$  because, it is straightforward to see 
that there exists no solution to $Z_{n, h}=0$ for  $h\in  \bbR$: in fact,   $Z_{n, h}>0$ is even increasing in $h$. Therefore, 
for $n$ fixed the (complex) solutions to $Z_{n, h}=0$, say for $h$ bounded away from $-\infty$ and $\infty$, are bounded away from the real axis and the neglected terms yield a real analytic contribution. 

The fact that $f(\cdot)$ is real analytic away from the origin can be seen as consequence of 
Proposition~\ref{th:res0} that guarantees that for every $\gep>0$ there exists $n_0$ such that 
\begin{equation}
\inf_{n\ge n_0} \inftwo{j=1, \ldots, n-1:}{\vert \Re (h_{n,j} )\vert> \gep } \vert \Im  (h_{n,j} )\vert\, >\, 0\,. 
\end{equation}

The lack of real analyticity in the origin is more subtle, but let us first show that $f$ is $C^\infty$ also at the origin. 
The 
 $k$-th derivative of $f$ is
 \begin{equation}
 {f}^{(k)}(h)\, =\, -(k-1)! \sum_{n=n_0}^\infty 
	p^n
	\sum_{j=1}^{n-1} (h_{n,j}-h)^{-k}\, ,
 \end{equation}
 a priori at least for $h\neq 0$. To see  that $f$ is $C^k$, for every $k$, also in $0$ 
 it suffices to find an appropriate bound on the internal sum, the one over $j$, for $h$ is a neighborhood of the origin. For this we remark that, uniformly for $h$ such that  $\vert \Im(h)\vert \le 1/n^2$, we have
 \begin{equation}
 Z_{n,h} \, =\, \sum_{j=1}^n e^{j\Re(h)} \bP(\tau_j=n)e^{j\Re(h)}e^{j \Im(h)}\, =\,  Z_{n,\Re(h)} \left(1+O(1/n) \right)\, ,
 \end{equation}
 which implies that $Z_{n,h}$ has no zero in the strip $\vert \Im(h)\vert \le 1/n^2$ for $n\ge n_0$ and $n_0$ appropriately chosen. Therefore, for $h\in \bbR$, we have that   $\vert \sum_{j=1}^{n-1} (h_{n,j}-h)^{-k} \vert \le n^{2k+1}$ which  suffices to show that $f\in C^k$ for every $k$. 
 
 We are now to the lack of real analyticity for which we identify the sharp $k \to \infty$ behavior of
\begin{equation}
\label{eq:ftilde_taylor}
	\frac{1}{k!}
	{f}^{(k)}(0)
	\,=\,-
	\frac{1}{k}
	\sum_{n=n_0}^\infty 
	p^n
	\sum_{j=1}^{n-1}
	(h_{n,j})^{-k}\, .
\end{equation}
For this we start by setting
\begin{equation}
	\eta_{n,j}:=
	n^{\alpha} h_{n,j}
	\, ,
\end{equation}
so we can write 
\begin{equation}
\label{eq:ftilde_taylor2}
	\frac{1}{k!}
	{f}^{(k)}(0)
	\,=\, -
	\frac{1}{k}
	\sum_{n=n_0}^\infty 
	p^n
	n^{\alpha k}
	\sum_{j=1}^{n-1}
	( \eta_{n,j})^{-k}\, .
\end{equation}
By Proposition~\ref{th:sharper0} we know that 
\begin{equation}
\label{eq:fpropr1}
 \eta_{n,1}\, =\, z_0+ \frac{z_1}{\sqrt{n} }+  \frac{z_2}n + O\left( \frac 1{n^{3/2}} \right)
 \, =\, z_0 \exp \left( \frac{z_1}{z_0 \sqrt{n}}+ \left( \frac {z_2}{z_0} - \frac {z_1^2}{2 z_0^2}\right) \frac 1n +
 O\left( \frac 1 {n^{3/2}}\right) \right)
 \,,
\end{equation}
and that there  exist $n_0\in \bbN$ and  $q\in (0,1)$ such that for $n \ge n_0$
\begin{equation}
\label{eq:fpropr2}
\sup_{j= 3, \ldots, n-1}\frac
{
\left \vert \eta_{n,1}\right \vert
}
{
\left \vert \eta_{n,j}\right \vert}
\le  q\,.
\end{equation} 
From \eqref{eq:fpropr1} we directly have also that 
$1-\gep\le \vert \eta_{n,1}\vert/ \vert z_0\vert \le 1+ \gep$ for every $n\ge n_0$: choosing $\gep$ close to zero  amounts to choosing $n_0$ larger. 
It is therefore natural, in the limit $k \to \infty$, to single out the contribution due to $\eta_{N,1}$ and
 $\eta_{N,2}= \overline{\eta_{N,1}}$ so we write
\begin{equation}
\label{eq:ftilde_taylor3.1}
\begin{split}
\frac{1}{k!}{f}^{(k)}(0)
	\,&=\, -
	\frac{2}{k}\sum_{n=n_0}^\infty p^n n^{\alpha k}
	\vert \eta_{n,1}\vert ^{-k} \cos \left( k \arg (\eta_{n,1}) \right)
	- \frac{1}{k}
	\sum_{n=n_0}^\infty p^n n^{\alpha k}
	\sum_{j=3}^{n-1}
	( \eta_{n,j})^{-k}
	\\
&=\, T_k +E_k
	\, \, .
\end{split}	
\end{equation}
We bound $E_k$ by  using \eqref{eq:fpropr2},  
 so the 
terms in the sum over $j$ can be bounded by  $q/\vert z_0 (1-\gep) \vert  $ to the power $k$  which does not depend on $j$. Therefore
\begin{equation}
\left \vert E_k \right \vert \,\le\,  \frac{1 }{k} \left(\frac q{\vert z_0 \vert (1-\gep) } \right)^k
	\sum_{n=n_0}^\infty p^n n^{\alpha k+1}
	  \, .
\end{equation}
At this point it is useful to introduce the polylogarithm of parameter $s\in \bbR$:
\begin{equation}
\mathrm{Li}_s(z)\, =\, \sum_{n=1}^\infty \frac {z^n}{n^s}\, ,
\end{equation}
for $z$ in the open unit ball. For $x\in (0,1)$ we have (see App.~\ref{sec:polylog} for references and more details) 
\begin{equation}
\label{eq:Li-asympt}
\mathrm{Li}_s(x) \stackrel{s \to -\infty} \sim \Gamma(1-s) (-\log x)^{s-1}\, .
\end{equation}
This tells us that for $k \to \infty$
\begin{equation}
\begin{split}
E_k\, &=\, O \left( \frac 1 k\vert \log p \vert ^{-2-\ga k} \Gamma(\ga k +2) \left(\frac q{\vert z_0 \vert (1-\gep)} \right)^k
\right)
\\ 
&=\, O \left( \vert \log p \vert ^{-1-\ga k} \Gamma(\ga k +1) \left(\frac q{\vert z_0 \vert (1-\gep)} \right)^k
\right)\, .
\end{split}
\end{equation}
This is relevant because if we neglect the cosine modulation in $T_k$ we have
\begin{equation}
T^+_k\, :=\,
\frac{1}{k}
	\sum_{n=n_0}^\infty 
	p^n
	n^{\alpha k} \vert \eta_{n,1}\vert ^{-k}\, \ge \, - c^k+ \frac 1k \left( (1+ \gep) \vert z_0 \vert\right)^{-k} \mathrm{Li}_{-\ga k} (p)
\, ,
\end{equation}
where  the term $c^k$ takes care of the first $n_0$ terms in the sum of the polylogarithm. Since 
$\mathrm{Li}_{-\ga k} (p)
\sim   \Gamma (\ga k +1) \vert \log p \vert ^{-1-\ga k}$ for $k \to \infty$, by choosing $\gep$ adequately small we readily see  that $T^+_k$
 is much larger, in fact exponentially larger, than $E_k$: 
\begin{equation}
E_k\, =\, O \left( k q^k  (1+ 3\gep)^k T^+_k \right)\,.
\end{equation} 
\medskip 

Going on to estimating $T_k$ turns out to be somewhat technical, so we move some of the estimates to 
App.~\ref{sec:polylog}. We are going to see, as a  byproduct of App.~\ref{sec:polylog} that,  if we introduce  $\ell_k=\sqrt{k}\log k$ and 
$n_{p,k}:= \lfloor \ga k/ \vert \log p \vert\rfloor $, $T_k$ is asymptotically equivalent to  the truncated sum
\begin{equation}
\label{eq:forfnal}
-\frac{2}{k}\sum_{n:\, \vert n -n_{p,k} \vert \le \ell_k  } p^n n^{\alpha k}
	\vert \eta_{n,1}\vert ^{-k} \cos \left( k \arg (\eta_{n,1}) \right)\, .
\end{equation}	
This motivates the following lemma.

\medskip

\begin{lemma}
\label{th:2estimates} 
We write $n=n_{p, k} +j$. 
There exists real constants $a, b, c, d, A, B$ and $C$, whose explicit expressions are given in the proof, such that
for   with $ \vert j \vert \le \ell_k=\sqrt{k}\log k$ and for $k \to \infty$
\begin{equation}
\label{eq:2estimates-1} 
\cos \left( k \, \mathrm{arg} \left(  \eta_{n,1}
\right)\right)\, =\, 
\cos\left(  a k  +   \sqrt{k} + d  \frac j {\sqrt{k} }  +  c \right) 
+O\left( \frac {(\log k)^2} {\sqrt{k}}\right)
\, , 
\end{equation}
and 
\begin{equation}
\label{eq:2estimates-2} 
\vert \eta_{n,1} \vert ^{-k} \, =\, \vert z_0 \vert ^{-k} 
\exp \left ( A \sqrt{k} +C \frac{j}{\sqrt{k}} +B\right) \left( 1 +O\left( \frac{ \log k)^2}{\sqrt{k} }\right)\right) \, .
\end{equation}
\end{lemma}
\medskip 

\begin{proof}
By Taylor expansion we obtain for $n\to \infty$
\begin{multline}
\label{eq:prev52}
 \text{arg} \left(  \eta_{n,1}\right)\, =\, 
 \text{arg} \left( z_0\right) 
 + \frac 1{\sqrt{n}} \frac{\Im (z_1) \Re(z_0)-\Im(z_0) \Re(z_1)}{\vert z_0\vert^2}
 \\
+ \frac 1n \bigg(\frac{
-\Re(z_2) \Im(z_0)^3 +\Im(z_0)^2 \Im(z_1) \Re(z_1)+\Im(z_0)^2 \Im(z_2) \Re(z_0)-\Im(z_0) \Im(z_1)^2 \Re(z_0)}
{\vert z_0\vert^4} \ \ \ \
\\
 \ \ \ \ \ \ \ \ \ \ \ \ +\frac{
-\Im(z_0) \Re(z_0)^2 \Re(z_2)+\Im(z_0) \Re(z_0) \Re(z_1)^2-\Im(z_1) \Re(z_0)^2 \Re(z_1)+\Im(z_2) \Re(z_0)^3
}
{\vert z_0\vert^4}\bigg)
\\+ O \left( \frac 1 {n^{3/2}} \right)\, =:  \text{arg} \left( z_0\right) 
 + \frac 1{\sqrt{n}} b_1+ \frac 1n b_2 + O \left( \frac 1 {n^{3/2}} \right)\,.
\end{multline} 
Therefore with $n=n_{p, k} +j$, $\vert j \vert \le \ell_k$ and  for $k \to \infty$ we have
\begin{equation}
 \cos \left( k \, \text{arg} \left(  \eta_{n,1}\right)\right) \, =\, 
 \cos\left (k \,  \text{arg} \left( z_0\right) + b_1 \frac{k}{\sqrt{n}} +b_2 \frac k n +O\left(\frac 1{\sqrt{k}}\right) \right)\, , 
\end{equation}
where this equation defines $b_1$ and $b_2$ by comparison with \eqref{eq:prev52}.
If we set $c^2_p:= \vert \log p \vert/ \ga$ ($c_p>0$) we have
\begin{equation}
\frac k{\sqrt{n}}\, =\, c_p \sqrt{k} -  \frac {c_p^3}{2} \frac{j}{ \sqrt{k}}
+ O\left(\frac{(\log k )^2}{ \sqrt{k}}\right)\ \text{ and } \ \frac k n\, =\, c_p^2  + O \left( \frac {\log k}{\sqrt{k}} \right)\, ,
\end{equation}
 so
\begin{equation}
\begin{split}
\cos \left( k \, \textrm{arg} \left(  \eta_{n,1}
\right)\right)\, &=\, 
\cos\left(  \mathrm{arg}(z_0)k + b_1c_p  \sqrt{k} - \frac 12 b_1 c_p^3   \frac j {\sqrt{k} }  + b_2 c_p^2 \right) 
+O \left( \frac {(\log k)^2}{\sqrt{k}} \right)
\\
&=:\, 
\cos\left(  a k  + b  \sqrt{k} + d  \frac j {\sqrt{k} }  + c \right) 
+O \left( \frac {(\log k)^2}{\sqrt{k}} \right)
\, , 
\end{split}
\end{equation}
and the last line is the definition of the constants $a,b,c$ and $d$. This completes the verification of \eqref{eq:2estimates-1}. 

\smallskip

A similar Taylor expansion computation yields \eqref{eq:2estimates-2}. Here we give just the constants:
\begin{equation}
\label{eq:ABC}
A:=
-\Re \left( \frac {z_1} {z_0} \right)\sqrt{ \frac{\vert \log p \vert}{ \ga}}\, , \ 
 \ B:=- \left( \frac{z_2}{z_0} - \frac {z_1^2}{2 z_0^2} \right) \frac{\vert \log p \vert}{\ga}\, ,
\ 
C :=\,  \frac 12 \Re \left( \frac {z_1} {z_0} \right) \left( \frac{\vert \log p \vert}{\ga} \right)^{3/2}
\,.
\end{equation} 
\end{proof}

\medskip

It is now a matter of applying the results of Appendix~\ref{sec:polylog}, notably \eqref{eq:leading}
 and \eqref{eq:leading3},  to see that ($c_0$ is a positive constant)
\begin{multline}
-\frac{2}{k}\sum_{n\ge n_0} p^n n^{\alpha k}
	\vert \eta_{n,1}\vert ^{-k} \cos \left( k \arg (\eta_{n,1}) \right)\, =
	\\
-\frac{2}{k}\sum_{n:\, \vert n -n_{p,k} \vert \le \ell_k  } p^n n^{\alpha k}
	\vert \eta_{n,1}\vert ^{-k} \cos \left( k \arg (\eta_{n,1}) \right)
	\left (1+ O \left( e^{-c_0 (\log k)^2 }\right) \right)
	\\
	=
	\,  
-\frac{2e^{ A \sqrt{k}+B} }{k \vert z_0\vert^k} \times 
\\
\sum_{n:\, \vert n -n_{p,k} \vert \le \ell_k  } p^n n^{\alpha k}
e^{C \frac{n-n_{p,k}}{\sqrt{k}}} 
\left( \cos\left(  a k  + b  \sqrt{k} + d  \frac {(n -n_{p,k})} {\sqrt{k} }  + c \right) 
+O \left( \frac {(\log k)^2}{\sqrt{k}} \right)\right)
	\, .
\end{multline}	
We  now use
\begin{multline}
 \cos\left(   a k  +   b  \sqrt{k} -   \frac {   dj } {\sqrt{k} }  +   c \right) \, =\\ 
  \cos\left(   a k  +   b  \sqrt{k} +   c \right) \cos\left(   \frac {   dj } {\sqrt{k} }  \right)+
   \sin\left(   a k  +   b  \sqrt{k} +   c \right) \sin\left(   \frac {   dj } {\sqrt{k} }  \right)
   \, ,
\end{multline}
and we are  therefore left with estimating 
\begin{equation}
\label{eq:hes67}
-\frac{2e^{ A \sqrt{k}+B} }{k \vert z_0\vert^k} \cos\left(  a k  + b  \sqrt{k}  + c \right) 
\sum_{n:\, \vert n -n_{p,k} \vert \le \ell_k  } p^n n^{\alpha k}
e^{C \frac{n-n_{p,k}}{\sqrt{k}}} 
 \cos\left(  d  \frac {(n -n_{p,k})} {\sqrt{k} }  \right) \,,
\end{equation}
plus the analogous expression with $\cos(\cdot)$ replaced by $\sin(\cdot)$.
For this we apply  \eqref{eq:leading}, \eqref{eq:leading3} and \eqref{eq:forleading3}, with $\gb=  \ga k$ (we remark also the very mild effect  due to using  $\ell_k$ instead go  $\ell_{\ga k}$). The net result is that
the expression in \eqref{eq:hes67} is equal, up to a multiplicative error of $ 1 + O( \log k /\sqrt{k} )$, to
\begin{equation}
\label{eq:hes67-1}
-\frac{2e^{ A \sqrt{k}+B} }{k \vert z_0\vert^k} \cos\left(  a k  + b  \sqrt{k}  + c \right) 
\frac{\Gamma(1+ k \ga)} {\vert \log p\vert ^{1+ \ga k}} 
\exp \left( \frac{ \left(C^2-d^2\right) }{2 ( \log p )^2} \right) \cos\left( \frac{ Cd}{( \log p )^2} \right)
  \,,
\end{equation}
and for the sine case we obtain exactly the same expression, with the $\cos(\cdot)$ replaced by $\sin(\cdot)$ in the two occurrences. We therefore conclude that
\begin{multline}
\label{eq:almthr}
-\frac{2}{k}\sum_{n\ge n_0} p^n n^{\alpha k}
	\vert \eta_{n,1}\vert ^{-k} \cos \left( k \arg (\eta_{n,1}) \right)\, =
	\\
-\frac{2e^{ A \sqrt{k}+B} }{k \vert z_0\vert^k} \cos\left(  a k  + b  \sqrt{k}  + c+ \frac{ Cd}{( \log p )^2}  \right) 
\frac{\Gamma(1+ k \ga)} {\vert \log p\vert ^{1+ \ga k}} 
\exp \left( \frac{ \left(C^2-d^2\right) }{2 ( \log p )^2} \right) \left(1 + r_k \right)
  \,,	
\end{multline} 
where $r_k= O( \log k /\sqrt{k} )$.

\medskip

\medskip

\begin{lemma}
\label{th:uniform} 
For every $a, c \in \bbR$
 and for  $b\neq 0$ we have that 
\begin{equation}
\lim_{n\to \infty} \frac 1n \sum_{k=1}^n \gd_{(a k + b \sqrt{k} +c) \mathrm{mod}(2\pi)} \, =\, 
\gl_{2\pi}\, ,
\end{equation}
where $\gl_{2\pi}$ is the uniform probability on the circle $ \bbR/ (2 \pi \bbZ)$ and the convergence is the usual convergence in distribution. 
\end{lemma}

\smallskip

\begin{proof}
This is fully based on \cite[Chapter~2]{cf:KN}.
By rotation invariance,  
we can and do assume  that $c=0$. 
The case $a=0$  follows directly from \cite[Theorem~2.5]{cf:KN} while the case of $a/(2\pi)=p/q$ rational can be reduced to the case $a=0$ by separating the sums into $q$ terms. The case $a/(2\pi)$ irrational instead requires a different approach: this is treated by \cite[Theorem~3.3]{cf:KN}. We remark also that if $a/(2\pi)$ irrational the result holds also if $b=0$.  
\end{proof} 

\medskip

Lemma~\ref{th:uniform} can now be used in conjunction with 
\eqref{eq:ftilde_taylor3.1}, \eqref{eq:forfnal} and 
\eqref{eq:almthr}: it guarantees that the absolute value of the  oscillating term  $\cos\left(  a k  + b  \sqrt{k}  + c+ \frac{ Cd}{( \log p )^2}  \right)$ is bounded away from $0$ if $k$ stays out of a set of a density that can be made arbitrarily small, and in this case \eqref{eq:almthr} really gives the leading asymptotic behavior. One can then argue by contradiction to ensure that it suffices that $k$ stays out of a suitably chosen zero density set, and 
the leading asymptotic behavior is still given by \eqref{eq:almthr}. This completes the proof of Theorem~\ref{th:griffiths}.
\qed

\smallskip

\begin{rem}
\label{rem:constants}
Here is a guide to reconstruct the constants in Theorem~\ref{th:griffiths}. First of all $z_0= 1.22516\ldots + i\,  2.54713+ \ldots$ so
\begin{equation}
\mathtt{a}\,=\, \arg\left( z_0 \right)\,=\, 1.12247\ldots  
\end{equation}
and (recall that $z_1=z_0^2/2)$
\begin{equation}
\label{eq:mathttb}
\mathtt{b} \, =\, \sqrt{2 \vert \log p \vert }\,  b_1\ \text{ with } \ b_1=\frac{\Im (z_1) \Re(z_0)-\Im(z_0) \Re(z_1)}{\vert z_0\vert^2}= 1.27356\ldots,
\end{equation}
with the important fact that $\mathtt{b}\neq 0$ (cf. Lemma~\ref{th:uniform}).
The precise value of the other constants is less crucial: we have 
\begin{equation}
{\mathtt c}\, =\, 2 \vert \log p \vert b_2 - \  C b_1 \sqrt{\frac{2}{\vert \log p\vert}}\, ,
\end{equation} 
with $C$ in \eqref{eq:ABC} and $b_2$ in \eqref{eq:prev52} (like $b_1$, which however is also in \eqref{eq:mathttb}).
Moreover
\begin{equation}
C_1\, =\, - \frac {2} {\vert \log p \vert } \exp\left( B + \frac{\left(C^2-2 b_1^2 \vert \log p \vert^3\right)}{2 (\log (p))^2}
\right)\ \ \text{ and } \ \ C_2\,=\, \frac1{ \vert z_0 \vert \sqrt{ \vert \log p \vert }}\, , 
\end{equation}
and $B$ is also in \eqref{eq:ABC}, as well as $A$. 
\end{rem}

\appendix

\section{Probability estimates}
\label{sec:P-est}

\begin{proof}[Proof of Proposition~\ref{th:usingLLT}]
 We start with the leading asymptotic behavior of
 \begin{equation}
 Z_{n,\zeta/n^\ga}\, =\, \sum_{j=1}^n \exp\left(  j \zeta  /n^\ga\right) \bP\left( \tau_j=n \right)\, .
 \end{equation}
 and we are looking for a result that holds uniformly for  $\zeta$ is chosen in a compact set: we will just say ``uniformly in $\zeta$".
 For a positive (large) constant $L$
 we split the sum according to whether 
 $j\le n^\ga /L$, $j \in (n^\ga /L, L n^\ga )$ and $j \ge L n^\ga$. The intermediate segment, 
 $j \in (n^\ga /L, L n^\ga )$, can be treated by applying  Theorem~\ref{th:LLT}
 obtaining the asymptotic behavior claimed in \eqref{eq:usingLLT} with the integral spanning from $1/L$ to $L$, instead of from $0$ to $\infty$. It is therefore sufficient to show that the remaining two terms 
 are $\gep_L O(n^{1-\ga})$, with $\gep_L$ a positive constant that vanishes as $L\to \infty$.
 
 For the case $j\le n^\ga /L$ we are going to use that  for $n \to \infty$ and uniformly in $j$ such that $j/n^\ga \longrightarrow 0$ we have that $\bP(\tau_j=n) \sim j \bP(\tau_1=n)$ 
 \cite[Th.~A]{cf:Doney} so that for $n$ sufficiently large and uniformly in $\zeta$
 \begin{equation}
 \left \vert 
 \sum_{j\le n^\ga/L} \exp\left(  j \zeta  /n^\ga\right) \bP\left( \tau_j=n \right) \right\vert\, \le \, 2 
 \bP(\tau_1=n) \sum_{j\le n^\ga/L} j \, \le\, \frac{2 \mathtt{c}_K n^{1-\ga}}{L^2}\,, 
 \end{equation}

 For  $j\ge  L n^\ga $ we use instead  \cite[Lemma 4]{cf:Doney} that directly yields that for an appropriate choice of $C>0$, not depending on $L$, 
 we have that for $j\in [Ln ^\ga, n/L)$ 
 and $n$ sufficiently large
 \begin{equation}
 \bP(\tau_j=n) \, \le \, C \left( j/n^\ga \right)^{1/(2(1-\ga))} \exp \left( -\frac 1 C \left( j/n^\ga \right)^{1/(1-\ga)}\right)\, ,
 \end{equation}
 and that there exists $C_L>0$ such that 
 \begin{equation}
 \bP(\tau_j=n) \, \le \, \exp\left( - j/C_L\right)\, ,
 \end{equation}
 for $j \ge n/L$.
Therefore, with $b$ an upper bound for $\vert \Re( z) \vert$, and using again Riemann sum approximation we have that for $n$ sufficiently large 
\begin{multline}
 \left \vert 
 \sum_{j\ge  L n^\ga} \exp\left(  j \zeta  /n^\ga\right) \bP\left( \tau_j=n \right) \right\vert\, \le 
 \\
 \frac{2C}{ n^{1-\ga}} \int_L^\infty 
 y^{1/(2(1-\ga))} \exp \left( by - \frac 1 Cy^{1/(1-\ga)}\right) \dd y
 + \sum_{j\ge  n/L}  \exp\left(b\frac{j}{ n^\ga} - \frac j{C_L}\right)\, ,
  \end{multline} 
and we see that the first term in the right-hand side is $O(1/n^{1-\ga})$ times a term that can be made arbitrarily small by choosing $L$ large. The second term instead is $O(\exp(-n /(2C_L))$ and it is therefore much smaller. This completes the proof of \eqref{eq:usingLLT}.

For the  proof of proof of \eqref{eq:usingLLT'} we have to apply the very same arguments to
\begin{equation}
 Z'_{n,\zeta/n^\ga}\, =\, \sum_{j=1}^n \exp\left(  j \zeta  /n^\ga\right) j\, \bP\left( \tau_j=n \right)\, .
 \end{equation}
We skip the straightforward details.
\end{proof}

\section{Asymptotic behavior of modified polylogarithms}
\label{sec:polylog}
For  the standard polylogarithm we have 
\begin{equation}
\label{eq:leading}
\sum_n p^n n^\gb \stackrel{\gb \to \infty} \sim  T_\gb \, :=\, \frac{ \Gamma(1+\gb)}{ \vert\log p\vert^{1+\gb}}\,=\,
\frac{\exp(\gb \log \gb - \gb)}{ \vert\log p\vert^{1+\gb}} \left(\sqrt{2 \pi \gb} + O(1/ \gb) \right)
\, ,
\end{equation}
where $p \in (0,1)$. 
The first step in \eqref{eq:leading}   follows from \cite[(25.12.12)]{cf:DLMF} and the last step is Stirling formula with first order reminder. 

We now aim at recovering \eqref{eq:leading} by a direct saddle point analysis: the result will then be easily generalized to the case that interests us. For this
we start by introducing $\ell_\gb:= \sqrt{\gb } \log \gb$ and $n_\gb= \gb / \vert \log p \vert$.
We start by observing that the ratio 
\begin{equation}
\label{eq:ratioC1}
\sum_{n\notin [n_\gb- \ell_\gb,n_\gb+ \ell_\gb]}  p^n n^\gb \Big/ \left(\int_0^{n_\gb- \ell_\gb} p^x x^\gb \dd x
+ \int_{n_\gb+ \ell_\gb}^\infty 
p^x x^\gb \dd x
\right)\, ,
\end{equation}
is bounded away from $0$ and $\infty$ (in fact, it tends to one as $\gb \to \infty$, but \eqref{eq:ratioC1} is only used for tail bounds, which do not need to be sharp). 
With $x=n_\gb+y$ we have
\begin{equation}
\label{eq:expand1}
p^x x^\gb \,=\, \frac{\exp( \gb \log \gb - \gb)}{ \vert \log p \vert^{\gb}} \exp\left( -\frac {y^2 \vert \log p \vert ^2}{2 \gb} \right)
\left( 1 + O \left( \frac {y^3}{\gb^2} \right)\right)\,.
\end{equation}
The first application of this estimate is to show that the denominator, hence also the numerator, of \eqref{eq:ratioC1} is much smaller than $T_\gb$: more precisely,  for every $c>0$ it is  $O(T_\gb / \gb^{-c})$.
For the first integral in the denominator of \eqref{eq:ratioC1}  we use that the integrand is increasing in the interval of integration and \eqref{eq:expand1} with $y=-\ell_\gb$:
\begin{multline}
\label{eq:gid84}
\int_0^{n_\gb- \ell_\gb} p^x x^\gb \dd x \, \le \, 
p^{ n_\gb- \ell_\gb}  (n_\gb- \ell_\gb)^\gb 
\int_0^{n_\gb- \ell_\gb} \dd x \,
 \le 
 \\
 \frac{2 {\gb}   \exp( \gb \log \gb - \gb)}{ \vert \log p \vert^{1+\gb}} \exp\left( -\frac {(\log \gb )^2 \vert \log p \vert^2 }{2} \right)\, \le \, T_\gb \exp\left( -\frac {(\log \gb )^2 \vert \log p \vert^2 }{4} \right) \, .
\end{multline}
For the second integral
we use that the integrand is this time decreasing: since the interval of integration is unbounded we consider 
separately the integral from $n_\gb+ \ell_\gb$ to $\gb^2$ and from $\gb^2$ to $\infty$. 
We have 
\begin{equation}
 \int_{n_\gb+ \ell_\gb}^{\gb^2}
p^x x^\gb \dd x\, 
\le \, \gb ^2 p^{n_\gb+ \ell_\gb} (n_\gb+ \ell_\gb)^\gb 
\, \le\, 
\frac{\gb ^2 \exp( \gb \log \gb - \gb)}{ \vert \log p \vert^{\gb}} \exp\left( -\frac {(\log \gb )^2 \vert \log p \vert ^2}{3} \right)\,,
\end{equation} 
and precisely the final bound in \eqref{eq:gid84} is recovered. It is then straightforward to see that the integral 
from $\gb^2$ to $\infty$ vanishes as $\gb\to \infty$, yielding thus  a negligible contribution. 

We can then focus on 
\begin{equation}
\sum_{n\in [n_\gb- \ell_\gb,n_\gb+ \ell_\gb]}  p^n n^\gb \, =\, \int_{n_\gb- \ell_\gb}^{n_\gb- \ell_\gb} p^x x^\gb \dd x 
+E_\gb\, ,
\end{equation}
where $E_\gb$ can be bounded (first order Euler-Maclaurin formula) in terms of the value of the integrand at the two boundary points, this gives a contribution $O(T_\gb \exp(-c (\log \gb)^2))$ for some $c>0$ like in the previous estimates, plus the integral of the (absolute value) of the first derivative of the integrand.
Since $ \vert \partial_x  (p^x x^\gb) \vert = p^x x^\gb O(\ell_\gb / \gb)$ we readily find that $\vert E_\gb\vert=
O( T_\gb \log \gb / \sqrt{\gb})$. We can then work with the integral
and, by \eqref{eq:expand1}, we see that 
\begin{equation}
\begin{split}
 \int_{n_\gb- \ell_\gb}^{n_\gb+ \ell_\gb} p^x x^\gb \dd x \, &=\, 
 \frac{\exp( \gb \log \gb - \gb)}{ \vert \log p \vert^{\gb}}  \int_{- \ell_\gb}^{ \ell_\gb}\exp\left( -\frac {y^2 \vert \log p \vert ^2}{2 \gb} \right) \dd y
 \\ &= \frac{\exp( \gb \log \gb - \gb)}{ \vert \log p \vert^{\gb}}  \left(\int_{- \infty}^{ \infty}\exp\left( -\frac {y^2 \vert \log p \vert ^2}{2 \gb} \right) \dd y + o(1) \right)
 \\&=\,  \frac{\exp( \gb \log \gb - \gb)}{ \vert \log p \vert^{1+\gb}} \left( \sqrt{2 \pi \gb} +o(1) \right) \, ,
 \end{split}
 \end{equation}
 where $o(1)$ is actually $O\left(\exp\left( - {(\log \gb)^2 \vert \log p \vert ^2}/{3}\right)\right)$.
 Therefore we have recovered  and strengthened \eqref{eq:leading}: for $\gb \to \infty$
 \begin{equation}
\label{eq:leading2}
\sum_n p^n n^\gb \, =\, \sum_{n: \, \vert n -n_\gb\vert \le \ell_\gb}
 p^n n^\gb \left( 1+ O \left(\exp\left( -c (\log \gb)^2\right)\right)\right) \, =\,    T_\gb \left( 1+ O \left( \frac{\log \gb}{\sqrt{\gb}}\right)\right)\,,
\end{equation}   
with $c= (\log p)^2/4$. 
We have developed  in detail this procedure because of the control on the truncation error and  
because the steps generalize in a straightforward way to the following result: for 
$H(x):= \exp( Cx) h(x)$, 
$ C\in \bbR$ and $h$ a  bounded function with bounded first derivative, we have 
\begin{equation}
\label{eq:leading3}
\begin{split}
\sum_n p^n n^\gb H\left(  \frac{n-n_\gb}{\sqrt{\gb}}\right)  \, &=\,
\sum_{n: \, \vert n -n_\gb\vert \le \ell_\gb} p^n n^\gb H\left( \frac {n-n_\gb}{ \sqrt{\gb}}\right) \left( 1+ O \left(e^{ -c (\log \gb)^2}\right)\right) 
\\
 &=\, T_\gb  \left(\bE\left[ 
H \left( \frac Z{ \vert \log p \vert}\right)\right]
+ O\left(
    \frac{\log \gb}{\sqrt{\gb}}\right) \right)\,,
    \end{split}
\end{equation}   
where $Z$ is a standard Gaussian random variable.
The steps in the proof of \eqref{eq:leading3} are identical to those we performed for \eqref{eq:leading2} because the modulating function we have introduced  changes  the bounds only by constants (depending on $\vert  C\vert$, $\Vert h\Vert_\infty$ and $\Vert h'\Vert_\infty$). 

We need this result to $h(\cdot)=\cos(d\,\cdot)$ and to $\sin(d\, \cdot)$, with $d\in \bbR$: with these two special choices of $h(\cdot)$  we have
\begin{equation}
\label{eq:forleading3}
\bE\left[ 
\exp\left( C\frac Z{\vert \log p \vert}\right) h \left( \frac Z{ \vert \log p \vert}\right)\right]\, =\, 
\exp \left( \frac{ C^2-d^2}{2 ( \log p )^2} \right) h\left( \frac{ Cd}{( \log p )^2} \right)\,.
\end{equation}

\section{Monotonicity of the critical curve}
\label{sec:monotone}

Recall $f_1(\cdot)$ from \eqref{eq:f1} and $f_2(\cdot)$ from \eqref{eq:f2}. Recall moreover that $a=1-\ga \in (0,1)$.
One  directly checks that $f_1(\pi+\theta)= f_1(\pi-\theta)$ and 
$\tan(f_2(\pi+\theta))= -\tan(f_2(\pi-\theta))$ for $\theta \in [0, \pi]$. 
This allows to focus on $\theta\in [0, \pi]$ and, by continuity, it suffices to 
show that both $f'_1(\theta)$ and  $f'_2(\theta)$ are positive  for $\theta\in (0, \pi)$. 

\medskip

\begin{lemma}
\label{th:appC}
Both $f'_1(\theta)>0$ and $f'_2(\theta)>0$ hold for every $\theta\in (0, \pi)$ and every $\ga \in (0,1)$.
\end{lemma}

\medskip

\begin{proof}
We start by analyzing $f_2(\cdot)$.
If we differentiate the argument of the arctangent with respect to $\theta$ we find
\begin{equation}
\frac{2^{a-1} a \sin ^{a-1}\left(\frac{\theta}{2}\right) \left(2^a \sin ^{a+1}\left(\frac{\theta}{2}\right)+\cos \left(\frac{1}{2} (a \theta-\pi  a+\theta+\pi )\right)\right)}{\left(1-2^a \sin ^a\left(\frac{\theta}{2}\right) \cos \left(\frac{1}{2} a (\pi -\theta)\right)\right)^2}\, ,
\end{equation}
so the sign of this term is positive for $\theta \in (0, \pi)$ if and only if 
\begin{equation}
\label{eq:C.2}
2^a \sin ^{a+1}\left(\frac{\theta}{2}\right)+\cos \left(\frac{1}{2} (a \theta-\pi  a+\theta+\pi )\right) \, >0\, ,
\end{equation}
which is equivalent to
\begin{equation}
(2 \cos(\gp)) ^{1+a} \,>\, 2\cos \left( \gp (1+a) \right)\, , 
\end{equation}
for $\gp \in (0, \pi/2)$. Note now that it suffices to show this inequality for $\gp \in (0, \pi/(2(1+a))]$, because otherwise the right-hand side is negative. So it suffices to show, with $b =1+a \in (1,2)$ and $h(\cdot):= 2 \cos(\cdot)$, that for $\gp \in (0, \pi]$ 
\begin{equation}
(h(\gp/b)) ^{b} \,>\, h\left( \gp  \right)\ \ \Longleftrightarrow \ \ 
\frac{  \log h(\gp/b)}{\gp /b}\, >\, \frac{  \log h(\gp)}{\gp }\, ,
\end{equation}
and the inequality on the right holds because $\partial _\gp   \log h(\gp)= -(\gp  \tan (\gp)+\log (2 \cos (\gp))/{t^2}
$ is negative ($\gp  \tan (\gp)+\log (2 \cos (\gp)$ is equal to $\log 2 >0$ for $\gp=0$ and its derivative is $\gp/ \cos(\gp))^2>0$).
This completes the proof that $f_2^{'}(\theta)>0$  for $\theta \in (0, \pi)$ and $a \in (0,1)$ because $\arctan(\cdot)$ is increasing. 

\smallskip

\begin{rem}
The $\theta$-derivative of the square root of the denominator 
is 
\begin{equation}
-2^{a-1} a \sin \left(\frac{1}{2} (a+1) (\pi -\theta)\right) \sin ^{a-1}\left(\frac{\theta}{2}\right)\,, 
\end{equation}
which is negative  for $\theta\in [0, \pi)$ and it is zero at $\theta=\pi$. So $\theta=\pi$ is the minimum 
of the square denominator which takes value  $1$ in $\theta=0$ and value $1-2^a<0$ in $\theta=\pi$.
Hence the denominator hits zero only in one point $\theta_a \in (0, \pi)$. At this point the expression 
for $f_2(\theta)$ would have a jump of $-\pi$ had we chosen $\arctan(\cdot)$ instead of $\arctan_0(\cdot)$.  
\end{rem}

\smallskip

We are then left with showing that 
\begin{equation}
\frac{f_1'(\theta)}{f_2'(\theta)}\, =\, 
\frac{\sin \left(\frac{1}{2} (a \theta-\pi  a+\theta+\pi )\right)-2^a \cos \left(\frac{\theta}{2}\right) \sin ^a\left(\frac{\theta}{2}\right)}{2^a \sin ^{a+1}\left(\frac{\theta}{2}\right)+\cos \left(\frac{1}{2} (a \theta-\pi  a+\theta+\pi )\right)}\, ,
\end{equation}
for $\theta \in (0, \pi)$ and $a \in (0,1)$. By \eqref{eq:C.2} it suffices to show positivity of  the numerator and this is equivalent to showing
\begin{equation}
\sin ((1+a) \gp ) \, >\, 2^a \sin(\gp) \cos(\gp)
\ \ \Longleftrightarrow \ \ 2^b \sin ((2-b) \gp) \, > \, \sin (2 \gp)
\, ,
\end{equation}
for every $\gp \in (0, \pi/2)$ and $b= 1-a \in (0,1)$. If $(2-b) \gp\ge \pi/2$ the inequality holds because it holds even without the $2^b$ 
factor. So it suffices to focus on $\gp \in (0, \pi/(2(2-b)))$. We then make the change of variable $\psi=(2-b) \gp \in (0, \pi/2)$ and 
we boil down to the inequality
\begin{equation}
\sin(\psi) \, > \, 2^{-b} \sin \left( \frac{\psi} {1-(b/2)} \right)\,.
\end{equation} 
But $2^{-b} < 1-(b/2)$ for $b\in (0,1)$ so we are done if we can show the previous inequality with $2^{-b}$ replaced by $1-(b/2)$.
This amounts to showing that $\sin (\psi) > c \sin(\psi/c)$ for $c\in (1/2, 1)$ and $\psi \in (0, \pi/2)$: this last inequality holds even for $c\in (0,1)$ as one directly verifies. The proof of Lemma~\ref{th:appC} is therefore complete.
\end{proof}

\section{About numerics}

As pointed out in \cite{cf:KM}, in the restricted framework of  \eqref{eq:Kspecial} with
$\ga=1/2$ there is the explicit formula for 
$\bP(\tau_j=N)= (j/(2n -j)) 2^{-2n+j} C^n_{2n-j}$, with $C^k_n=n!/(k! (n-k)!)$. In general, one can
 obtain the coefficients
\begin{equation}
\cP_N\,:=\,(\bP(\tau_j=N))_{j=1,2, \ldots, N}\, ,
\end{equation} 
by an iterative procedure that
consists in building $\cP_{n+1}$ from $(\cP_{k})_{k=1, \ldots, n}$ via
\begin{equation}
\bP(\tau_{j+1}=n+1) \stackrel{j=1, \ldots, n}= \sum_{m=j}^n\bP(\tau_j=m)K(n+1-m)\,,
\end{equation}
and $\bP(\tau_{1}=n+1)=K(n+1)$,
starting from $\cP_1=( K(1))$. This way we can deal with arbitrary inter-arrival laws $K(\cdot)$ with $N$ up to $500$ 
with standard computers and moderate amount of time. This of course allows a large spectrum of numerical investigations even if the reachable $N$ are still rather small to really guess what the $N =\infty$ behavior could be (see for example Figure~\ref{fig:lac}).


\begin{figure}[h]
\centering
\includegraphics[width=13 cm]{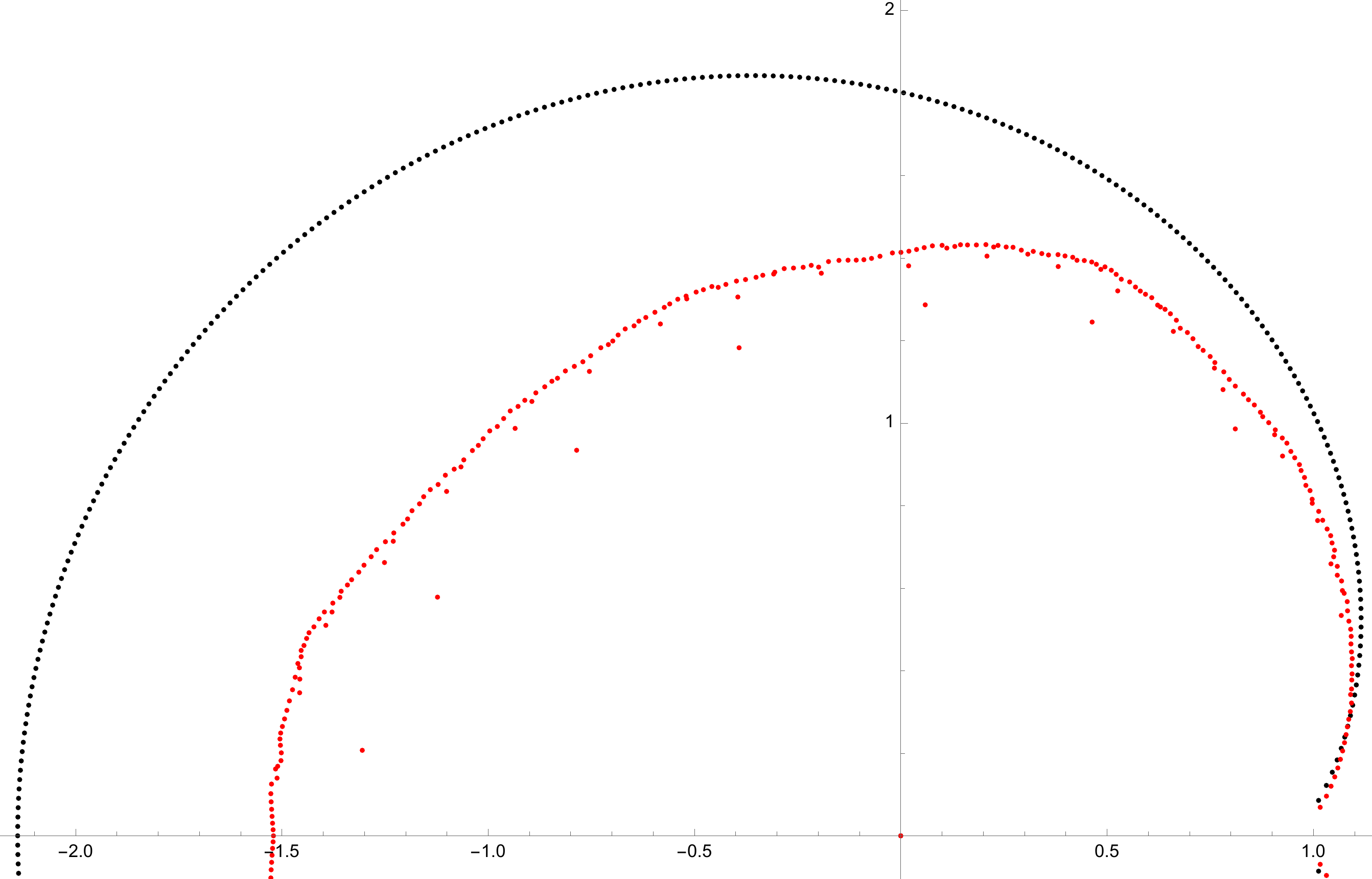}
\caption{\label{fig:lac} 
The black dots are the zeros of $P_{500}(w)= Z_{500, \Log w}$ with 
$K(n)=(K_{1}(n)+K_2(n))/2$ where $K_1(\cdot)$ is the inter-arrival law in \eqref{eq:Kspecial} with $\ga=1/2$ and 
$K_2(n)=1/(n^2 \zeta(2))$, where the normalization $\zeta(\cdot)$ is the Riemann $\zeta$ function. The red dots  instead are the zeros 
in the case in which $K_2(n)= 1/(n^2 \zeta(4))$ if $\sqrt{n}\in \bbN$ and $K_2(n)=0$ otherwise. We are therefore in the framework evacuated in Remark~\ref{rem:fabry}: the major effect of using a \emph{lacunary} distribution (albeit subleading!) is apparent, even if much larger values of $N$ would be needed to draw predictions from such a numerical observation (only $22$ entries of $K_2(\cdot)$ are non zero).
}
\end{figure}

\section*{Acknowledgements}

\begin{sloppypar}
	G.G. is very grateful to Bernard Derrida for several exchanges on the content of this work. We thank Romain Dujardin for pointing out \cite{cf:KN}, crucial for Lemma~\ref{th:uniform}.
	G.G.\ is partially supported by ANR–19–CE40–0023 (PERISTOCH).
	The work of R.L.G.\ was supported by the European Research Council (ERC) under the European Union's Horizon 2020 research and innovation programme (ERC CoG UniCoSM, grant agreement No.\ 724939 and ERC StG MaMBoQ, grant agreement No.\ 802901). 
\end{sloppypar}

\section*{Competing Interests}

The Authors declare they have no competing interests.

\end{document}